%% file: main_arxiv.tex
\pdfoutput=1
\documentclass[11pt]{article}
\usepackage[margin=1in]{geometry}
\usepackage{booktabs}
\usepackage{array}
\usepackage{multirow}
\usepackage{multicol}

\usepackage{graphicx} % Required for inserting images

\usepackage{multirow,amssymb,amsmath,tabularx, amsthm, amsfonts,color}
\usepackage{import}
\usepackage{placeins}
\usepackage[ruled,linesnumbered]{algorithm2e}
\usepackage{chngcntr, apptools}
\usepackage{mathtools}
\usepackage{tikz}
\usepackage{pgfplots}

\definecolor{col1}{rgb}{0.192,0.510,0.741}
\definecolor{col2}{HTML}{E6550D}
\definecolor{col3}{rgb}{0.2470588, 0.5019608, 0.3803922}
\definecolor{gold}{RGB}{255,215,0}

\let\oldnl\nl% Store \nl in \oldnl
\newcommand{\nonl}{\renewcommand{\nl}{\let\nl\oldnl}}% Remove line number for one line

\newlength\mylen

\newtheorem{proposition}{Proposition}

\usepackage{natbib}
\usepackage[hidelinks]{hyperref}
\usepackage{url}
\usepackage{cleveref}

\title{A Convolution Process for Sea
Surface Temperature Hot-Spot Identification in the Mediterranean Sea}
\date{}
\author{Leonardo Marchesin\textsuperscript{1}, Alessandra Menafoglio\textsuperscript{1} and Piercesare Secchi\textsuperscript{1}
}

\begin{document}

\maketitle
\let\thefootnote\relax
\footnotetext{\hspace{-0.57cm}\textsuperscript{1}MOX - Department of Mathematics, Politecnico di Milano, Milan, Italy.}
%\title{Title} %%%%%%%%%%%%
%\author[Initial Surname]{Author}
%\address{Address}

\begin{abstract}
\input{abstract.tex}\\

\noindent \textbf{Keywords:} Network-based spatial processes; Flow-informed covariance modeling; Moving-average processes; Sea surface temperature; Extreme events; Hot-spot identification.
\end{abstract} %%%%%%%%%

\bigskip

\input{body.tex}

\section*{Acknowledgments}
L. Marchesin gratefully acknowledges the financial support of his PhD fellowship from Polis-Lombardia. P. Secchi acknowledges the PRIN2022 project CoEnv - Complex Environmental Data and modelling (CUP2022E3RY23) founded by the European Union - NextGenerationEU program and by the Italian Ministry for University and Research (MUR).  All authors acknowledge the support of MUR, grant Dipartimento di Eccellenza 2023–2027. The simulations discussed in this work were, in part, performed on the HPC Cluster of the Department of Mathematics of Politecnico di Milano which was funded by MUR grant Dipartimento di Eccellenza 2023-2027.
\bibliographystyle{Chicago}

\bibliography{bibliography.bib}

\newpage

\appendix
\input{appendix}

\end{document}

%% file: abstract.tex
Sea surface temperature (SST) is a fundamental determinant of global climate dynamics and economic activity. Reliable projections of future SST patterns depend critically on a rigorous characterization of the underlying spatial random field. In this study, we introduce a novel convolution-based covariance framework tailored to geostatistical domains constrained by physical barriers and influenced by vector-driven flows. By discretizing the continuous marine domain into a directed linear network that preserves the orientation of ocean currents, we construct a moving-average stochastic process whose dynamic is encoded via a Markovian transition‐probability matrix on the network’s vertices. The induced covariance structure emerges as a weighted combination of a spatial kernel and flow‐dependent weights, giving rise to a complex estimation problem. To stabilize inference, we propose a penalized estimator that regularizes covariance parameters while enforcing consistency with known hydrodynamic properties. We then embed this covariance model into a Monte Carlo simulation framework to refine RCP-based SST projections and to identify thermal “hot spots” of heightened ecological risk. Our approach delivers a statistically principled framework that prevents physical inconsistencies -- such as correlations across land barriers -- providing a robust basis for quantifying uncertainty in future SST forecasts and for guiding targeted environmental assessments.

%% file: body.tex
\section{Introduction}
\label{sec:intro}

Sea surface temperature (SST) is a vital component in sustaining life and global prosperity. Extreme values of SST can lead to the loss of part of the biodiversity in the sea \citep{heatwaves2, heatwaves3}. In particular, the Mediterranean sea is a hot-spot for climate change \citep{hotspot}: despite its limited extent it hosts a large part of the World's marine wildlife and plants. Accurately assessing the risks associated with rising water temperatures requires moving beyond point predictions, which fail to capture the full spectrum of variability necessary for such a complex and sensitive issue. Consequently, robust statistical methods are essential for evaluating the risk of exceeding critical thresholds and conducting insightful analyses \citep{extremeeventsHuser, ExtremeeventsBolin, ExtremeeventsFrench}.

Future SST projections are commonly studied using Representative Concentration Pathway (RCP) scenarios, which simulate future greenhouse gas concentrations and their impacts on global systems. These scenarios are implemented within complex climate models using ensemble methods that aim to capture the intricate interactions between greenhouse gas concentrations and Earth's features \citep{CMIP5}. However, such projections are limited to spatial point estimates of the mean SST, restricting their ability to provide comprehensive statistical insights into potential risks and uncertainties. This work introduces a new framework for spatial statistics which allows for physics-based predictive distribution of future SST, enabling assessment and uncertainty quantification of the risks associated with rising temperatures.

Spatial statistics concerns the analysis of spatially indexed data and the dependence relationships arising therein. Traditionally, dependence among observations has been characterized via covariance functions, with most methodologies formulated under the assumption of a Euclidean distance metric \citep{Cressie_1993}. However, when one substitutes a non‐Euclidean distance into these classical parametric covariance models, the resulting covariance function may fail to satisfy the requirements of positive definiteness, producing relations that are invalid for statistical inference \citep{Curriero_2006}.

There are many scenarios where an Euclidean framework cannot accurately represent the domain. Spatial connectivity, driven by underlying physical phenomena, may create stronger relationships between particular pairs of locations while weakening others. Additionally, discontinuities (for instance, a gap within the spatial domain) may render some connections infeasible. In such cases, the traditional Euclidean metric proves inadequate for capturing the complexities of the domain. Therefore, alternative metrics need to be employed to more effectively and comprehensively describe the relationships between locations. Water resources, in particular, represent a typical scenario where the challenges mentioned above emerge.

\Cref{img:whole} illustrates the extent of the area studied in this work, where a velocity field shapes the structure of a spatial domain. The figure displays the point estimates projections of the sea‐surface temperatures in the northern Thyrrenian Sea in August 2050, with arrows representing the prevailing marine currents. These currents impose directional connectivity across the domain, and any spatial modeling framework that aims to accurately capture relationships between locations must incorporate these flow-driven constraints. Accounting for the physical structure induced by the velocity field is essential for a more faithful representation of spatial dependence.

\begin{figure}[t]
    \centering
    \begin{minipage}{0.8\linewidth}
        \hspace{3 cm}
        \includegraphics[width=0.55\linewidth]{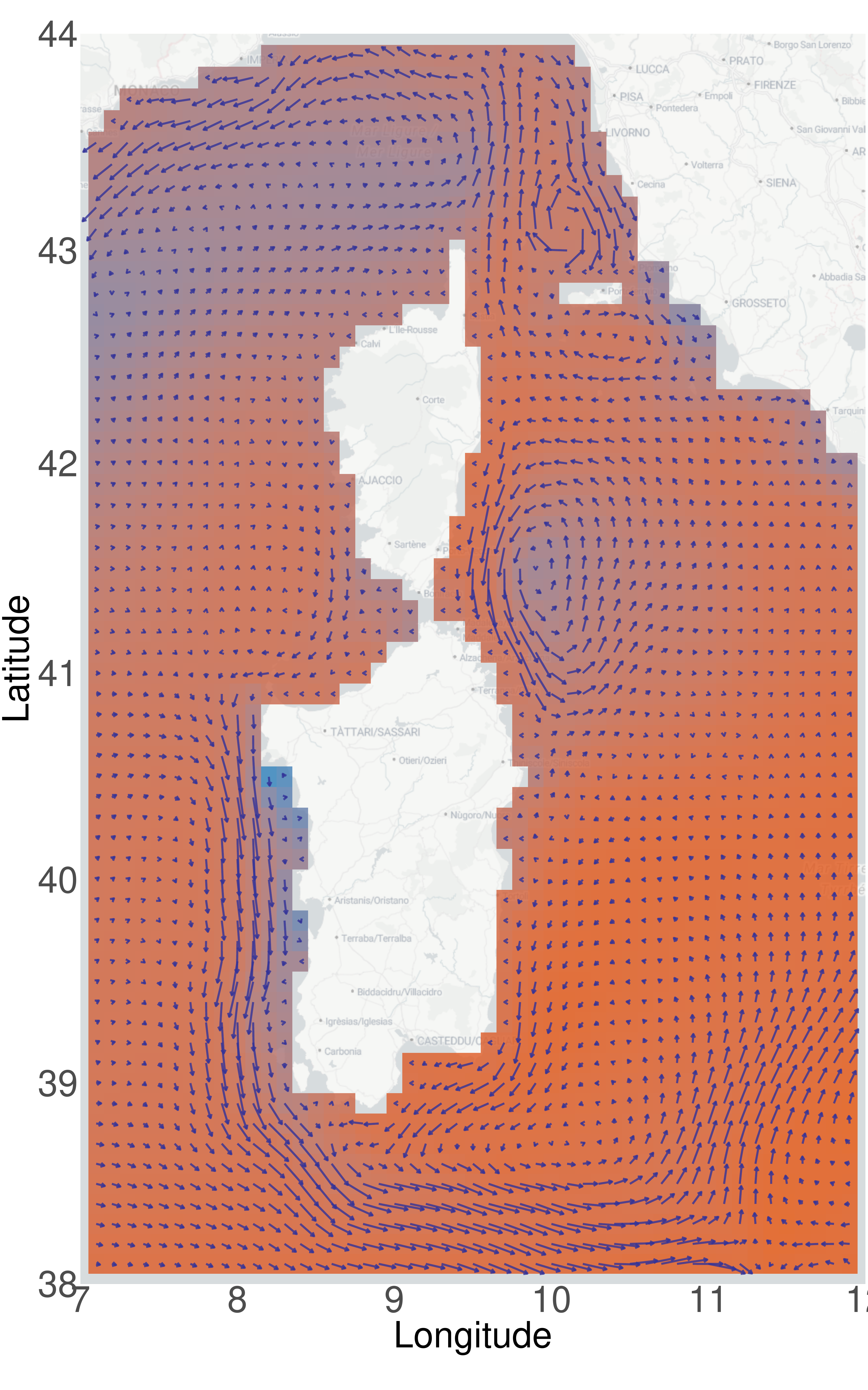}
    \end{minipage}%
    \begin{minipage}{0.15\linewidth}
        \hspace{-4 cm}
        \centering
        \begin{tikzpicture}
            \begin{axis}[
                hide axis,
                scale only axis,
                height=0pt,
                width=0pt,
                ymin=0, ymax=1, 
                xmin=0, xmax=1, 
                colorbar,
                colormap={col1-to-col2}{[1pt] color(0pt)=(col1); color(100pt)=(col2)},
                point meta min=24,
                point meta max=28,
                colorbar style={
                    title=Temperature $C^{\circ}$,
                    title style={font=\small},
                    xtick={24, 26, 28},
                    xticklabels={24,26,28},
                    width=0.3cm, 
                    height=6cm,
                    xticklabel style={font=\small}
                }
            ]
            \end{axis}
        \end{tikzpicture}
    \end{minipage}
    
    \caption{The spatial domain of the northern Tyrrhenian Sea is depicted with the point estimates projections of mean sea-surface temperatures in August 2050, provided by the RCP 4.5 emissions pathway. The arrows represent the velocity fields. Sardinia, Corsica, and the Italian peninsula coastline are clearly recognizable.}
    \label{img:whole}
\end{figure}

This work introduces a novel framework for developing spatial statistical models that explicitly incorporate an underlying velocity field governing the domain. The objective is to capture and represent domain connectivity through a covariance structure informed by the flow dynamics, thereby enabling spatial statistical analysis, uncertainty quantification, and simulation studies that reflect the physical interactions implied by geographic configuration.

The first step in the analysis involves defining a suitable representation of the spatial domain that accounts for the directionality and connectivity induced by the current field. To this end, the linear network representation is adopted, as it provides a natural and effective structure for modeling connected domains. This representation captures the directional and relational characteristics of the system, with the flow dynamics encoded directly through the network’s edges, preserving the physical coherence of the spatial dependencies.

Recent advancements in spatial statistics have rigorously formalized the theory of random fields on metric graphs \citep{anderes2020, bolin2024}. These foundational works provide valid covariance models based on the geodesic metric, effectively overcoming the limitations of Euclidean distances in non-convex domains. However, standard metric graph models are typically symmetrical, assuming that dependence relies solely on the distance along the graph, regardless of direction. While this assumption holds for many diffusive processes, it does not account for the strong advection characteristic of sea currents.

To model spatial dependence over complex, irregular and physics-driven domains, one widely used framework is based on stochastic partial differential equations (SPDEs) \citep{spdeLindgren, spdeLindgren-2, clarotto}. These methods construct Gaussian random fields with Matérn-type covariance structures by solving SPDEs, providing a principled bridge between continuous-domain models and sparse representations on triangulated meshes. This framework has proven effective in many applications, particularly due to its computational efficiency and scalability. However, classical SPDE approaches typically rely on symmetric differential operators, which naturally yield spatially symmetric covariance structures. This assumption presents challenges when modeling processes on directed networks, where edge asymmetry breaks operator self-adjointness. Moreover, extending SPDEs to accommodate nonstationarity, anisotropy, or complex boundary behavior often requires careful design of the differential operator and mesh structure. While there are emerging efforts to relax these assumptions, a general and practical framework for covariance modeling on directed networks remains underdeveloped.

Another closely related methodology is spatial smoothing with differential regularization proposed by \cite{clemente2026, Tomasetto_Arnone_Sangalli_2024}. This method applies spatial smoothing to the available data, considering a differential penalization that incorporates possible prior knowledge about the phenomenon, expressed through a partial differential equation \citep{ramsay2002, Azzimonti_Sangalli_Secchi_Domanin_Nobile_2015}. However, unlike geostatistical methods, spatial regularization does not provide a spatial covariance function, but rather captures the spatial variability through a deterministic function, estimated through smoothing. Having an (estimated) covariance model at our disposal is crucial not only for enabling spatial prediction but also for simulating multiple realistic scenarios—an essential requirement in studies focused on uncertainty quantification and scenario analysis. This ability to generate diverse and plausible scenarios captures the inherent variability in the data and gives geostatistical methods a significant advantage over purely smoothing-based approaches.

In this work, we propose a novel class of valid covariance models for directed linear networks, based on a new construction of a moving-average-type convolution process. Our approach is inspired by the moving average framework developed for stream network systems \citep{VerHoef2006, peterson2007, cressie2006}, where the process value at a given vertex is defined as a local average of random noise weighted by a spatial kernel. In the proposed framework the flow dynamics governing the moving average construction are modeled stochastically, through a Markov chain defined on the network. This stochastic encoding of the transport mechanism allows the model to naturally propagate the information about the velocity field into the covariance structure.

The remainder of this work is organized as follows. \Cref{sec:physics-driven} introduces the northern Tyrrhenian Sea domain of analysis and its approximation as a linear network. \Cref{sec:convprocess} defines the proposed convolution-based process and presents the functional form of the resulting covariance model. \Cref{sec:estimation} outlines the covariance estimation procedure. \Cref{sec:summarysimulation} summarizes the results of the simulation study. \Cref{sec:casestudy} applies the methodology to the spatial field defined by the sea surface temperature of the northern Tyrrhenian Sea, estimates and analyzes its spatial dependence structure and uses it to simulate multiple realizations of the sea surface temperature, leveraged for conducting an extreme event analysis. Proofs of all theoretical results, along with a simulation study and detailed algorithmic specifications, are provided in the Appendix.

\section{Representing a physics-driven spatial domain}
\label{sec:physics-driven}

\subsection{SST and current data}
\label{subsec:data}

The Copernicus Climate Change Service \citep[C3S,][]{copernicusprojections} provides Representative Concentration Pathway (RCP) water temperature projections up to 2099, evaluated using the European Regional Seas Ecosystem Model \citep[ERSEM, ][]{ERSEM}. These projections offer high-resolution data on various environmental variables, including sea surface temperature (SST) and sea currents, across the entire northern hemisphere.

Representative Concentration Pathways (RCPs)  are named according to the range of radiative forcing values they are expected to achieve by the year 2100 \citep{rcp}. In particular, we consider the RCP 4.5. It is described as an intermediate scenario. It is more optimistic than the current situation, as it assumes a moderate reduction of emissions.

For our analysis, we employ monthly mean SST for August, which represents the period of maximum thermal stress on marine ecosystems, both for the sea currents and the SST. We focus on the northern Mediterranean basin (specifically, the northern Tyrrhenian Sea), as this region may provide thermal refuge for marine organisms when other areas experience extreme warming during peak summer conditions. Understanding projected temperature changes in this region is essential for assessing habitat persistence under future climate scenarios \citep{watertemperature}.

\Cref{img:whole} illustrates the spatial extent of the study area. It presents projected August sea temperature and velocity fields for the mid-twenty‐first century (2050), serving as representative snapshots of future warming under the RCP 4.5 emissions pathway. The available temperature projections are point estimates, which, while informative, are insufficient on their own for deriving statistically robust conclusions about critical regions. To enhance inferential reliability, a detailed analysis of the associated spatial dependence -- and therefore, uncertainty -- is essential.

Specifically, let \(\{Y_s, \hspace{.1 cm} s \in D\}\) be a Gaussian random field representing sea surface temperature (SST) over the analysis domain $D$ \citep{SSTGaussian}. We observe this field at locations \(\{s_1, \dots, s_n\}\), yielding the partial realization \(\{Y_{s_1}, \dots, Y_{s_n}\}\). The field’s mean surface is prescribed by the Copernicus SST projections. The primary objective of this work is to characterize the covariance structure of SST. To this end, we compute yearly empirical residuals by comparing Copernicus projections with observed SST data. Indeed, the projection period begins in 2006 and overlaps with available satellite data that can be found in the Copernicus Marine Environment Monitoring Service \citep[CMEMS][]{copernicus}; the observations are provided until 2022. For this study, the selected CMEMS data set describes temperature and currents field over the Mediterranean Sea, and it is available on a regular grid. These data can be used to compute empirical residuals between RCPs and observations, which are then integrated with observed current-velocity measurements to inform the estimation of a physics-informed spatial covariance model.

Since the data are provided on a regular grid, it is natural to represent $D$ as a linear network, where each grid node is a vertex and the current field induces a directed edge structure between adjacent neighbors. This construction requires no boundary conditions, which are instead needed by PDE-based approaches and are particularly difficult to specify reliably over the irregular coastal geometries of the Mediterranean basin.

\subsection{Linear network domain}
\label{subsec:lin_net}

A directed linear network in $\mathbb{R}^2$ is defined as a pair $(\mathcal{L}, V)$ where $\mathcal{L} = \cup_i l_i$ is the finite set of all directed edges that make up the network and $V\subset\mathbb{R}^2$ is the finite set of vertices. Each edge is defined as $l = l_{[a, b]} = \{ u \in \mathbb{R}^2 : u = ta + (1-t)b;\, 0 \leq t \leq 1 ;\, a,b \in V \}$. The intersection of any two edges is either empty or a vertex. Let $\mathcal{N} = V \cup  \mathcal{L}$ to be the whole domain of points in $\mathbb{R}^2$ obtained as the union of all vertices and points on the edges of the linear network $(\mathcal{L}, V)$. Hence, any $u \in \mathcal{N}$ either lies on some unique edge $l_{\left[a,b\right]}$, or is a vertex.

In fact, a linear network may more generally be defined in any metric space, with vertices embedded in that space and edges representing connections between them. For what follows, a fundamental attribute of each edge is its length as measured by the chosen metric. The specific metric is problem-dependent -- any appropriate distance function may be employed, and multiple metrics can even coexist within the same ambient space. In this work, we adopt the edge length -- $length(l_{\left[a,b\right]}) = \|a-b\|_2$, $\|\cdot\|_2$ denoting the Euclidean distance -- as the distance metric between the vertices of an edge.

A path $ p=p_{(v,x)}$ between two vertices $v$ and $x$ in $V$ is an ordered finite sequence of connected edges in $\mathcal{L}$ such that $ p_{(v,x)}=\{l_i=l_{ \left[a_i, b_i\right]} \in \mathcal{L},\, i = 1,...,n: a_1 = v; \, b_n = x; b_i = a_{i+1}\}$. The length of the path is the sum of the lengths of all edges of the path, i.e.,  $length(p_{\left(v,x\right)}) = |p_{\left(v,x\right)}| = \sum_{i = 1}^{n} \|a_i - b_i\|_2$. In the following, we denote by $\mathcal{P}(v,x)$ the set of all possible paths from vertex \(v\) to \(x\). Note that in a directed network the set of paths in different directions can be different -- i.e., the paths from $v$ to $x$ might be different than those from $x$ to $v$. If there are no paths connecting \(v\) and \(x\), then $\mathcal{P}(v,x)$ is the empty set. Finally, observe that cyclic paths may occur. For instance, the set $\mathcal{P}(v,x)$ contains as distinct elements both the path $p_{(v,x)}$ and the path $p_{(v,x)} \cup p_{(x,x)}$.

More in general, the natural definition of the sub-path $p_{\left(u,x\right)}$ between a point in the network $u \in \mathcal{N}$ and a vertex $x \in V$ goes as follows. Let $u \in l_{[v,b_1]}$. If $b_1=x$, then $p_{(u,x)} = [u,x]$ is the sub-edge and its length is $|p_{(u,x)}|=lenght([u,x])=\|u-x\|_2$. If $b_1 \not = x,\) let $p_{\left(v,x\right)} = \left\{l_{[v,b_1]}, l_2 \dots, l_n\right\}$ and set $p_{\left(u,x\right)} = \left\{\left[u,b_1\right], l_2, \dots, l_n\right\}$. Denoting $\left\{l_2, \dots, l_n\right\}$ as $p_{(b_1,x)}$, we say that $p_{\left(b_1,x\right)}\subseteq p_{\left(u,x\right)} \subseteq p_{\left(v,x\right)}$. Finally, we define the length of the sub-path $p_{\left(u,x\right)}$ as $|p_{\left(u,x\right)}| = length(\left[u,b_1\right]) + |p_{(b_1,x)}|$. Note that, although individual edge lengths coincide with Euclidean distances, the distance between a point \(u \in \cal N\) and a vertex \(x \in V\) constrained to the network topology may differ substantially from the direct Euclidean distance of the two points in \(\mathbb{R}^2.\)

To construct a linear network $(\mathcal{L}, V)$ approximating the water domain \(D\) represented in \Cref{img:whole}, we outline a procedure for defining its vertices and edges, which leverages the regularity of the grid \({\cal G}\) on which both current and temperature data are collected. The vertices $V$ correspond to the nodes $\{s_1,...,s_n\}$ of \({\cal G}\) within the water domain \(D\). Note that only those nodes for which we observe a temperature value are included, thus explicitly excluding non-water regions of the domain. The edges in $\cal L$ are designed to uphold the directionality dictated by the velocity field, ensuring that the network faithfully represents the natural connectivity of locations linked by water flow. Specifically, on a regular grid each node has eight surrounding neighbors, as illustrated in \Cref{img:network_construction}a. An edge connects a vertex $a \in V$ to one of its eight adjacent nodes on the grid. Drawing inspiration from river topography \citep{Dinf, hydroProximityMeasure}, the non-null velocity vector $\textbf{v}_a$ observed at the vertex $a$ naturally suggests two likely flow directions -- those most closely aligned with its orientation. The two directed edges are defined as the vectors along these two directions, whose sum -- according to the parallelogram rule -- equals $\textbf{v}_a$. As illustrated in \Cref{img:network_construction}b, these connect $a$ to the two most adjacent nodes, $e$ and $f$ in the figure. We include the edge \(l_{[a,e]}\) in \({\cal L}\) as an edge of the network if and only if \(e \in V,\) and the same goes for \(l_{[a,f]}\). If \(e\) or \(f\) do not belong to \(V\), the vertex \(a\) is called an \textit{outlet}. If the vertex \(a\) has no incoming edges, it is called a \textit{source}. Note that a vertex can simultaneously be a source and a outlet. Repeating this procedure for all vertices in \(V,\) a complete directed linear network $(\cal{L}, V)$ is constructed, as represented in \Cref{img:network} where we use different colors to identify sources and outlets. The network corresponds to current data in \Cref{img:whole}.

\begin{figure}
    \centering
    \begin{minipage}{0.48\linewidth}
        \centering
        \includegraphics[width=0.65\linewidth]{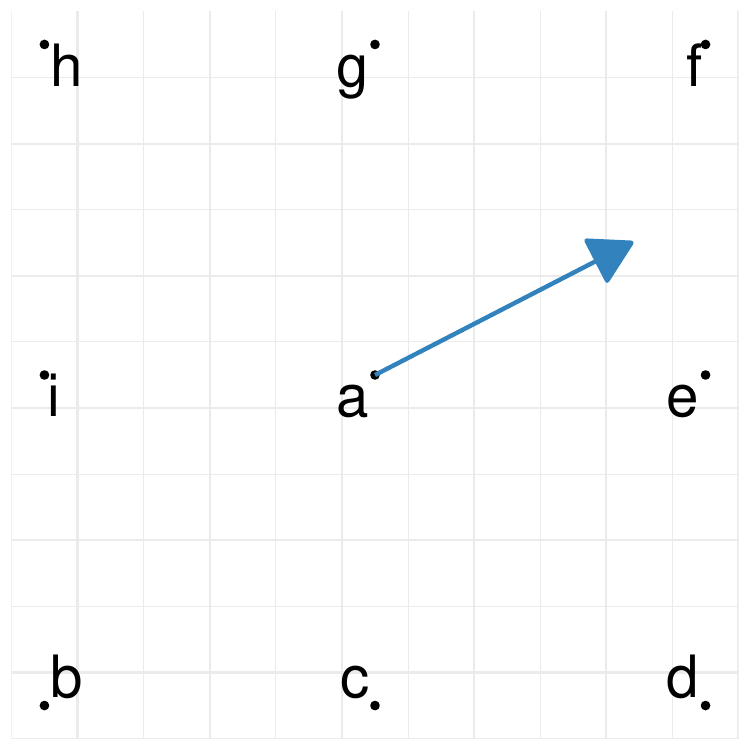}
        \vspace{0.1cm} % Small gap before the label
        \\ (a) 8 locations around the point of the grid.
    \end{minipage}
    \hfill
    \begin{minipage}{0.48\linewidth}
        \centering
        \includegraphics[width=0.65\linewidth]{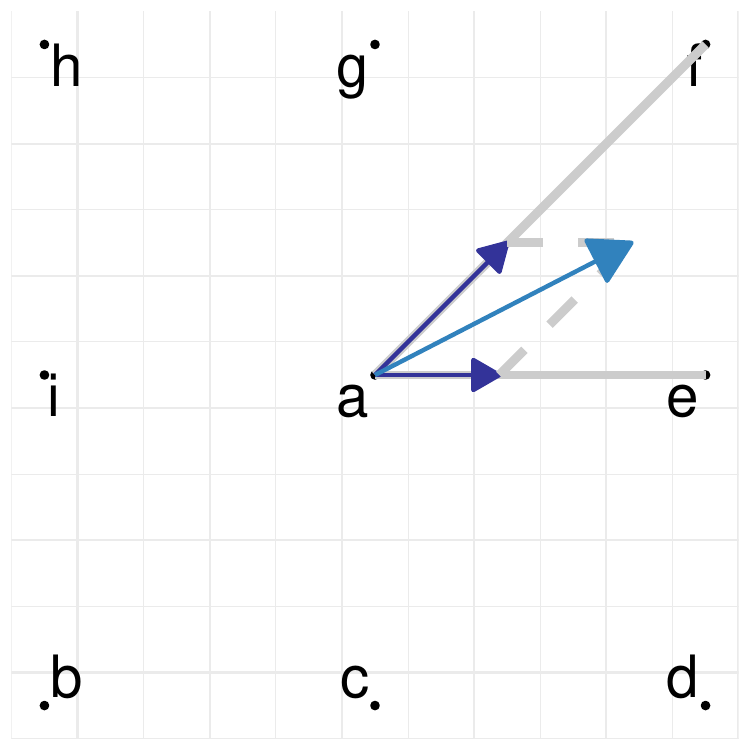}
        \vspace{0.1cm}
        \\ (b) Definition of the edges.
    \end{minipage}
    
    \caption{Construction of the edges.}
    \label{img:network_construction}
\end{figure}

In this work, velocity and temperature data are co-located on the same regular rectangular grid. We exploit this co-location in constructing the network representation. More generally, when observations are not co-located or are irregularly spaced, one could define a sufficiently fine grid such that the observations lie on it. Extension to such non-co-located settings will be explored in future work.

The connectivity patterns induced by the sea currents drive the network construction. We focus on the direction of the velocity field in each point, hence overlooking the magnitude of the velocity of these vectors. If this additional information were to be incorporated, we could account for both the Euclidean distance between vertices and the velocity of the water in that direction by setting the edge length as proportional to the time required for the water flow to travel the edge -- namely, proportional to the ratio between the Euclidean distance of the edge and the magnitude of the corresponding  water velocity vector. Note that this change would not affect the following construction.

We do not impose restrictive boundary conditions on the perimeter vertices; in particular, we allow water to both exit and enter the domain. This choice reflects the physical structure of the system: the study area, as illustrated in \Cref{img:whole}, is inside of a bigger system, comprised of the surrounding seas and the open ocean. Consequently, the domain \(D\) should be regarded as an open system, with water entering from external sources and necessarily flowing outward beyond its boundaries without any specific constraint. Thus the network has some specific points, the  \textit{sources}, from which water come in from outside, and some points from which the water flows outside \(D\), the \textit{outlets}.

It should be noted that, although the topology of the linear network has a clear influence on the subsequent covariance structure, the mathematical construction of the following sections does not rely on the specific procedure adopted here to encode through a graph the physics-driven knowledge of the case study under consideration. Indeed, any alternative procedure capable of producing a meaningful linear directed graph could be employed.

\begin{figure}[t]
    \centering
    \includegraphics[width=.45\linewidth]{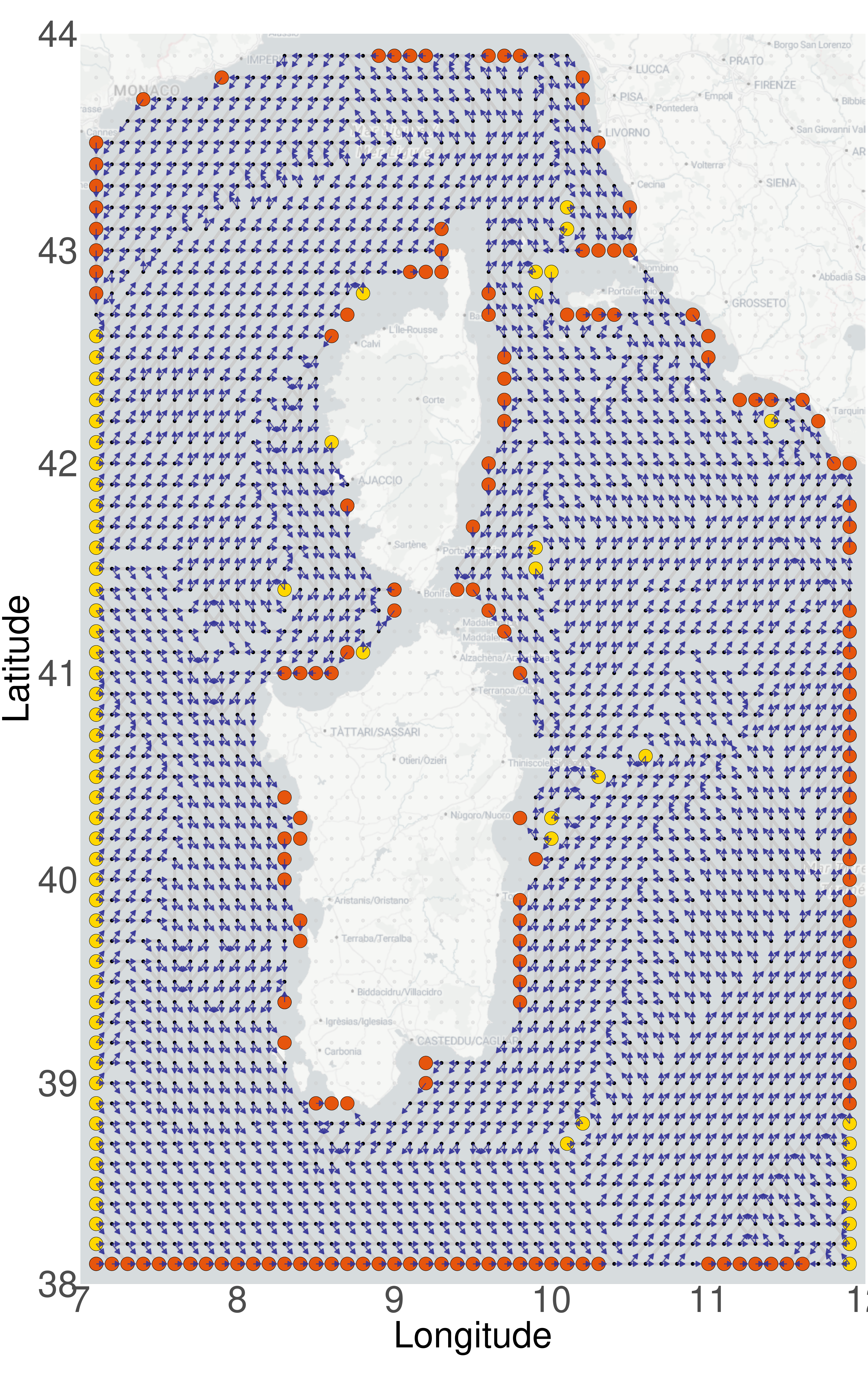}

    \vspace{0.3cm}

\begin{tikzpicture}[y=0.5cm, x=0.5cm]
\fill[gold] (1.8,0) circle [radius=0.5];
\node[below] at (2,-0.5) {Sources};

\fill[col2] (6,0) circle [radius=0.5];
\node[below] at (6,-0.5) {Outlets};

\end{tikzpicture}

    \caption{The linear network in the whole domain of analysis. In blue the directed edges composing the network are represented. Grey points unconnected with the others are non-water locations, explicitly excluded from the domain of analysis. The sources and the outlet points are highlighted.}
    \label{img:network}
\end{figure}
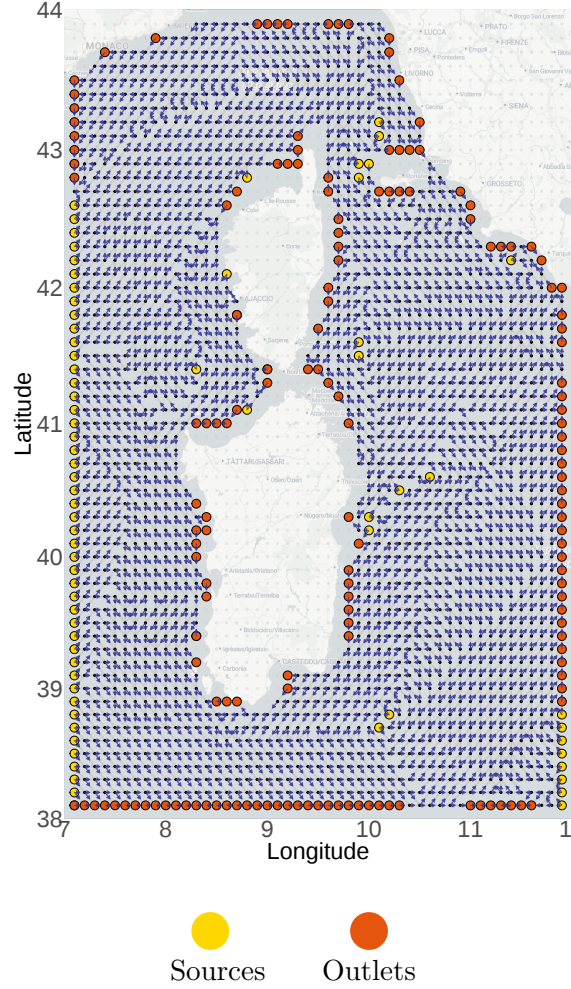

\subsection{A Markov chain representation}
\label{subsec:network_dynamic}

Crucial to the following construction is the definition of how the physical phenomenon underlying the spatial system influences the dynamics within the domain supporting the random process. In this work, the physics of water within the linear network is modeled as a Markov chain $(V \cup \{S\}, \pi)$, whose state space is the union of the set $V$ of vertices of the network with an external absorbing state $S$, called the \textit{sink}. At each vertex, multiple potential paths forward may exist; each edge of the network is thus assigned a transition probability by the transition matrix $\pi$.  The Markovian framework provides a simple yet powerful approach for modeling network dynamics, preserving the system's complexity and generality. Moreover, this approach allows an intuitive understanding of the network's behavior.

To construct the Markov chain, we start by defining its transition probability matrix $\pi$. For every vertex $a \in V$, the velocity $\textbf{v}_a$ is decomposed along the two directions identified by the rule illustrated in \Cref{img:network_construction}. The magnitudes of the resulting two components, normalized by their sum, are collected in a weighted adjacency matrix \(M.\) Details on the computation of \(M\) are reported in the Appendix. If a component identifies an edge \(l_{[a,e]} \in {\cal L}\) of the network, the transition probability \(\pi_{[a,e]}\) is defined as its normalized magnitude. If a component points to a node of the grid \({\cal G}\) that does not belong to \(D\), -- i.e. \(a\) is an outlet -- then its normalized magnitude is added to \(\pi_{[a,S]}\). Finally, we set \(\pi_{[S,S]}=1\). All remaining transition probabilities are set to zero. Note that each row of $\pi$ sums to one. This construction captures the directionality of the underlying velocity field in a manner consistent with the network's topology. Let $\{X_n\}_{n\geq 0}$ be a Markov chain on \(V \bigcup \{S\}\) with transition matrix $\pi$.

As detailed in \Cref{subsec:lin_net}, water enters the system from external sources and eventually flows out. This implies that all vertices of $V$ are transient states of $\{X_n\}_{n\geq 0}$, with $S$ as the unique absorbing state.

We aim at defining the random path connecting two vertices $v, x \in V$ in the network. Fix $v \in V$ and condition on $X_0 = v$. Then $\{X_n \mid X_0 = v\}_{n \geq 0}$ is the random trajectory of the Markov chain starting from $v$. Given the starting point \(v\), let $L_x = \sup\{n \geq 0 : X_n = x\}$ be the last hitting time of $x$. Since $x$ is transient, the chain visits $x$ at most finitely many times, with probability one. Now set $T(v,x)$
to be the random path that records the trajectory from $v$ up to the last visit to $x$. Formally, if \(L_x >0,\)
\[
    T(v,x) = \bigl(l_{\left[X_0=v, X_1\right]},\, \ldots,\,
    l_{\left[X_{L_x - 1},\, X_{L_x}\right]}\bigr),
\]
which is an element of  $\mathcal{P}(v,x)$, equipped with the discrete $\sigma$-algebra; naturally, \(T(v,x)=\emptyset,\) if \(L_x=0.\) Note that $T(v,x)$ may trace a cyclic path, consistently with the definition of $\mathcal{P}(v,x)$.

For \(p_{(v,x)} \in \mathcal{P}(v,x),\)
\begin{equation}
    \mathbb{P}\!\left(T(v,x) = p_{(v,x)}\right) =
    \left(\prod_{l_{\left[a,b\right]} \in p_{(v,x)}} \pi_{\left[a,b\right]}\right)
    U(x),
    \label{eq:probabilityT}
\end{equation}
where the product aggregates the transition probabilities $\pi_{\left[a,b\right]}$ along the edges composing the path $p_{(v,x)}$ (counted with multiplicity for cyclic
paths), and $U(x)$ denotes the probability of never returning to $x$ after arrival. Under transience, $U(x) > 0$ for all $x \in V$. If $x$ is an \textit{acyclic} vertex -- meaning no path connects $x$ to itself -- then the return probability is zero,
implying $U(x) = 1$.

We define a path-dependent distance, denoted by \(dist_{T{\left(v,x\right)}}(\cdot, \cdot)\), by evaluating the length along a specific realization of the random path connecting two vertices. Specifically, given two vertices $v,x \in V$, consider the realization \(T{\left(v,x\right)} = p_{\left(v,x\right)}=\left\{l_1, \dots, l_n\right\}\) and, for $u \in l_1$, the sub-path $p_{(u,x)} \subseteq p_{(v,x)}$. Then, set \(dist_{T{\left(v,x\right)}}(u, x)=|p_{(u,x)}|\). We set the distance to infinity in case the vertices are not connected according to $T(v,x)$ -- i.e. $dist_{T(v,x)}(\cdot, x) = +\infty$ if $T(v,x) = \emptyset$.

Finally, for ease of exposition, let us assume that the collection of Markov chains starting from different vertices \(v \in V\) are independent. In fact, this condition does not impose practical restrictions on the resulting covariance model, as further remarked in \Cref{sec:convprocess}.

\section{A convolution process on linear networks}
\label{sec:convprocess}

\subsection{Moving average construction}
\label{subsec:processDefinition}

The primary objective of this section is to derive a novel, valid covariance structure over the topologically complex domain of a linear network. It is important to note that our primary interest lies not in modeling the underlying generative process itself, but in obtaining a flexible and mathematically valid spatial dependence structure. To achieve this, we introduce a moving average convolution process building upon the stream network strategies of \cite{VerHoef2006} and \cite{rivers2010}. We utilize this process to deduce a closed-form, positive-definite covariance function.

Choosing a linear network $\mathcal{N}$ as the domain for this convolution process offers the key analytical advantage of evaluating spatial contributions via line integrals. By constraining the velocity field to the edges and vertices, the formulation naturally admits an integral representation along the network structure. While this moving average construction is inspired by the dendritic stream network models of \cite{rivers2010}, we integrate an additional source of randomness. This extension is necessary to accommodate the greater topological complexity—such as cycles and multidirectional connectivity—inherent to general ocean current networks compared to strictly directed streams.

In the following construction, we denote by $W$ a Wiener process independent of the Markov process \(\{X_n\}_{n \geq 0}\), introduced in the previous section, and by $\int_{A} [\cdot]\, W(du)$ the stochastic integral over $A\subset \mathbb{R}$. We let $g$ be a square-integrable function, which will play the role of moving average function. Finally, for \(v,x \in V\), we denote by $X_1$ the (random) right vertex of the first edge in the path \(T{\left(v,x\right)}\). 

Given a family of positive parameters \(\beta_p\) indexed by paths between vertices in \(V,\) we now define a process $\{Z_x , x \in V\}$ on the network $\mathcal{N}.$ For \(x \in V,\) set 
\begin{equation}
    Z_x = \sum_{v \in V} \int_{L_{\left[v, X_1\right]}} \frac{g\left(dist_{T{\left(v,x\right)}}\left(u,x\right)\right)}{\sqrt{\beta_{T{\left(v,x\right)}}}} W(du).
    \label{eq:myprocess}
\end{equation}
In \eqref{eq:myprocess}, the integral is constrained to the (random) segment $l_{\left[v, X_1\right]}$, ensuring that each vertex $v \in V$ contributes exactly once to the process, and only if it is connected to $x$ via the path $T(v,x)$. Moreover, the random paths $T{\left(v,x\right)}$'s influence both the interval of integration and the distance between points in the network. Square-integrability of $g$ guarantees the well-definition of \eqref{eq:myprocess}. The role of $\beta_{T(v,x)}$ will become clearer in the following subsections, where it plays a key part in ensuring a relaxed condition for stationarity.

For a more thorough interpretation of the process, the terminology of upstream and downstream points will be borrowed from works on stream networks, together with the specific type of images depicting the effects of the moving averages functions \citep{VerHoef2006, peterson2007, rivers2010}. A single realization of the process $Z_x$ can be viewed as representing an individual unit within a larger flow. The random quantity modeled by the process corresponds to this unit, which evolves according to the specific realizations of the random variables $T$'s. When the moments are evaluated, ``global'' measures for the whole flow are deducted.

For instance, refer now to \Cref{img:splitnetwork}. The dependency is represented in the picture by the colored areas, picturing the moving average kernels. A split $v^*$ in the linear network generates a dependence between the vertex $v$ before the split and the (possibly more than two) vertices after it, say $x_1$ and $x_2$ -- where the terms \emph{before} (resp. \emph{after}) and \emph{upstream} (resp. \emph{downstream}) are given according to the direction of the flow. The flow passing through $l_{[v,v^*]}$ influences the two downstream vertices $x_1$ and $x_2$ differently. According to the realizations of the random variables $T(v,\cdot)$'s, at each split, each individual unit of the flow in the network follows one and only one direction. In \Cref{img:splitnetwork}a, the unit selects path toward the vertex $x_1$, hence leading to the realizations $T_{\left(v,x_1\right)} = \{l_{\left[v,v^*\right]}, l_{\left[v^*x_1\right]}\}$, and $T_{\left(v,x_2\right)} = \emptyset$. Thus, the random value of the process in $x_1$, i.e. $Z_{x_1}$, depends on the Wiener process defined over $l_{\left[v,v^*\right]}$, while $Z_{x_2}$ does not. Conversely, in \Cref{img:splitnetwork}b a different realization is pictured, and $Z_{x_2}$ depends on the Wiener process of $l_{\left[v,v^*\right]}$, while $Z_{x_1}$ does not. When two points are not connected -- i.e., $T(v,x) = \emptyset$ -- their distance is set to infinite, which removes any contribution from the upstream vertex to the random value of the process at the downstream vertex.

\begin{figure}
    \centering
    \begin{minipage}{0.45\linewidth}
        \centering
        \begin{tikzpicture}[scale=1]
            \vspace{.5 cm}
            \draw[thick,-] (-3,0) -- (-1,0);
            \draw[thick,-] (-1,0) -- (2,1);
            \draw[thick,-] (-1,0) -- (2,-1);
            \fill (-3,0) circle (2pt) node[above left]{$v$};
            \fill[col3!100] (-2.6,0.05) -- (-1,0.05) -- (1.5,0.88) -- (1.2,1.5) -- (-1,0.4) -- cycle;
            \fill (1.5,0.83) circle (2pt) node[below right]{$x_1$};
            \fill (1.5,-0.83) circle (2pt) node[above right]{$x_2$};
            \fill (-1,0) circle (2pt) node[below left]{$v^*$};

            \draw[ultra thick, ->] (-2.2,-1.1) -- (-1, -1.1) node[below left] {Flow};
        \end{tikzpicture}
        \caption{(a)}
    \end{minipage}
    \hfill
    \begin{minipage}{0.45\linewidth}
        \centering
        \begin{tikzpicture}[scale=1]
            \vspace{.5 cm}
            \draw[thick,-] (-3,0) -- (-1,0);
            \draw[thick,-] (-1,0) -- (2,1);
            \draw[thick,-] (-1,0) -- (2,-1);
            \fill (-3,0) circle (2pt) node[above left]{$v$};
            \fill[col1!100] (-2.6,-0.05) -- (-1,-0.05) -- (1.5,-0.88) -- (1.2,-1.5) -- (-1,-0.4) -- cycle;
            \fill (1.5,0.83) circle (2pt) node[below right]{$x_1$};
            \fill (1.5,-0.83) circle (2pt) node[above right]{$x_2$};
            \fill (-1,0) circle (2pt) node[above left]{$v^*$};

            \draw[ultra thick, ->] (-2.2,-1.1) -- (-1, -1.1) node[below left] {Flow};
        \end{tikzpicture}
        \caption{(b)}
    \end{minipage}
    \caption{A split $v^*$ in the network. The two areas green (up) and blue (down) represent the moving average kernels of the two random values defined in the two downstream vertices ($x_1$ and $x_2$ respectively). At the split, the Markov chain will select either one of the two vertices downstream. In (a), $T_{\left(v,x_1\right)}(\omega) = p_{(v,x_1)}$ where $p_{(v,x_1)}$ is the path connecting $v$ and $x_1$, while $T_{\left(v,{x_2}\right)}(\omega) = \emptyset$. Conversely, in (b) the opposite happens: $T_{\left(v,x_1\right)}(\omega) = \emptyset$; $T_{\left(v,x_2\right)}(\omega) = p_{(v,x_2)}$.}
    \label{img:splitnetwork}
\end{figure}

The following propositions provides the expressions of the moments of the process $\{Z_x , x \in V\}$ defined in \Cref{eq:myprocess}. Proofs are reported in the Appendix.

\begin{proposition}
    The process $\{Z_x , x \in V\}$ is zero-mean: for $x \in V$,
    $
    \mathbb{E}\left[Z_x\right] = 0.
    $
\end{proposition}

We evaluate the moments of the process in \Cref{eq:myprocess} by first exploiting the stochastic integral, and secondly evaluating the distribution of the $T(v,\cdot)$.

\begin{proposition}
    For $x, y$ in $V$, 
    \begin{equation}
    Cov(Z_x, Z_y) = \sum_{v \in V} \mathbb{E}\left[ \int_{L_{\left[v,X_1\right]}} \frac{g\left(dist_{T{\left(v,x\right)}} \left(u,x\right)\right)g\left(dist_{T{\left(v,y\right)}} \left(u,y\right)\right)}{\sqrt{\beta_{T{\left(v,x\right)}}\beta_{T{\left(v,y\right)}}}} du \right].
    \label{eq:moment_second}
    \end{equation}

    In particular, 
    \begin{equation}
        Var(Z_x) = Cov(Z_x, Z_x) = \sum_{v \in V} \mathbb{E}\left[ \int_{L_{\left[v,X_1\right]}} \frac{g^2\left(dist_{T{\left(v,x\right)}} \left(u,x\right)\right)}{\beta_{T{\left(v,x\right)}}} du \right].
    \label{eq:variance_moment}
    \end{equation}
    \label{prop:moments}
\end{proposition}

Note that the variance of the process in \Cref{eq:variance_moment} is finite. The sum is on a finite number of vertices, and each term is given by finite integrals -- because of the square-integrability of $g$.

Importantly, the assumption of path independence does not limit the flexibility of our model, which aims at specifying a valid second-order covariance structure on linear networks. In the convolution framework, integrating against the Wiener process effectively neutralizes any joint dependence between any two paths $T(v_1, x)$ and $T(v_2, x)$ during the evaluation of the second moment (see the proof of Proposition \ref{prop:moments}). The resulting spatial covariance depends exclusively on the marginal path probabilities. 

Given the non-Euclidean topology of the domain and the structure of the covariance in \eqref{eq:moment_second}, enforcing classical second-order stationarity is generally infeasible. Instead, inspired by approaches used in stream network models \citep{VerHoef2006}, we adopt a relaxed notion of stationarity based on spatial homogeneity of the first two moments--specifically, constant mean and constant variance.

\subsection{Weighted covariance model}
\label{subsec:stationary_model}

We enforce stationary variances following the approach of \citep{VerHoef2006, rivers2010}. The parameters $\beta_{p}$'s are used to this end. Indeed, for \(v,x \in V,\) we specify $\beta_{p_{(v,x)}}$ as depending on the total probability mass in the end points of the edges composing the path, and on the probability of not returning to the endpoint:
\begin{equation}
    \beta_{p_{(v,x)}} = \left(\prod_{l_{[a,b]}\in p_{(v,x)}} \Big[\sum_k \pi_{[k,b]}\Big]\right) U(x).
    \label{eq:betaexplicit}
\end{equation}
Note that, for any edge $l_{\left[a,b\right]}$ belonging to a path with a non zero probability of realization, the factor $\sum_k \pi_{[k,b]} > 0$. Moreover, as observed in \Cref{subsec:network_dynamic}, $U(x) > 0$ for each $x \in V$. 

The following results provide the expressions for the variance and covariance of the random process defined in \eqref{eq:myprocess}, under the specification of the normalization constants $\beta_{p(v,x)}$ set in \eqref{eq:betaexplicit}.

\begin{proposition}
Let $(\mathcal{L}, V)$ be a linear network, $\{Z_x, x \in V\}$ be the process defined in \Cref{eq:myprocess} and the normalization constants $\beta$'s be defined as in \Cref{eq:betaexplicit}. Then for every $x\in V$ the marginal variance is constant and equals

\begin{equation}
    Var(Z_x) \;=\; \int_0^{+\infty} g^2(r)\,dr.
    \label{eq:variance}
\end{equation}
\label{prop:variancestationary}
\end{proposition}

Note the the variance in \Cref{eq:variance} is finite thanks to the square-integrability of the moving average function $g$, as noted in \Cref{subsec:processDefinition}. More generally, constant marginal variance holds for any choice of the constants $\beta$'s inducing a \emph{unit-influx} normalization such that \(\sum_k \pi_{\left[k,x\right]}/\beta_{p(k,x)} = 1\), analogously to the river network literature \citep{VerHoef2006, cressie2006}.
 
For \(x,y \in V,\) we denote by $U(x,y)$ the probability of not returning to either $x$ or $y$, when the chain starts from $x$. In addition, let $\tilde{\mathcal{P}}(x,y)$ denote the set of paths from $x$ to $y$ that do not contain cycles at the initial vertex $x$. Notably, for an acyclic point $x^*$, $\tilde{\mathcal{P}}(x^*,y) = \mathcal{P}(x^*,y)$.

\begin{proposition}
Let $(\mathcal{L}, V)$ be a linear network, $\{Z_x, x \in V\}$ be the process defined in \Cref{eq:myprocess} and the normalization constants $\beta$'s be defined as in \Cref{eq:betaexplicit}. Let \(x,y \in V\), \(x \neq y.\) 

\begin{itemize}
\item[(i)] If $\, \, \tilde{\mathcal{P}}(x,y) \bigcup \tilde{\mathcal{P}}(y,x) \neq \emptyset$, 
    \begin{equation*}
        Cov(Z_x, Z_y) = \sum_{p \in \tilde{\mathcal{P}}(x,y) \bigcup \tilde{\mathcal{P}}(y,x)} w_{p} C(|p|),
    \end{equation*}
where, for any two vertices \(v_1,v_2 \in V,\)
    \begin{equation}
    w_{p_{(v_1,v_2)}} = \left(\prod_{l_{\left[a,b\right]} \in p_{(v_1,v_2)}} \frac{\pi_{\left[a,b\right]}}{\sqrt{\sum_{k} \pi_{\left[k,b\right]}}}\right) \frac{U(v_2,v_1)}{\sqrt{U(v_1)\, U(v_2)}},
    \label{eq:weightexplicit}
    \end{equation}
    and $C(\cdot)$ is the spatial covariance function defined, for $h \geq 0,$ as $$C(h) = \int_0^{+\infty} g(r)g(r + h) dr.$$    
\item[(ii)]
If $\, \, \tilde{\mathcal{P}}(x,y) \bigcup \tilde{\mathcal{P}}(y,x) = \emptyset,$
\begin{equation*}
        Cov(Z_x, Z_y) = 0.
    \end{equation*}
\end{itemize}    

\label{prop:covariancestationary}
\end{proposition}

The covariance structure is highly sensitive to the choice of the moving average function. In particular, specific forms of \( g \) result in classical parametric expressions for the spatial covariance function, analogous to those introduced in \cite{VerHoef2006, rivers20102, Barbi_Menafoglio_Secchi_2023}.
\Cref{table:typeofg} illustrates three common choices used as spatial covariance function $C(\cdot)$, with their corresponding moving average kernels $g$. These representations of \(C\) depend on two parameters which have a similar role as those used in classical geostatistics \citep{Cressie_1993}: the range $\theta_r>0$, which scales the distance between vertices, and the sill $\theta_s>0$, which sets the overall variance at any given vertex. Indeed, note that by setting \(\int_0^{+\infty}g(r)^2dr = \theta_s\) the unique parameter controlling the amount of variability in the covariance structure is the sill parameter $C(0) = \theta_s$. 

\begin{table}[b]
\caption{Traditional moving average and spatial covariance function.}
\label{table:typeofg}
\renewcommand{\arraystretch}{1.75}
    \begin{tabular}{@{}lrl@{}}
     & \multicolumn{1}{c}{\textbf{Moving average function}} & \multicolumn{1}{c}{\textbf{Spatial covariance function}} \\
    \hline
    \textbf{Linear with sill} & $g(r) = \sqrt{\theta_s} \, \mathbb{I}_{\left\{0 \leq r/\theta_r \leq 1 \right\}}$ & $C(h) = \theta_s \, (1-h/\theta_r) \mathbb{I}_{\left\{0 \leq h/\theta_r \leq 1 \right\}}$ \\
    \hline
    \textbf{Spherical} & $g(r) = \sqrt{3 \, \theta_s}(1 - r/\theta_r) \mathbb{I}_{\left\{0 \leq r/\theta_r \leq 1 \right\}}$ & $C(h) = \theta_s(1 - \frac{3}{2}h/\theta_r + \frac{1}{2}h^3/\theta_r) \mathbb{I}_{\left\{0 \leq h/\theta_r \leq 1 \right\}} $ \\
    \hline
    \textbf{Exponential} & $g(r) = \sqrt{2\, \theta_s}e^{-r/\theta_r} \mathbb{I}_{\left\{0 \leq r/\theta_r \right\}}$ & $C(h) = \theta_se^{-h/\theta_r}$ \\
    \hline
    \end{tabular}
\end{table}

With this notation, for any two vertices $x$ and $y$, the covariance model of the process defined in \Cref{eq:myprocess}, with the family of parameter $\beta$'s defined as in \Cref{eq:betaexplicit} can be represented as
\begin{equation}
    Cov(Z_x,Z_y) = 
    \begin{cases}
        C(0) = \theta_s & \text{if } x \equiv y, \\
        0 & \mbox{if} \, x\neq y; \, \tilde{\mathcal{P}}(x,y) \bigcup \tilde{\mathcal{P}}(y,x)= \emptyset,\\
        \sum_{p \in \tilde{\mathcal{P}}(x,y) \bigcup \tilde{\mathcal{P}}(y,x)} w_p C(|p|) & \mbox{if} \, x\neq y; \, \tilde{\mathcal{P}}(x,y) \bigcup \tilde{\mathcal{P}}(y,x) \neq \emptyset.
\end{cases}
    \label{eq:covariance}
\end{equation}

The formulation in \Cref{eq:covariance} recovers the weighting structure typical of stream network models \citep{VerHoef2006, rivers2010}. In particular, the model is consistent with the stream network construction when the domain is restricted accordingly, as shown in the Appendix. We note that we here assume the range \(\theta_r\) and the sill \(\theta_s\) to be spatially homogeneous. The extension to spatially varying parameters is left to future research. 

Note that the covariance between pairs of vertices depends solely on the set of possible paths connecting them within the linear network. In particular, it is determined by the path weights and their associated distances. This construction enables the direct encoding of classical geostatistical effects -- such as scale and marginal variance -- onto the network domain itself, without relying on Euclidean embeddings or latent spatial representations. As a result, the model naturally adapts to the topology of the network, preserving interpretability while remaining intrinsic to the domain.

To illustrate the distinction between our network-based approach and conventional Euclidean modeling, Figure \ref{img:contour} presents covariance heatmaps for both frameworks. Both use an exponential covariance with $\theta_s = 1$ and range set to ten times the maximum distance in each metric. Reference points are located on the western ($40.8^{\circ}$N, $7.1^{\circ}$E) and eastern ($41.1^{\circ}$N, $10.7^{\circ}$E) coasts of Sardinia and Corsica. The Euclidean model assumes isotropy, with covariance depending solely on straight-line distance. This produces unrealistic dependencies: points on opposite sides of Sardinia exhibit strong covariance despite having no hydrological connection due to the island barrier. In contrast, the network-based model respects directional constraints imposed by sea currents, yielding anisotropic covariance structures that reflect true marine connectivity. The resulting patterns are shaped by flow dynamics, capturing spatial dependencies consistent with underlying physical transport processes.

\begin{figure}
  \centering
  \begin{minipage}[b]{0.32\linewidth}
    \centering
    \includegraphics[width=0.75\linewidth]{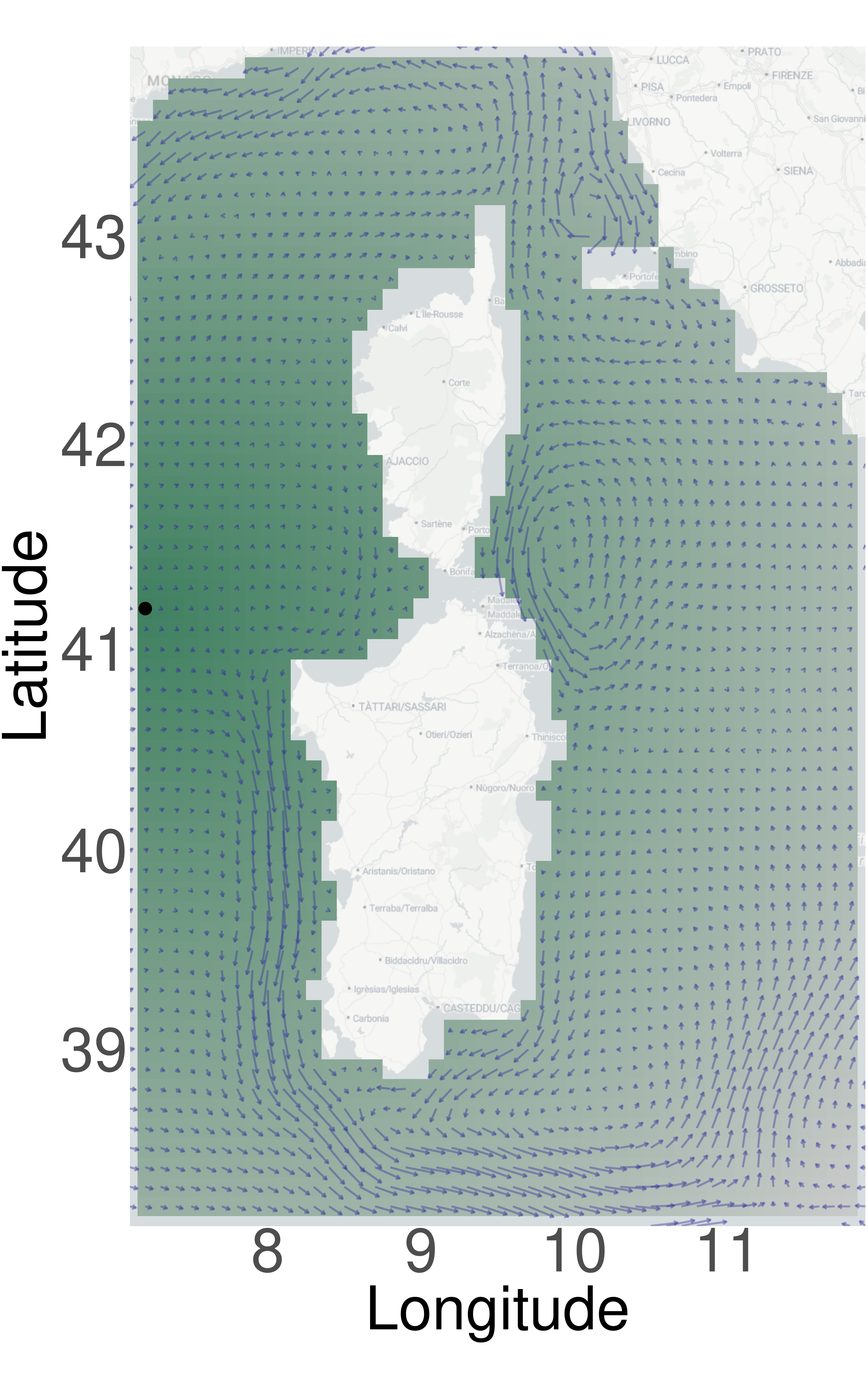}
    \vspace{0.2cm}
    \\ (a) Euclidean metric.\newline Point $41.1^{\circ} \text{N},  10.7^{\circ} \text{E}$
  \end{minipage}\hfill
  \begin{minipage}[b]{0.32\linewidth}
    \centering
    \includegraphics[width=0.75\linewidth]{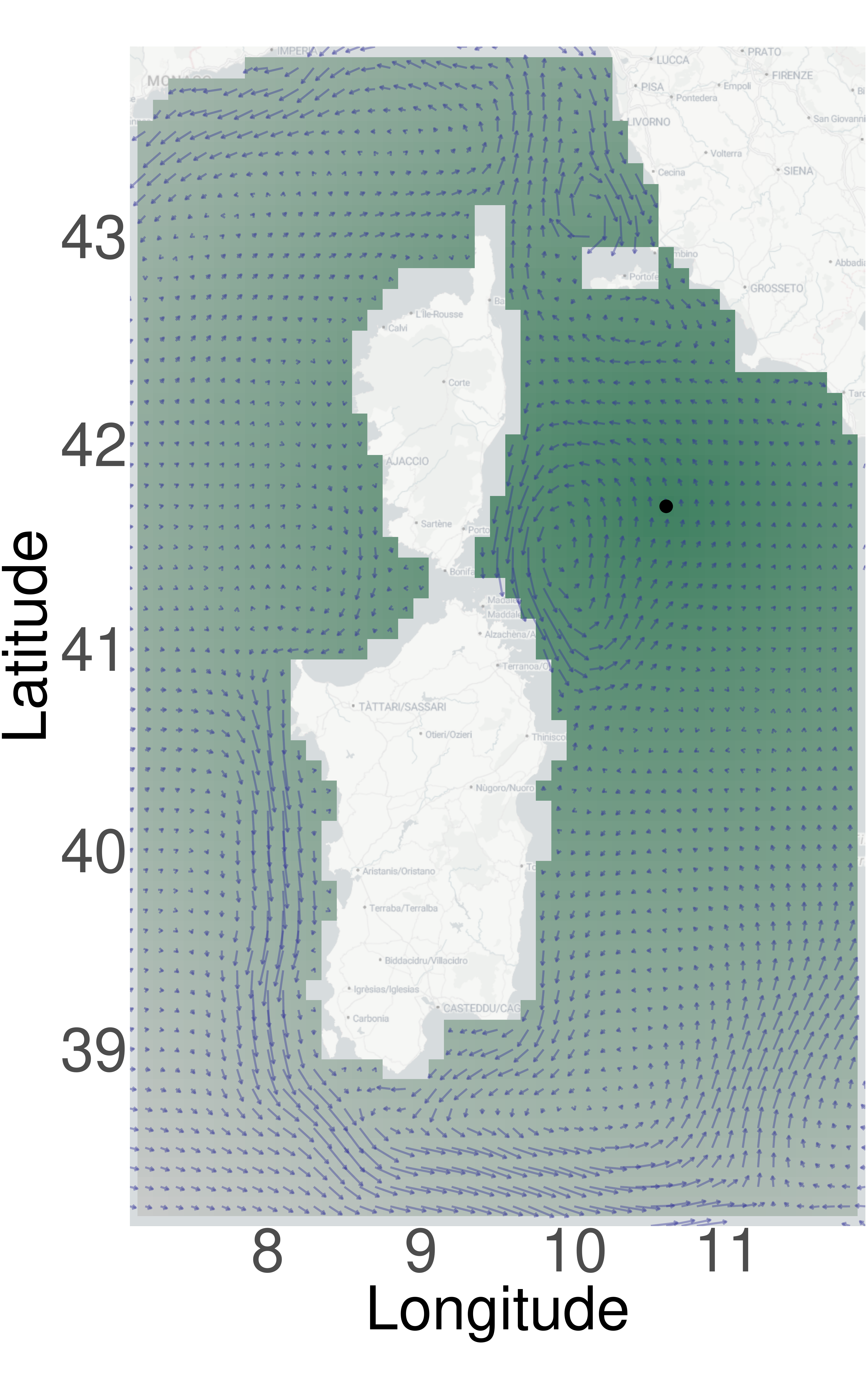}
    \vspace{0.2cm}
    \\ (b) Euclidean metric.\newline Point $40.8^{\circ} \text{N}, 7.1^{\circ} \text{E}$
  \end{minipage}\hfill
  \begin{minipage}[b]{0.32\linewidth}
    \centering
    \includegraphics[width=0.75\linewidth]{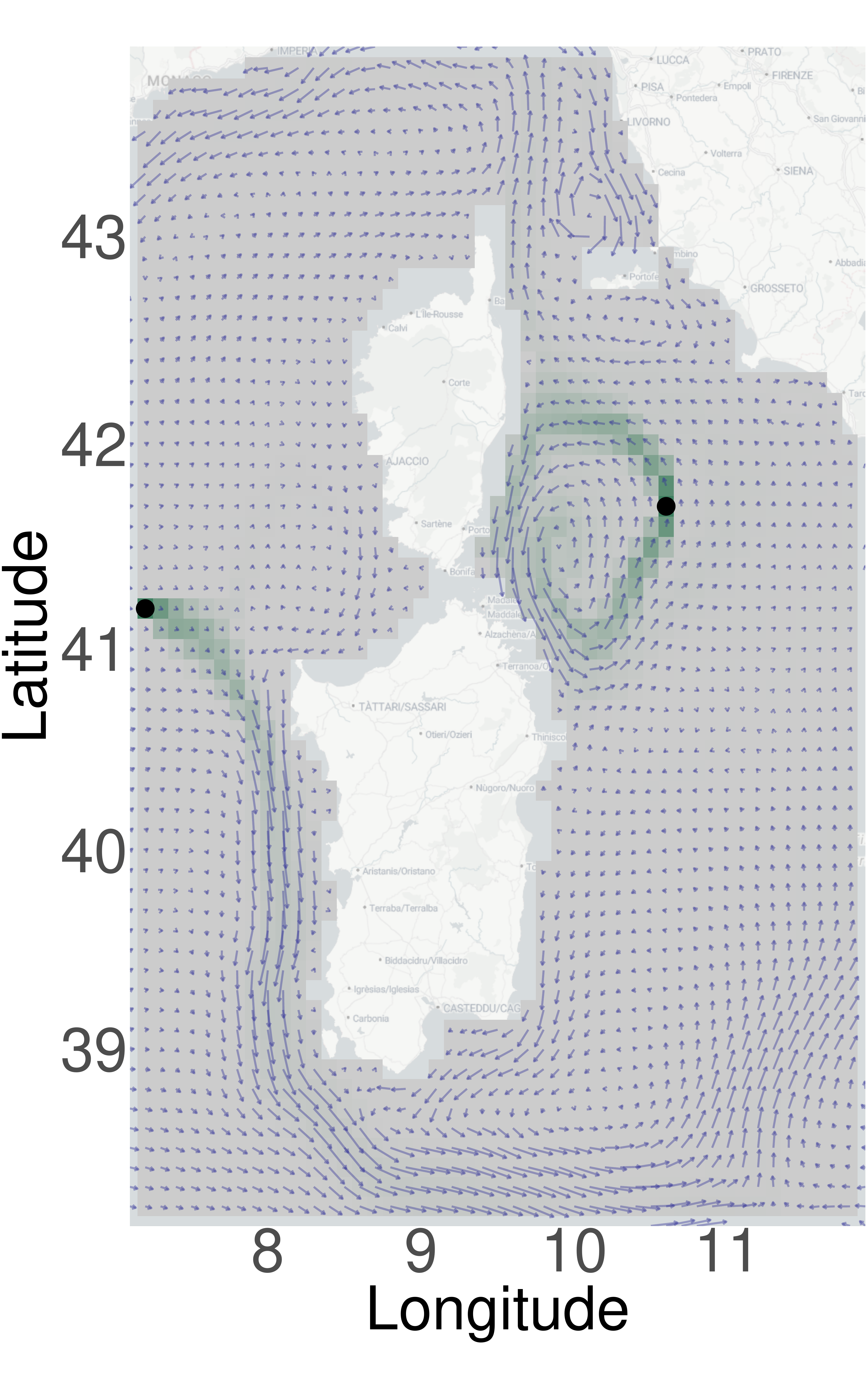}
    \vspace{0.2cm}
    \\ (c) Physics based metric.\newline Point $40.8^{\circ} \text{N}, 7.1^{\circ} \text{E}$ \newline Point $41.1^{\circ} \text{N},  10.7^{\circ} \text{E}$
  \end{minipage}

  \vspace{0.5cm}
  \begin{tikzpicture}
  \begin{axis}[
    hide axis,
    scale only axis,
    height=0pt,
    width=0pt,
    xmin=0, xmax=1,
    ymin=0, ymax=1,
    colorbar horizontal,
    colormap={bw-to-col3}{[1pt] color(0pt)=(white); color(100pt)=(col3)},
    point meta min=0,
    point meta max=1,
    colorbar style={
      xtick={0,1},
      xticklabels={0,1},
      width=6cm, 
      height=0.3cm,
      xticklabel style={font=\small}
    }
  ]
  \end{axis}
\end{tikzpicture}
  \caption{Spatial heatmaps for different covariance models. The color gradient indicates the covariance of each location with respect to the reference point (black dot). In the proposed framework, the effects of both influence points are superimposed in a single map, whereas the Euclidean framework requires two separate maps to avoid visual interference between the individual influence points.}
  \label{img:contour}
\end{figure}

\section{Estimation of the covariance structure}
\label{sec:estimation}

When data become available, the primary goal of a geostatistical analysis is to estimate the covariance structure of the underlying random field. However, the specific form introduced in \Cref{eq:covariance} presents challenges that make classical estimation procedures difficult to apply directly. In particular, the summation over paths complicates the use of standard empirical covariances based on point pairs, as it introduces ambiguity in attributing observed covariance to specific distances.

Weighted covariance structures arise in stream network geostatistics \citep{rivers2010}, where efforts to develop unbiased estimators incorporating these weights have faced significant difficulties. Such estimators tend to exhibit high variability and have proven unreliable in practice, leading the literature to favor biased but more stable unweighted alternatives \citep{torgegram}. In this work, we construct a novel penalized estimator that explicitly accounts for both the path summation in the covariance model and the associated weights, using penalization to ensure stability and reliability.

We develop an estimation algorithm assuming the stationary variance condition introduced in \Cref{subsec:stationary_model}. This assumption is implemented by defining the family of parameters $\beta$'s as in \Cref{eq:betaexplicit}. The weights $w_p$ are used in the sense of their definition given in \Cref{eq:weightexplicit}.

First notice that for a zero-mean stationary process $\{Z_x, x \in V\}$ whose covariance function has a finite sill \(\theta_s,\)
\begin{equation}
Cov(Z_x,Z_y) = \theta_s - \frac{Var(Z_x - Z_y)}{2}
\label{eq:semivariances}
\end{equation}
for all \(x,y \in V\).

We begin by estimating the sill parameter, \(\theta_s\), which represents the variance of the process. Consider pairs of vertices that are unconnected. For such pairs, the covariance is zero and the variance of their difference reduces to \(2\,\theta_s\). Define $\mathcal{H}_{\infty} = \{(x,y) \in V\times V : \tilde{\mathcal{P}}(x,y) \bigcup \tilde{\mathcal{P}}(y,x) = \emptyset\}$ as the set of all such pairs. This yields the unbiased estimator
\begin{equation}
    \hat{\theta}_s = \frac{1}{2|\mathcal{H}_{\infty}|} \sum_{(x,y) \in \mathcal{H}_{\infty}} (Z_x - Z_y)^2.
    \label{eq:varianceprocessestimator}
\end{equation}

Let now $\boldsymbol{\gamma} \in \mathbb{R}^n$ denote the vector of empirical semi-variances computed from all distinct pairs of vertices, where $n$ is the number of pairs in the linear network. That is, the $i$-th component of $\boldsymbol{\gamma}$ is 
$$
\gamma_i= \frac{(Z_{x_i}-Z_{y_i})^2}{2}
$$ 
if $(x_i,y_i)$ is the $i$-th pair among the distinct pairs of vertices of the linear network.
Because of \Cref{eq:semivariances}, we consider the method-of-moments system of estimating equations 
\begin{equation}\label{eq:mom}
\hat\theta_s - \gamma_i=\sum_{p \in \tilde{\mathcal{P}}(x_i,y_i) \cup \tilde{\mathcal{P}}(y_i,x_i)} w_p C(|p|).
\end{equation}
where each equation is indexed by a pair of distinct vertices $(x_i,y_i).$

Following standard geostatistical practice \citep{Cressie_1993}, we now reduce the dimensionality of the system (\ref{eq:mom}) by introducing distance classes $\mathcal{H}_1,\ldots,\mathcal{H}_l$ that partition all path lengths into $l$ disjoint bins. We approximate the covariance function \(C\) via $l$ representative values $C(h_j)$ for $j = 1,\ldots,l$, where $h_j$ denotes the mean path distance within bin $\mathcal{H}_j$.
Let \(\boldsymbol{C} = (C(h_1) \ldots C(h_l))^\top\) be the vector of unknown covariance values. Then, the system of equations (\ref{eq:mom}) is approximated by the linear system:
\begin{equation}\label{eq:mom_approx}
\hat{\theta}_s\boldsymbol{1}- \boldsymbol{\gamma}= W\,\boldsymbol{C},
\end{equation}
where $\boldsymbol{1} \in \mathbb{R}^n$ and the weight matrix $W \in \mathbb{R}^{n \times l}$ aggregates the path weights. Specifically, the entry $W_{\left[i,j\right]}$ corresponds to the $i$-th vertex pair and the $j$-th distance lag:

\begin{equation}
    W_{\left[i,j\right]} =  \sum_{\substack{p \in \tilde{\mathcal{P}}\left(x_i,y_i\right) \cup \tilde{\mathcal{P}}\left(y_i,x_i\right) \\ |p| \in \mathcal{H}_j}} w_p.
    \label{eq:constructionW}
\end{equation}
Thus, $W_{\left[i,j\right]}$ represents the total weight of all paths connecting the $i$-th pair that fall within the $j$-th distance class.

Inverting \Cref{eq:mom_approx} poses challenges analogous to those encountered in stream network geostatistics \citep{torgegram}. The weight matrix $W$ is frequently ill-conditioned, rendering ordinary least-squares estimates unstable. To enforce stability, we adopt a penalized least-squares approach. The estimator is obtained by minimizing the regularized objective function:
\begin{equation}
\mathcal{L}(\boldsymbol{C}) = \frac{1}{2}\| W\boldsymbol{C} - (\hat{\theta}_s \boldsymbol{1} - \boldsymbol{\gamma}) \|^2 + \frac{1}{2} \lambda\|\boldsymbol{C}\|^2
\label{eq:minimizationproblem}
\end{equation}
where $\lambda > 0$ is a regularization parameter. The ridge-type penalty is a standard choice in inverse problems; it stabilizes the solution by reducing the variance of the estimator.

\begin{proposition}
    The penalized estimator for the covariance function of the process is

    \begin{equation}
    \hat{\boldsymbol{C}} = \left(W^TW + \lambda I\right)^{-1}W^T\left(\hat{\theta}_s\boldsymbol{1} - \boldsymbol{\gamma}\right).
    \label{eq:penalizedestimator}
    \end{equation}

    Moreover, $\lambda \geq \frac{\|W^T\left(\boldsymbol{\hat{\theta}}_s - \boldsymbol{V}\right)\|_{\infty}}{\hat{\theta}_s} - \min_i \delta_i$, where \(\delta_i = \left|\left(W^TW\right)_{\left(i,i\right)}\right| - \sum_{j \neq i}\left|\left(W^TW\right)_{\left(i,j\right)}\right|\), guarantees that $\|\hat{\boldsymbol{C}}\|_{\infty} \leq \hat{\theta}_s$.
    \label{prop:penalizedestimator}
\end{proposition}

The regularization parameter $\lambda$ governs the bias-variance trade-off: larger values enforce stability at the expense of bias, while smaller values approach the non regularized solution, exhibiting higher variance. In our framework, we select the minimal $\lambda$ that satisfies the admissibility condition established in Proposition \ref{prop:penalizedestimator}.

Once the discretized covariance estimates $\hat{\boldsymbol{C}}$ are obtained, we fit a parametric model (e.g., see \Cref{table:typeofg}) to capture the continuous spatial dependence. Specifically, we estimate the range parameter $\theta_r$ by minimizing the distance between the theoretical model and the non-parametric estimates. This two-stage approach results in a fully specified covariance structure for the spatial process.

We remark that a variogram estimator could theoretically be derived via the identity
\[
\hat{\gamma}(h) = \hat{\theta}_s - \hat{C}(h).
\]

\section{Summary of the supporting simulation studies}
\label{sec:summarysimulation}

We briefly describe the behavior of the proposed covariance structure in a simulation setting; a detailed account is provided in the Appendix.

We consider multiple realizations from a Gaussian field whose covariance structure is specified through a network-based exponential model as built in Section \ref{sec:convprocess}. We assess the ability of our approach to estimate this covariance structure from the simulated data by comparing our framework with the classical Euclidean approach. The comparison is conducted across five values of the range parameter, while keeping the sill parameter fixed. The domain is a less refined version of the network in \Cref{img:network}, constructed from real data following the procedure in \Cref{subsec:lin_net}. The covariance matrix between vertices is then built in a train-test setting.

The first step concerns estimation of the covariance parameters $\theta_s$ and $\theta_r$. The sill is consistently recovered, while the range parameter is systematically underestimated. This behavior is consistent with findings from analogous studies in the stream-network setting \citep{torgegram, Barbi_Menafoglio_Secchi_2023}. Nevertheless, the Euclidean framework fails to retain meaningful spatial dependence, producing semi-variogram estimates close to a flat line, whereas our estimator recovers a meaningful spatial structure even under range underestimation. The result of this comparison is confirmed quantitatively via mean squared estimation error.

We then assess covariance reconstruction, measuring how closely the matrix built from estimated parameters approximates the true one. Our framework clearly outperforms the Euclidean baseline on both considered metrics.

Finally, we evaluate out-of-sample prediction via ordinary kriging. Our approach reconstructs the spatial field more accurately, while the Euclidean framework fails to capture the dependence induced by the velocity field.

The simulation results confirm the superior performance of our construction and estimation procedures compared with Euclidean benchmarks. They highlight how physics-driven analyses allow for a more faithful representation of the spatial dependence structure, which in turn has a significant impact on the predicted patterns.

\section{Case Study: quantifying uncertainty for RCP scenarios}
\label{sec:casestudy}

\subsection{Estimation of the residual's covariance structure}
\label{subsec:estimationvariogram}

To generate simulations that capture realistic spatial patterns, we first characterize the residual covariance structure of the water temperature projections (\Cref{subsec:data}). We define residuals as the point-wise difference between RCP projections and satellite observations over the overlapping period 2006–2022. These residuals serve as the basis for selecting and estimating the parametric covariance model. Under a time-invariant covariance assumption, residuals and covariance structure are evaluated and estimated independently for each year.

Data alignment is achieved by assigning each satellite observation to its nearest RCP grid point. Given the fine resolution of the RCP grid, hereafter we neglect the spatial discretization error resulting from this assignment. We first estimate the bias occurring in the projection. Then, using the velocity field derived from the satellite data (\cite{copernicus}), we build, for each year, the linear network according to the procedure described in \Cref{subsec:lin_net}. Based on the unbiased residual observations at the network vertices, we estimate the residuals' covariance structure following the estimation procedure presented in \Cref{sec:estimation}, separately for each year from 2006 to 2022.

\Cref{img:covarianceresiduals} displays the annual empirical covariance functions using functional box-plots \citep{funcboxplot}. The estimates are computed over $l = 15$ distance bins. The visualization highlights the functional depth of the curves: the dark shaded area represents the 50\% central region (the "bag"), while the lighter area delineates the 90\% envelope. The black solid line indicates the point-wise mean, which is used for parameter estimation. We selected the exponential kernel as it best describes the observed spatial decay among the candidates in \Cref{table:typeofg}.

\begin{figure}
    \centering
    \includegraphics[width=0.5\linewidth]{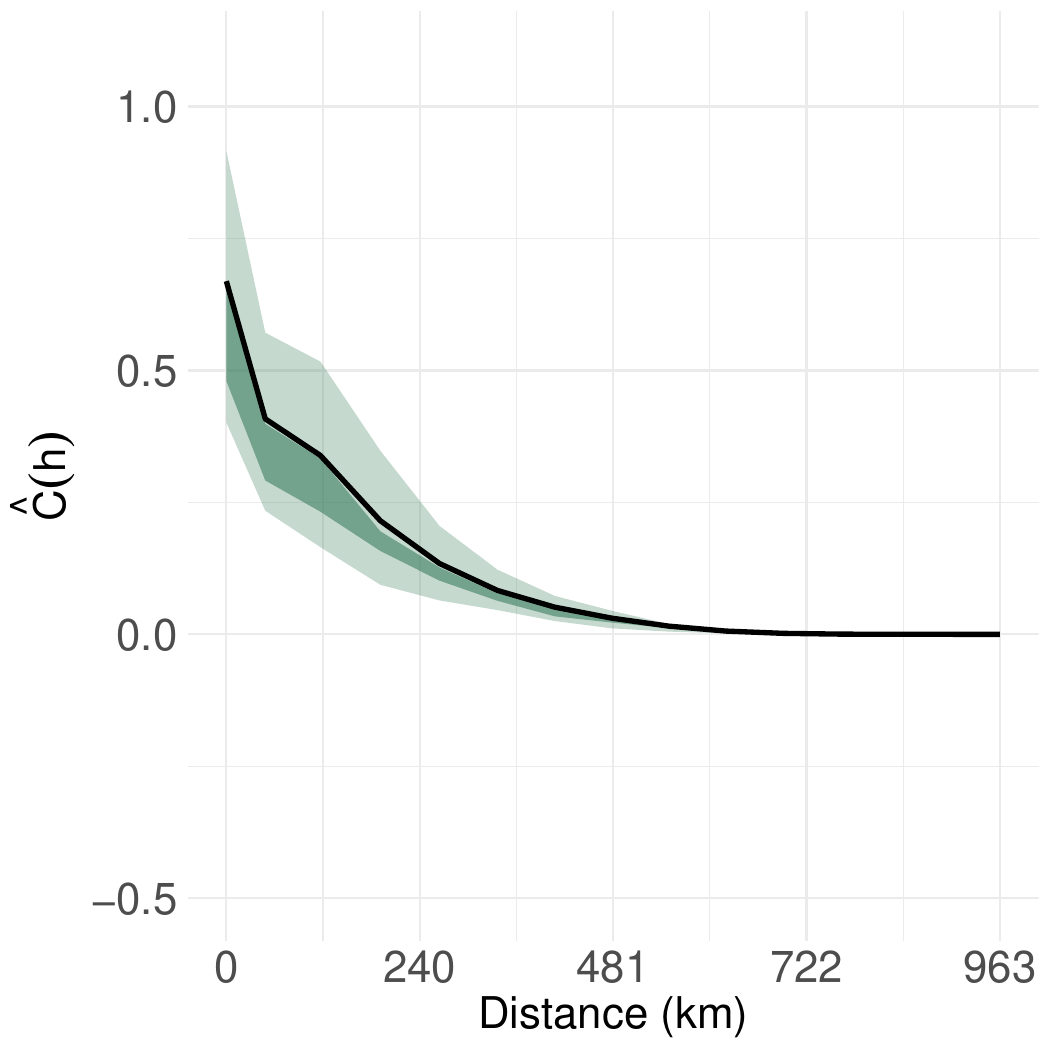}
    \caption{Empirical covariance functions for the period 2006--2022 and point-wise mean. The estimated parameters from the point-wise mean are $\hat{\theta}_s = 0.66$ and $\hat{\theta}_r = 154$ km.}
    \label{img:covarianceresiduals}
\end{figure}

Leveraging the estimated residual covariance structure, we simulate realizations of Sea Surface Temperature (SST) within the Gaussian framework established in \Cref{subsec:data}. Formally, we treat the SST field $\{Y_s, s \in D\}$ as a Gaussian Random Field observed in $\{s_1, \dots, s_n\} \in D$, hence $\{Y_{s_1}, \dots, Y_{s_n}\} \sim \mathcal{N}_n(\boldsymbol{\mu}, \Sigma)$. The mean is estimated by $\hat{\boldsymbol{\mu}}$, the RCP projections corrected by the bias estimated in the observational period 2006-2022. The covariance is estimated by $\hat{\Sigma}$ which follows the construction detailed in \Cref{sec:convprocess}, leveraging the latter estimated parameters. The network used to evaluate $\hat{\Sigma}$ is the one depicted in \Cref{img:network}, describing the physical trend in the prediction year of 2050.

To perform the extreme event analysis, we generate $M = 500$ Monte Carlo realizations of the SST field $\{Y_{s_1}, \dots, Y_{s_n}\}$, enabling a probabilistic assessment of future climate scenarios and the identification of high-risk regions.

\subsection{Extreme event analysis}
\label{subsec:extreme_event}

As a first analysis, we evaluate joint exceedance probabilities \citep{extremeeventsHuser}. We investigate joint exceedance probabilities for the Elba Island region using RCP 4.5 SST projections through August 2050. Specifically, for a given temperature $t$, we estimate the probabilities of two distinct spatial events within a Euclidean neighborhood $D_0(r)$ of radius $r$ centered on the island: the union event, $\cup_{s \in D_0(r)} \{Y_s > t\}$, which happens when in at least one location in $D_0(r)$ the SST exceeds the threshold $t$; and the intersection event, $\cap_{s \in D_0(r)} \{Y_s > t\}$, which happens when the temperature in all locations in $D_0(r)$ exceed $t$. We examine radii $r \in \{0, 10, 15, 20, 30, 50\}$ km. Note that for $r=0$, the two probabilities collapse to the marginal exceedance probability.

A critical distinction underpins our analysis: while the region of interest $D_0(r)$ is defined geometrically using Euclidean distance, the probabilistic assessment relies on the flow-directed network structure. The Euclidean metric serves only to identify the monitoring area; but the joint probabilities are governed by the anisotropic influence of ocean currents. By incorporating directionality, our framework captures flow-driven connectivity that a purely Euclidean dependence model would miss, leading to a more realistic representation of spatial risk.

\Cref{img:jointexceedance} displays these probabilities for varying radii. The shaded area between the lower curve (union event) and the upper curve (intersection event) corresponds to the probability of partial exceedance—where extreme temperatures occur within the region but do not saturate it. The opacity of these bands increases with the radius $r$, visually encoding the effect of spatial extent. For the year 2050, the probability of exceedance remains notably high for thresholds up to $27^{\circ}C$, signaling a significant shift in the thermal baseline. Even at the extreme threshold of $29^{\circ}C$, the probability drops substantially but remains non-negligible.

\begin{figure}[t]
    \centering
    \begin{minipage}[b]{0.45\linewidth}
        \centering
        \includegraphics[width=0.6\linewidth]{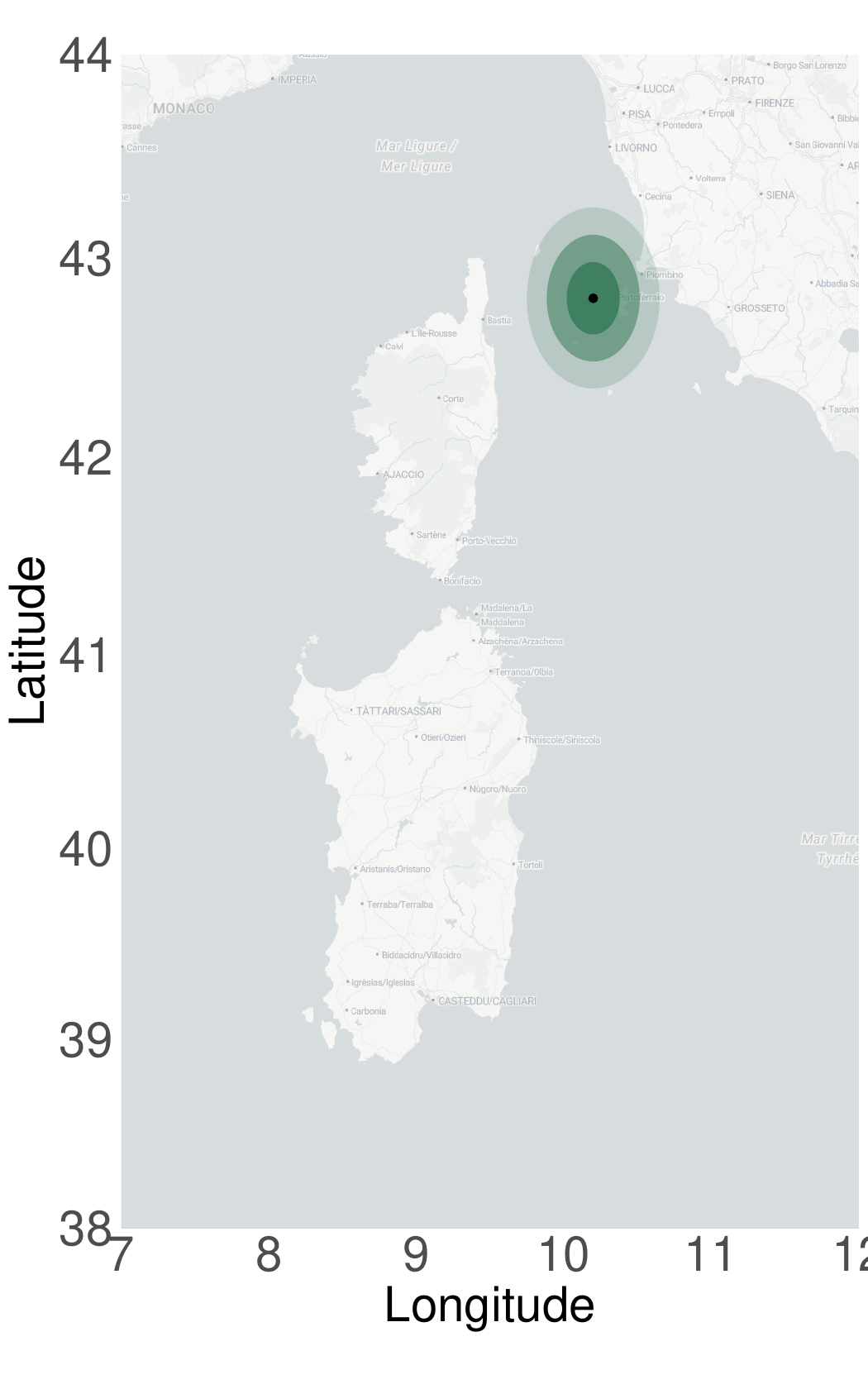}
        \vspace{0.2cm}
        \\ (a) Extension of the set $D_0(r)$.
    \end{minipage}
    \begin{minipage}[b]{0.45\linewidth}
        \centering
        \includegraphics[width=0.6\linewidth]{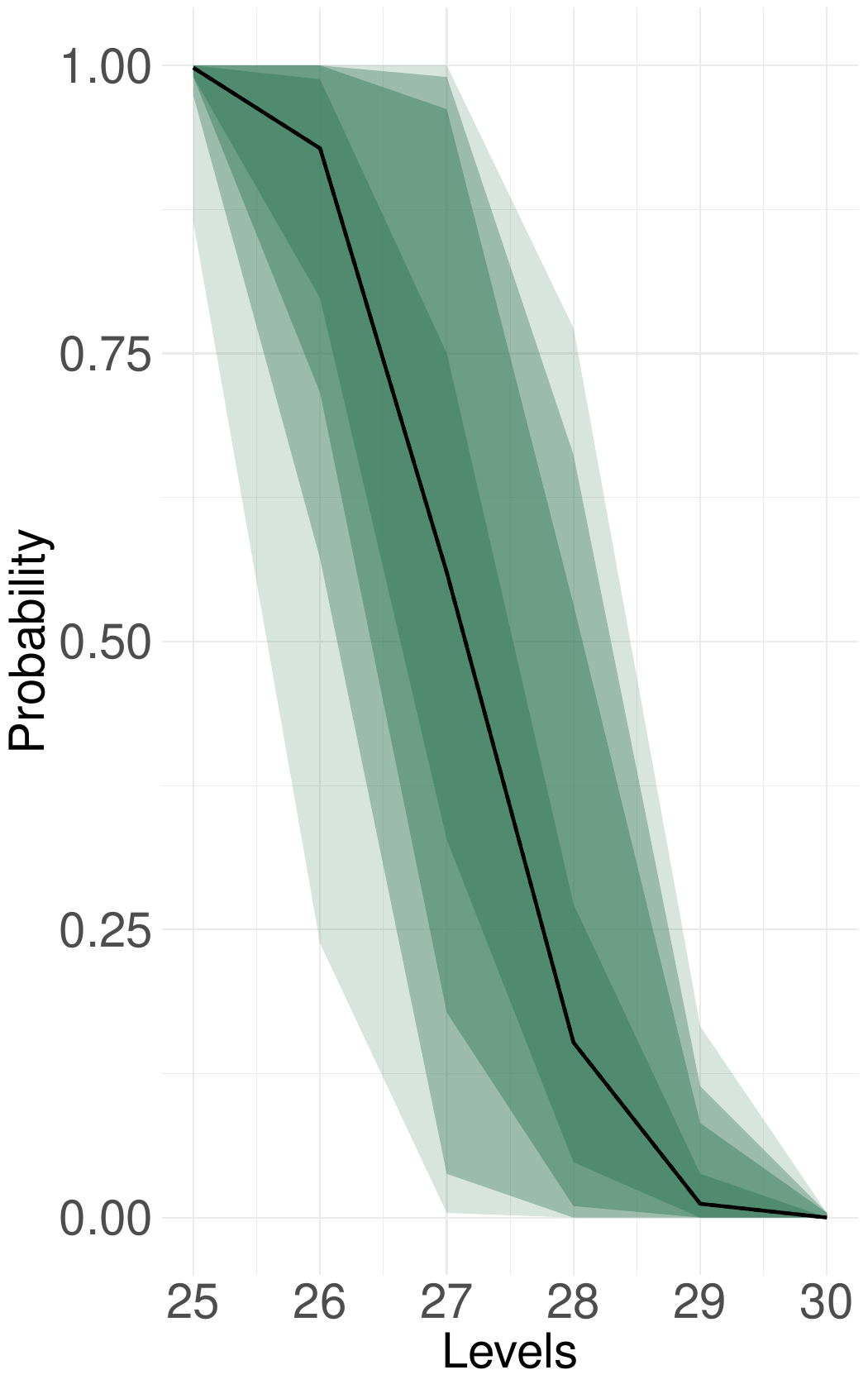}
        \vspace{0.2cm}
        \\ (b) Year 2050.
    \end{minipage}

    \vspace{0.5cm}

    \begin{tikzpicture}[y=0.5cm, x=0.5cm]
        \fill[white!0!col3] (1.8,0) rectangle (2.8,1);
        \node[below] at (2.3,-0.2) {10\,km};

        \fill[white!20!col3] (3.6,0) rectangle (4.6,1);
        \node[below] at (4.1,-0.2) {15\,km};

        \fill[white!40!col3] (5.4,0) rectangle (6.4,1);
        \node[below] at (5.9,-0.2) {20\,km};

        \fill[white!60!col3] (7.2,0) rectangle (8.2,1);
        \node[below] at (7.7,-0.2) {30\,km};

        \fill[white!80!col3] (9.0,0) rectangle (10.0,1);
        \node[below] at (9.5,-0.2) {50\,km};
    \end{tikzpicture}

    \caption{Joint exceedance probabilities for neighborhoods $D_0(r)$ centered on the Elba Island for varying radii $r$ (in km). For each radius, the shaded band is bounded below by the probability that \emph{all} vertices within $D_0(r)$ exceed the threshold, and bounded above by the probability that \emph{at least one} vertex exceeds it.}
    \label{img:jointexceedance}
\end{figure}
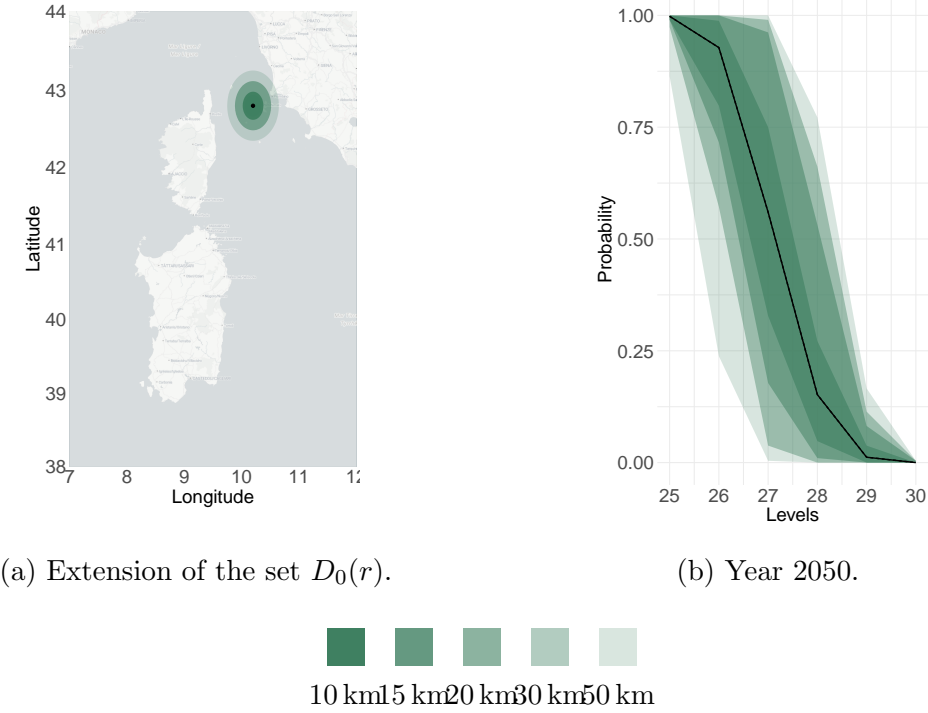

We further characterize extreme behavior by identifying regions where the random field exceeds a threshold $t$ with a specified probability $1 - \alpha$ \citep{ExtremeeventsBolin, extremeeventsHuser}. Formally, we study the random excursion set defined as $E_{t^+} = \{s \in D : Y_s > t\}$. Following the methodology of \cite{ExtremeeventsFrench} and \cite{extremeeventsHuser}, we construct two specific credible regions: the outer excursion set $S_{t^+}$, which contains the random set $E_{t^+}$ with high probability ($\mathbb{P}(E_{t^+} \subseteq S_{t^+}) = 1 - \alpha$); and the inner excursion set $S_{t^-}^C$, which is contained within $E_{t^+}$ with high probability ($\mathbb{P}(S_{t^-}^C \subseteq E_{t^+}) = 1 - \alpha$). In our analysis, we set the credibility level at $1-\alpha = 0.95$. \Cref{img:exceedance} displays these regions for the year 2050 and the threshold temperatures $t \in \{25^{\circ}C, 27^{\circ}C, 30^{\circ}C\}$, overlaid with the velocity field.

The analysis reveals distinct risk dynamics as the threshold increases. At $t=25^{\circ}C$, the inner set $S_{t^-}^C$ already encompasses specific areas in the southeastern domain—particularly east of Sardinia—indicating a high-confidence of exceedance. Conversely, the outer set $S_{t^+}$ covers the entire domain, implying that no location can be confidently excluded from the risk of surpassing $25^{\circ}C$. At the more extreme threshold of $30^{\circ}C$, the risk becomes highly localized. The outer set identifies potential critical areas in southern and eastern Sardinia, as well as the Tyrrhenian sector between Sardinia and mainland Italy (within the Italian EEZ).

These results provide a robust quantitative basis for risk assessment. The inner set $S_{t^-}^C$ demarcates zones of imminent hazard. Meanwhile, the outer set $S_{t^+}$ defines the maximum potential extent of the event, guiding the spatial allocation of monitoring resources to ensure that no exceedance goes undetected.

The bottom row compares our results with the classical Euclidean framework. Specifically, we apply the same yearly analysis on the empirical residuals from 2006-2022. The covariance parameters (sill and range) are estimated by fitting a parametric model to the point-wise mean of the empirical variograms. The two frameworks yield notably different hot-spot estimates, particularly at the $25^{\circ}C$ threshold. By incorporating velocity field data, the proposed framework narrows the extent of the high-risk hot-spot in southeastern Sardinia—a direct result of the pronounced currents in that area. Additionally, at the $30^{\circ}$ threshold, the spatial pattern of the low-risk hot-spot along the eastern coast of Corsica differs visibly between the two frameworks.

\begin{figure}
  \centering
  
  \begin{minipage}[b]{0.32\linewidth}
    \centering
    \includegraphics[width=0.7\linewidth]{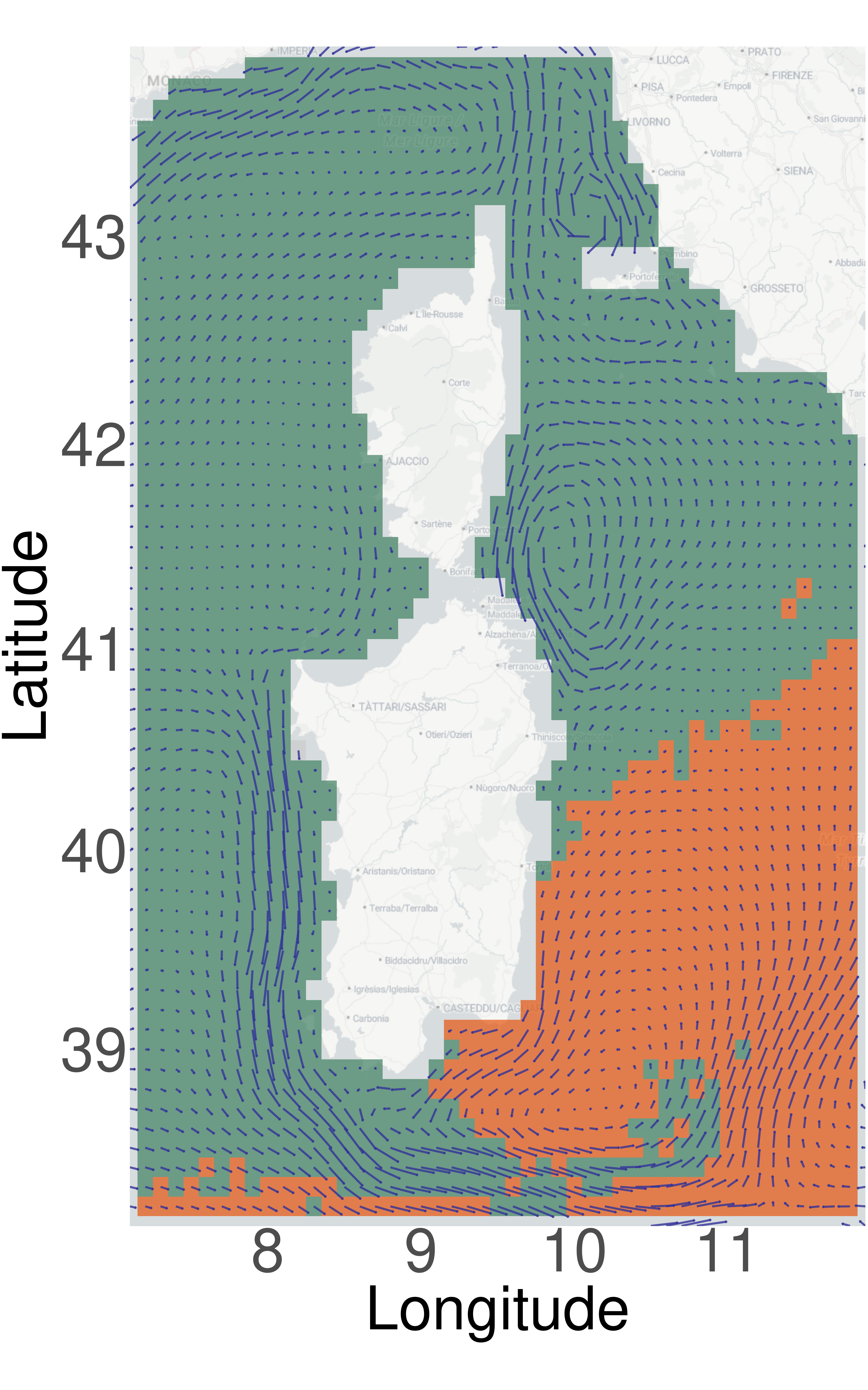}
    \vspace{0.2cm}
    \\ (a) $t = 25^{\circ}$
  \end{minipage}
  \hfill
  \begin{minipage}[b]{0.32\linewidth}
    \centering
    \includegraphics[width=0.7\linewidth]{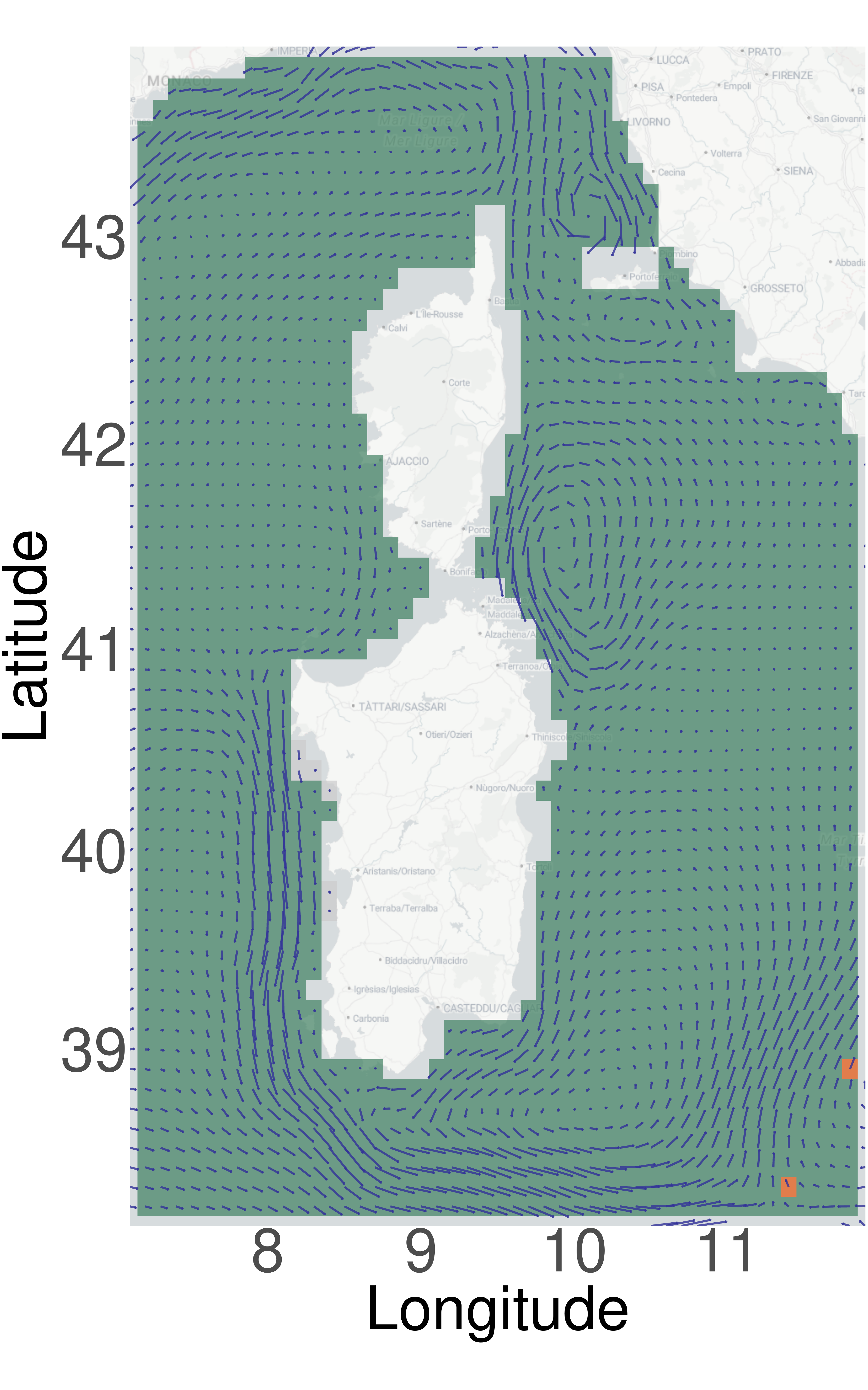}
    \vspace{0.2cm}
    \\ (b) $t = 27^{\circ}$
  \end{minipage}
  \hfill
  \begin{minipage}[b]{0.32\linewidth}
    \centering
    \includegraphics[width=0.7\linewidth]{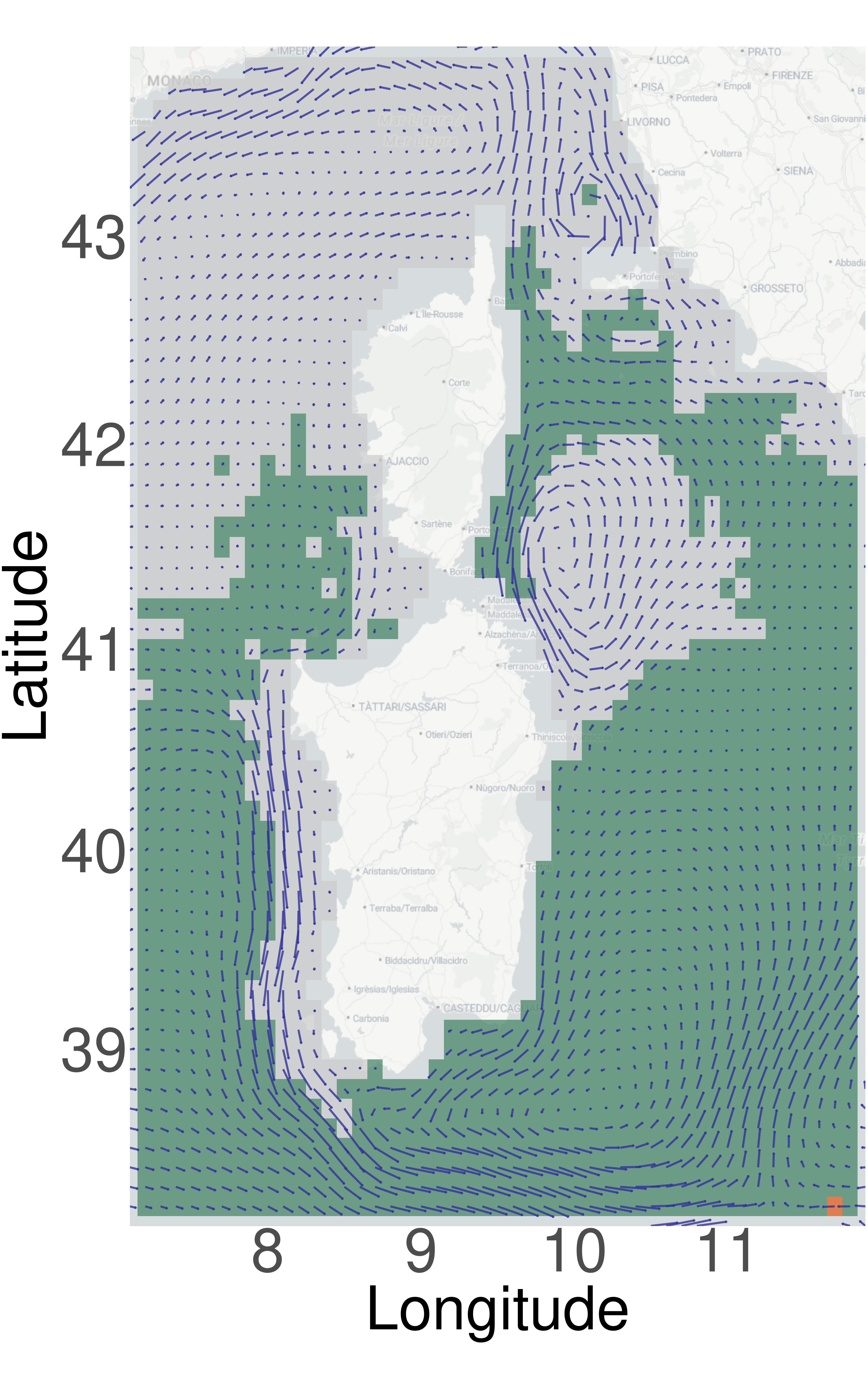}
    \vspace{0.2cm}
    \\ (c) $t = 30^{\circ}$
  \end{minipage}

  \vspace{0.5cm}

  \begin{minipage}[b]{0.32\linewidth}
    \centering
    \includegraphics[width=0.7\linewidth]{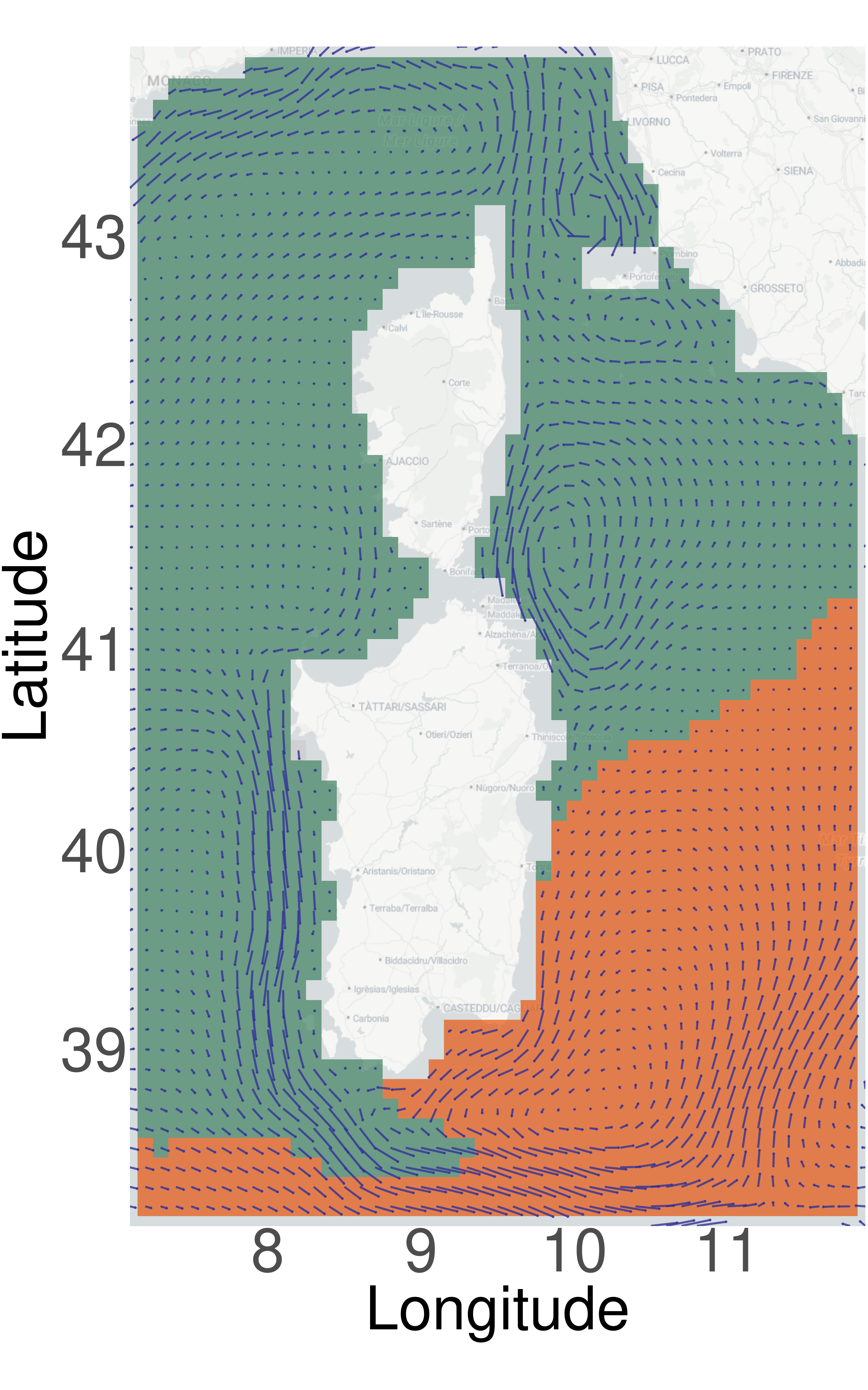}
    \vspace{0.2cm}
    \\ (d) $t = 25^{\circ}$ (Euclidean)
  \end{minipage}
  \hfill
  \begin{minipage}[b]{0.32\linewidth}
    \centering
    \includegraphics[width=0.7\linewidth]{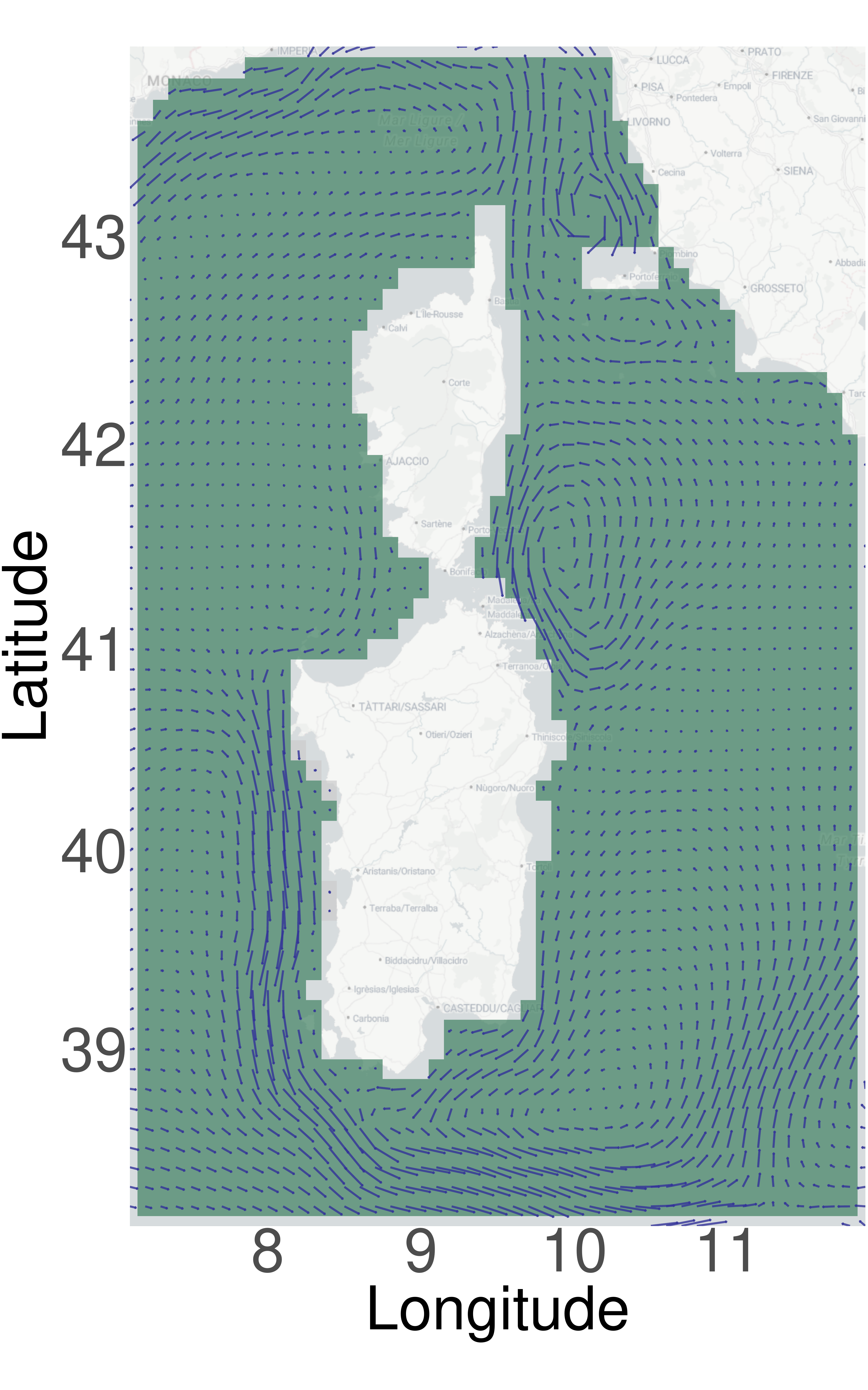}
    \vspace{0.2cm}
    \\ (e) $t = 27^{\circ}$ (Euclidean)
  \end{minipage}
  \hfill
  \begin{minipage}[b]{0.32\linewidth}
    \centering
    \includegraphics[width=0.7\linewidth]{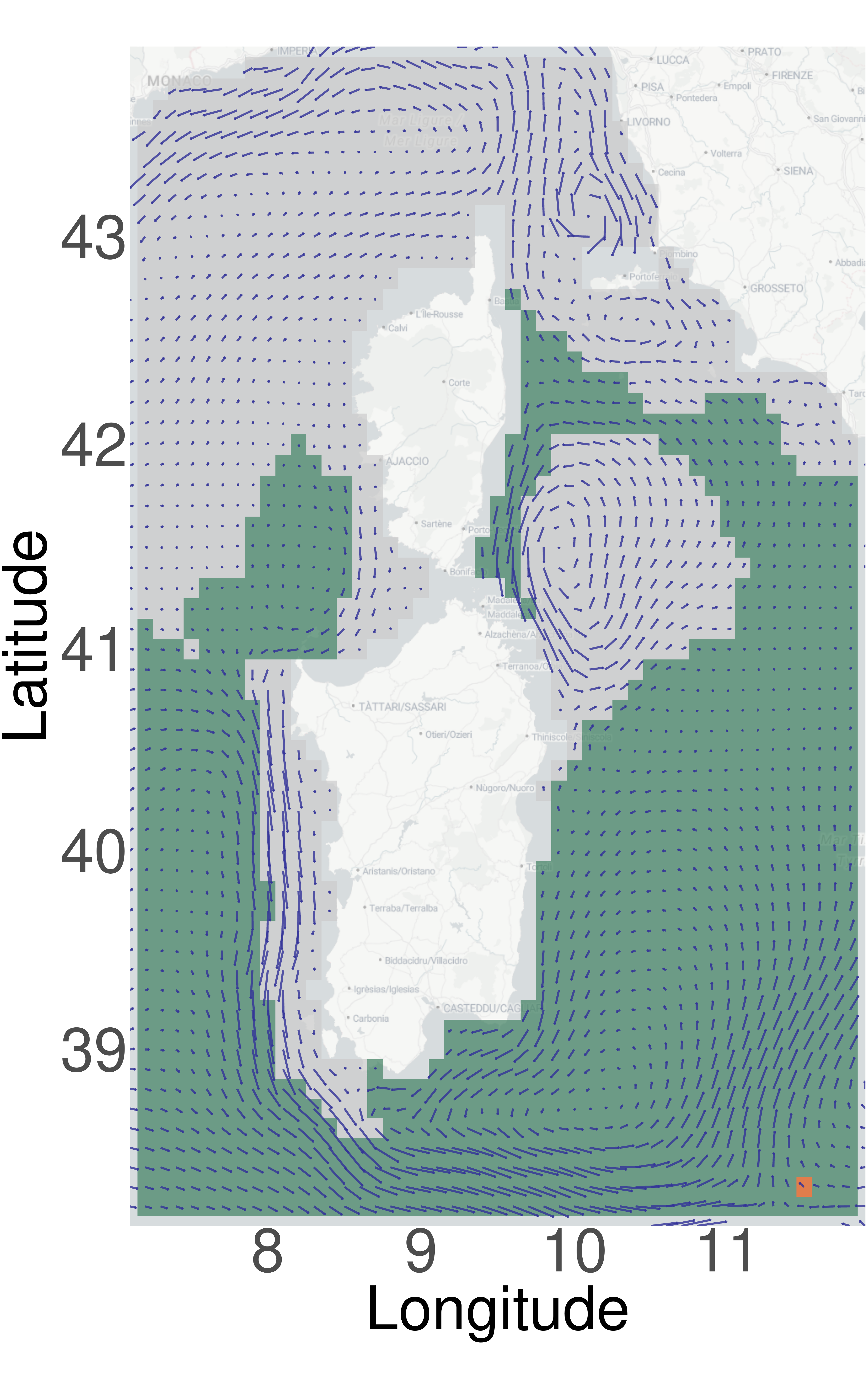}
    \vspace{0.2cm}
    \\ (f) $t = 30^{\circ}$ (Euclidean)
  \end{minipage}

  \vspace{0.5cm}

  \begin{tikzpicture}
    \fill[col3] (0,0) rectangle (0.4,0.4);
    \node[right] at (0.5,0.2) {$S_{t^+}$};

    \fill[col2] (4,0) rectangle (4.4,0.4);
    \node[right] at (4.5,0.2) {$S_{t^-}^C$};
  \end{tikzpicture}
  
  \caption{Hotspot estimates for RCP 4.5 -- Year 2050. Top row: New framework. Bottom row: Euclidean Framework.}
  \label{img:exceedance}
\end{figure}

%-----------------------------------------------------------------------------
% CONCLUSION
%-----------------------------------------------------------------------------

\section{Conclusions and discussion}

This work establishes a rigorous mathematical framework for modeling spatial processes governed by intrinsic directionality. By constructing valid covariance models on directed linear networks, we address the fundamental challenge of capturing flow-driven dependence. Central to our methodology is the integration of a Markovian dynamic with a specialized weighting procedure, which ensures both the positive definiteness of the covariance function and the stationarity of the process. Furthermore, the development of a penalized estimator tailored to this weighted setting resolves the instabilities inherent in network-based estimation.

The practical power of this framework is demonstrated through the analysis of projected water temperatures using the ERSEM model. By encoding ocean currents into the network topology, we derived a flow-informed covariance structure. This approach enabled a physics-consistent simulation of temperature fields, providing a probabilistic assessment of risk areas in the Mediterranean.

Beyond marine applications, the proposed framework offers a general solution for systems where topology and transport shape spatial dependence, such as river networks, urban traffic flows, and wind corridors. Its adaptability makes it a versatile tool for diverse domains requiring non-Euclidean geostatistical modeling.

Looking forward, several research avenues emerge to extend the flexibility and scope of the proposed framework.

A primary theoretical objective is to relax the transience assumption to accommodate self-contained (recurrent) networks. Generalizing the moving average rationale to closed systems—where mass is preserved rather than dissipated—would significantly broaden the model's applicability. Crucially, this extension would lay the groundwork for bridging our graph-based approach with state-of-the-art techniques based on Stochastic Partial Differential Equations (SPDEs) \citep{spdeLindgren} and spatial regression on manifolds \citep{Azzimonti_Sangalli_Secchi_Domanin_Nobile_2015}, effectively linking discrete network processing with continuous domain geostatistics.

In parallel, enhancing the physical fidelity of the model is a key priority. This includes incorporating a temporal dimension to model spatiotemporal evolution and exploring non-Markovian dynamics. Moving beyond the Markovian dynamic would allow for the representation of more complex physical phenomena, such as inertia or momentum in flow-driven systems, thereby offering a more realistic approximation of environmental processes.

Finally, the framework is well-positioned to integrate with Object-Oriented Spatial Statistics \citep{objectoriented}. Extending the domain from point-referenced observations to complex spatial objects—such as trajectories, curves, or sub-regions moving along the network—would open new frontiers for ecological and environmental studies involving structured data.

%% file: appendix.tex
\setcounter{theorem}{0}
\setcounter{lemma}{0}
\setcounter{proposition}{0}
\setcounter{equation}{0}
\setcounter{figure}{0}
\setcounter{algocf}{0}

\renewcommand{\thetheorem}{A\arabic{theorem}}
\renewcommand{\thelemma}{A\arabic{lemma}}
\renewcommand{\theproposition}{\arabic{proposition}}
\renewcommand{\thecorollary}{A\arabic{corollary}}
\renewcommand{\theequation}{A\arabic{equation}}
\renewcommand{\thefigure}{A\arabic{figure}}
\renewcommand{\thealgocf}{A\arabic{algocf}}

\newcommand{\cA}{\mathcal{A}}
\newcommand{\cC}{\mathcal{C}}
\newcommand{\cD}{\mathcal{D}}
\newcommand{\cF}{\mathcal{F}}
\newcommand{\cI}{\mathcal{I}}
\newcommand{\hDelta}{\hat{\Delta}}

\section{Proofs of the propositions}

\begin{proposition}
    The process $\{Z_x , x \in V\}$ is zero-mean: for $x \in V$,
    $
    \mathbb{E}\left[Z_x\right] = 0.
    $
\end{proposition}

\begin{proof}
By the tower property and independence between the Wiener process $W$ and the random variables $T$, we have
\begin{align*}
\mathbb{E}[Z_x]
&= \mathbb{E}\left[\sum_{v\in V}\int_{L_{[v,X_1]}}
\frac{g\big(dist_{T(v,x)}(u,x)\big)}{\sqrt{\beta_{T(v,x)}}}\, W(du)\right] \\
&= \sum_{v\in V}\mathbb{E}\left[\,
\mathbb{E}\left[\int_{L_{[v,X_1]}}
\frac{g\big(dist_{T(v,x)}(u,x)\big)}{\sqrt{\beta_{T(v,x)}}}\, W(du) \;\Bigg|\; T(v,x)\right]\right] \\
&= \sum_{v\in V}\mathbb{E}\left[\,
\int_{L_{[v,X_1]}}
\frac{g\big(dist_{T(v,x)}(u,x)\big)}{\sqrt{\beta_{T(v,x)}}}\, \mathbb{E}\left[ W(du) \;\Bigg|\; T(v,x)\right]\right]\\
&= \sum_{v\in V}\mathbb{E}\left[\,
\int_{L_{[v,X_1]}}
\frac{g\big(dist_{T(v,x)}(u,x)\big)}{\sqrt{\beta_{T(v,x)}}}\, \mathbb{E}\left[ W(du) \right]\right] \;=\; 0,
\end{align*}
since the Wiener process has zero mean.

\end{proof}

\begin{proposition}
    For $x, y \in V$,
    \begin{equation}
        Cov(Z_x, Z_y) = \sum_{v \in V} \mathbb{E}\left[ \int_{L_{\left[v,X_1\right]}}
        \frac{g\left(dist_{T(v,x)}\left(u,x\right)\right)
              g\left(dist_{T(v,y)}\left(u,y\right)\right)}
             {\sqrt{\beta_{T(v,x)}\,\beta_{T(v,y)}}} \, du \right].
    \end{equation}
    In particular,
    \begin{equation}
        Var(Z_x) = Cov(Z_x, Z_x) = \sum_{v \in V} \mathbb{E}\left[
        \int_{L_{\left[v,X_1\right]}}
        \frac{g\left(dist_{T(v,x)}\left(u,x\right)\right)^2}{\beta_{T(v,x)}} \, du
        \right].
    \end{equation}
\end{proposition}

\begin{proof}

We condition on the joint realization of all random paths $\{T(v,x),\, T(v',y)\}_{v,v' \in V}$, which are independent of $W$. By the tower property,
\[
    \mathbb{E}[Z_x Z_y] = \mathbb{E}\!\left[\,
    \mathbb{E}\!\left[Z_x Z_y \;\middle|\;
    \{T(v,x),\,T(v',y)\}_{v,v' \in V}\right]\right].
\]

Expanding the product $Z_x Z_y$ yields a double sum over
$v, v' \in V$:

\begin{align*}
    \mathbb{E}\!\left[Z_x Z_y \;\middle|\; \{T(v,x),\,T(v',y)\}_{v,v' \in V}\right]&= \sum_{v \in V}\sum_{v' \in V}
    \mathbb{E}\!\left[
    \int_{L_{[v,X_1]}} \frac{g\!\left(dist_{T_{(v,x)}(u,x)}\right)}{\sqrt{\beta_{T_{(v,x)}}}} W(du)\right.\\ 
    & \hspace{-2 cm} \left.\int_{L_{[v',X'_1]}} \frac{g\!\left(dist_{T_{(v',y)}}(u',y)\right)}
                        {\sqrt{\beta_{T_{(v',y)}}}} W(du')
    \;\middle|\; \{T(v,x),\,T(v',y)\}_{v,v' \in V}
    \right].
\end{align*}

We now apply the white-noise isometry. Recall that for two square-integrable
deterministic functions $f_1, f_2$ and two edges $l_1, l_2$,
\[
    \mathbb{E}\!\left[\int_{l_1} f_1(u)\, W(du)\int_{l_2} f_2(u)\, W(du)\right]
    = \begin{cases} \displaystyle\int_{l} f_1(u)\,f_2(u)\,du & \text{if } l_1 = l_2 =: l,\\[6pt]
    0 & \text{if } l_1 \cap l_2 = \emptyset. \end{cases}
\]
Since $L_{[v,X_1]}$ and $L_{[v',X'_1]}$ are the first edges traversed by the Markov chains started at $v$ and $v'$ respectively, if $v \neq v'$  they are distinct edges of the network. Therefore, all cross terms with $v \neq v'$ vanish, and the double sum reduces to

\begin{align*}
    \mathbb{E}\!\left[Z_x Z_y \middle| \{T(v,x),T(v',y)\}_{v,v' \in V}\right] = \sum_{v \in V}
    \int_{L_{[v,X_1]}}\!
    \frac{g\!\left(dist_{T(v,x)}(u,x)\right)\,
          g\!\left(dist_{T(v,y)}(u,y)\right)}
         {\sqrt{\beta_{T(v,x)}\,\beta_{T(v,y)}}}du\\
\end{align*}
Taking the outer expectation with respect to the \(T\)'s yields \eqref{eq:moment_second}. Setting
$y = x$ gives \eqref{eq:variance_moment}.
\end{proof}

\begin{proposition}
Let $(\mathcal{L}, V)$ be a linear network, $\{Z_x, x \in V\}$ be the process defined in Equation (2) of the main text, and the normalization constants $\beta$'s be defined as in Equation (5) of the main text. Then for every $x\in V$ the marginal variance is constant and equals

\begin{equation}
    Var(Z_x) \;=\; \int_0^{+\infty} g^2(r)\,dr.
\end{equation}
\end{proposition}

\begin{proof}
Consistent with the framework established by \cite{VerHoef2006, rivers2010}, we address the unbounded domain of the moving average by introducing a virtual upstream node acting as a generic boundary condition. This ensures total probability mass preservation for the source nodes. Pushing the upstream boundary to infinity ensures that the influence of any specific external condition vanishes, effectively decoupling the internal covariance structure from unobserved upstream inputs.

The first edge of $T(v,x)$ is indicated with $L(v,X_1);$ if $T(v,x)=p(v,x),$ its first edge is therefore $l(v,x_1)$. By \Cref{prop:moments},
\begin{align}
Var(Z_x)
&= \sum_{v\in V}\,\mathbb{E}\!\left[\int_{L_{[v,X_1]}}
\frac{g\left(dist_{T(v,x)}(u,x)\right)^2}{\beta_{T(v,x)}}\,du\right]\nonumber\\
&= \sum_{v\in V}\sum_{p(v,x)\in\mathcal{P}(v,x)}
\mathbb{P}\left(T(v,x)=p(v,x)\right)\,
\frac{1}{\beta_{p(v,x)}}\int_{l_{[v,x_1]}} g\left(|p(u,x)|\right)^2\,du.
\label{eq:var-sum}
\end{align}
Along the edge $l_{[v,x_1]}$, the map $u\mapsto r:=|p(u,x)|$ is the arc-length parameter, so $dr=du$ while the integration range is $[\,|p(x_1,x)|,\,|p(v,x)|\,]$. Swapping summation and integral, one obtains:
\[
Var(Z_x)= \int_{0}^{\infty} g(r)^2\, S_x(r)\,dr,
\]
where the auxiliary function $S_x(r)$ is defined as:
\[
S_x(r)=\sum_{v \in V}\sum_{p(v,x)\in\mathcal{P}(v,x)}
\frac{\mathbb{P}(T(v,x)=p(v,x))}{\beta_{p(v,x)}}\,
\mathbb{I}_{\{|p(x_1,x)|\le r < |p(v,x)|\}}.
\tag{$\star$}\label{eq:Sx-def}
\]
We claim that $S_x(r)\equiv 1$ for all $r \in \mathbb{R}$. Note that $S_x(r)$ is a piecewise constant function by construction; discontinuities are possible only at values of $r$ corresponding to the length of a path ending in $x$. We demonstrate that the function is constant by analyzing the conservation of probability mass across a generic branching point, as illustrated in \Cref{img:schema}.

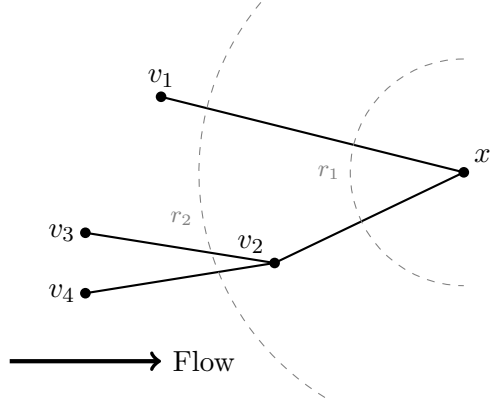
\begin{figure}[h]
    \centering
    \vspace{0.5 cm}
    \begin{tikzpicture}[scale=1]
        \coordinate (Join) at (1,0);
        \coordinate (V2) at (-1.5, -1.2);
        \coordinate (V1) at (-3, 1);
        \coordinate (SplitA) at (-4, -0.8);
        \coordinate (SplitB) at (-4, -1.6);
        \draw[thick,-] (V1) -- (Join);
        \draw[thick,-] (V2) -- (Join);
        \draw[thick,-] (SplitA) -- (V2);
        \draw[thick,-] (SplitB) -- (V2);
        \fill (Join) circle (2pt) node[above right]{$x$};
        \fill (V1) circle (2pt) node[above]{$v_1$};

        \fill (V2) circle (2pt) node[above left]{$v_2$};

        \fill (SplitA) circle (2pt) node[left]{$v_3$};
        \fill (SplitB) circle (2pt) node[left]{$v_4$};

        \draw[ultra thick, ->] (-5, -2.5) -- (-3, -2.5) node[right] {Flow};

        \draw[dashed, gray, thin] (Join) ++(0, 1.5) arc (90:270:1.5) node[midway, left, font=\footnotesize] {$r_1$};
        \draw[dashed, gray, thin] (Join) ++(140:3.5) arc (140:240:3.5) node[midway, left, font=\footnotesize] {$r_2$};

    \end{tikzpicture}
    \caption{Schematic representation of the network for the proof of mass conservation.}
    \label{img:schema}
\end{figure}

Consider the first radius $r_1 \leq |l_{\left[v_2,x\right]}|$ shown in the figure. The only paths for which the indicator function is non-zero are $p(v_1,x)$ and $p(v_2,x)$. In this scenario, the normalization constants \(\beta\)'s are given, up to the multiplying probability $U(x)$, by 
\[
\beta_{p(v_1,x)} = \beta_{p(v_2,x)} \propto \pi_{[v_1,x]} + \pi_{[v_2,x]} \coloneqq A.
\]

Thus, the non-zero probabilities and normalization constants are:

\begin{itemize}
    \item \(\mathbb{P}(p(v_1, x)) = \mathbb{P}(T(v_1,x)=p(v_1, x)) = U(x) \, \cdot \, \pi_{\left[v_1,x\right]}\),
    \item \(\mathbb{P}(p(v_2, x)) =  \mathbb{P}(T(v_1,x)=p(v_2, x)) = U(x) \, \cdot \, \pi_{\left[v_2,x\right]}\),
    \item \(\beta_{p(v_1,x)} = \beta_{p(v_2,x)} = U(x) \, \cdot \, A\).
\end{itemize}

The expression for $S_x(r_1)$ becomes:
\[
S_x(r_1) = \frac{\pi_{[v_1,x]}}{A} + \frac{\pi_{[v_2,x]}}{A} = \frac{\pi_{[v_1,x]} + \pi_{[v_2,x]}}{\pi_{[v_1,x]} + \pi_{[v_2,x]}} = 1.
\]

We now show that $S_x(r)$ remains equal to $1$ when moving to a radius $r_2 > |l_{\left[v_2,x\right]}|$, hence crossing the junction $v_2$. The path $p(v_2,x)$ is no longer included in the sum, but is replaced by the upstream paths originating from $v_3$ and $v_4$. The paths contributing to $S_x(r_2)$ are $p(v_1,x), \, p(v_3,x), \, p(v_4,x)$. Let $B = (\pi_{[v_3,v_2]} + \pi_{[v_4,v_2]})$. The relevant probabilities and normalization constants \(\beta\)'s are:

\begin{itemize}
    \item \(\mathbb{P}(p(v_3, x)) =  \mathbb{P}(T(v_3,x)=p(v_3, x)) = U(x) \, \cdot \, \pi_{[v_3,v_2]} \, \pi_{[v_2,x]}\),
    \item \(\mathbb{P}(p(v_4, x)) =  \mathbb{P}(T(v_4,x)=p(v_4, x)) = U(x) \, \cdot \, \pi_{[v_4,v_2]} \, \pi_{[v_2,x]}\),
    \item \(\beta_{p(v_3,x)} = \beta_{p(v_4,x)} = U(x) \, \cdot \, A \cdot B\).
\end{itemize}

Substituting these into the expression for $S_x(r_2)$, one obtains:
\begin{align*}
S_x(r_2) &= \frac{\mathbb{P}(p(v_1,x))}{\beta_{p(v_1,x)}} + \frac{\mathbb{P}(p(v_3,x))}{\beta_{p(v_3,x)}} + \frac{\mathbb{P}(p(v_4,x))}{\beta_{p(v_4,x)}}\\
&= \frac{\pi_{[v_1,x]}}{A} + \frac{\pi_{[v_3,v_2]} \, \pi_{[v_2,x]}}{A \cdot B} + \frac{\pi_{[v_4,v_2]} \, \pi_{[v_2,x]}}{A \cdot B}\\
&= \frac{\pi_{[v_1,x]}}{A} + \frac{\pi_{[v_2,x]}}{A} \, \left(\frac{\pi_{[v_3,v_2]} + \pi_{[v_4,v_2]}}{B}\right)\\
&=  \frac{\pi_{[v_1,x]}}{A} + \frac{\pi_{[v_2,x]}}{A}\\
&=1.
\end{align*}

This conservation of mass applies to any bifurcation. Hence, by induction on $r$, $S_x$ is a constant function equal to 1,
completing the proof.
\end{proof}

\begin{proposition}
Let $(\mathcal{L}, V)$ be a linear network, $\{Z_x, x \in V\}$ be the process defined in Equation (2) of the main text, and the normalization constants $\beta$'s be defined as in Equation (5) of the main text. Let \(x,y \in V\), \(x \neq y.\) 

\begin{itemize}
\item[(i)] If $\, \, \tilde{\mathcal{P}}(x,y) \bigcup \tilde{\mathcal{P}}(y,x) \neq \emptyset$, 
    \begin{equation*}
        Cov(Z_x, Z_y) = \sum_{p \in \tilde{\mathcal{P}}(x,y) \bigcup \tilde{\mathcal{P}}(y,x)} w_{p} C(|p|),
    \end{equation*}
where, for any two vertices \(v_1,v_2 \in V,\)
    \begin{equation}
    w_{p{(v_1,v_2)}} = \left(\prod_{l_{\left[a,b\right]} \in p{(v_1,v_2)}} \frac{\pi_{\left[a,b\right]}}{\sqrt{\sum_{k} \pi_{\left[k,b\right]}}}\right) \frac{U(v_2,v_1)}{\sqrt{U(v_1)\, U(v_2)}},
    \end{equation}
    and $C(\cdot)$ is the spatial covariance function defined, for $h \geq 0,$ as $$C(h) = \int_0^{+\infty} g(r)g(r + h) dr.$$    
\item[(ii)]
If $\, \, \tilde{\mathcal{P}}(x,y) \bigcup \tilde{\mathcal{P}}(y,x) = \emptyset$,
\begin{equation*}
        Cov(Z_x, Z_y) = 0.
    \end{equation*}
\end{itemize}    

\end{proposition}

\begin{proof}

As before, let $L(v,X_1)$ be the first edge of $T(v,x),$ and \(l(v,x_1)\) be the first edge of the realization $T(v,x)=p(v,x)$. Note that two simultaneous realizations $T(v,x)=p(v,x)$ and $T(v,y)=p(v,y)$ must share the first edge \(l(v,x_1)\).

By \Cref{prop:moments},

\begin{align*}
Cov(Z_x,Z_y)\\
&= \sum_{v \in V} \sum_{p(v,x) \in \mathcal{P}(v,x)} \sum_{p(v,y) \in \mathcal{P}(v,y)} \frac{\mathbb{P}(T(v,x) = p(v,x), T(v,y) = p(v,y))}{\sqrt{\beta_{p(v,x)} \, \beta_{p(v,y)}}}\\
&\int_{l_{[v,x_1]}} g(|p(u,x)|) \, g(|p(u,y)|) du.
\end{align*}

To simplify notation, for any path $p$, let $\Pi(p) = \prod_{l \in p} \pi_l$ and $B(p) = \prod_{l_{[a,b]} \in p} \sum_k \pi_{[k,b]}$, so that $\mathbb{P}(p) = \Pi(p)\,U(x)$ and $\beta_p = B(p)\,U(x)$, where $x$ is the terminal node of $p$.

We now turn to (i). Given the initial state $v$ of the Markov chain, $L_x$ and $L_y$ are both finite with probability one, because $x$ and $y$ are transient. Moreover they are different, since $x\not =y$. In the representation of $Cov(Z_x,Z_y)$, we partition each contribution to the external sum according to the order of the last visits: $I_v(x,y)$ will denote the contribution from realizations such that $L_x <L_y$, and $I_v(y,x)$ the contribution from realizations such that $L_x > L_y$.

Let's focus on $I_v(x,y)$. Since $L_x < L_y,$ the trajectory $T(v,y)$ can be decomposed as the concatenation of $T(v,x)$ and $\tilde{T}(x,y) = (l_{[X_{L_x},X_{L_x+1}]}, \dots, l_{[X_{L_y-1}, X_{L_y}]})$, that is the trajectory of the chain from the last visit to $x$ to the last visit to $y$. Note that $\tilde{T}(x,y)$ cannot revisit $x$; hence $\tilde{T}(x,y) \in \tilde{\mathcal{P}}(x,y).$ Moreover:
\begin{align*}
    &\mathbb{P}\!\left(T(v,x) = p(v,x), T(v,y) = p(v,y)\right)\\
    &=\mathbb{P}\!\left(T(v,x) = p{(v,x)},\; \tilde{T}(x,y) = p{(x,y)}\right)
    = \Pi(p{(v,x)})\cdot\Pi(p{(x,y)})\cdot U(y,x),
    \label{eq:jointprob}
\end{align*}
where $U(y,x)$ is the probability of never returning to either $x$ or $y$ after reaching $y$, as defined in the main text. The path $p(x,y) \in \tilde{\mathcal{P}}(x,y)$ is the path realizing $\tilde{T}(x,y),$ consistent with $p(v,x)$ and $p(v,y)$.  Conversely, any \(p(x,y) \in \tilde{\mathcal{P}}(x,y),\) and any \( p(v,x) \in \mathcal{P}(v,x)\) identify a \(p(v,y) \in \mathcal{P}(v,y).\)

Now note that
\(B(p{(v,y)}) = B(p{(v,x)})\,B(p{(x,y)})\), and therefore
\[
    \beta_{p{(v,x)}}\,\beta_{p{(v,y)}}
    = B(p{(v,x)})\,U(x)\cdot B(p{(v,y)})\,U(y)
    = B(p{(v,x)})^2\,B(p{(x,y)})\,U(x)\,U(y).
\]

We now turn to the integral over $l_{[v,x_1]}$ appearing in the term $I_v(x,y)$. Setting $h_p = |p{(x,y)}|$ and $r = |p(v,x)|$, we have \(|p(v,y)| = r + h_p\).

Summing all up, we obtain
\begin{align*}
    I_v(x,y) = 
    &\sum_{p{(v,x)} \in \mathcal{P}(v,x)} \sum_{p{(x,y)} \in \tilde{\mathcal{P}}(x,y)} \frac{\Pi(p{(v,x)})\,\Pi(p{(x,y)})\,U(y,x)}{B(p{(v,x)}) \, \sqrt{U(x)\,U(y)\,B(p{(x,y)})}}\\
    &\int_{|p{(x_1,x)}|}^{|p{(v,x)}|}
    g(r)\,g(r+h_p)\,dr.
\end{align*}

Rearranging, we isolate the dependence on $p{(x,y)}$ to factor out the weight $w_{p{(x,y)}}$ as defined in \eqref{eq:weightexplicit}, 

\begin{align*}
    I_v(x,y)\\
    &= \sum_{p{(x,y)} \in \tilde{\mathcal{P}}(x,y)}\frac{\Pi(p{(x,y)})\,U(y,x)}{\sqrt{U(x)\,U(y)\,B(p{(x,y)})}} \\ &
    \int_{0}^{\infty} 
    \sum_{p{(v,x)} \in \mathcal{P}(v,x)} \frac{\Pi(p{(v,x)})}{B(p{(v,x)})} \mathbb{I}_{\{|p{(x_1,x)}| \leq r < |p{(v,x)}|\}} g(r)\,g(r+h_p) \,dr\\
    &= \sum_{p{(x,y)} \in \tilde{\mathcal{P}}(x,y)} w_{p{(x,y)}} \\
    &\int_{0}^{\infty} \sum_{p{(v,x)} \in \mathcal{P}(x,y)} \frac{\Pi(p{(v,x)})}{B(p{(v,x)})} \mathbb{I}_{\{|p{(x_1,x)}| \leq r < |p{(v,x)}|\}} g(r)\,g(r+h_p) \,dr.
\end{align*}

Therefore, 
\begin{align*}
\sum_{v \in V} I_v(x,y)\\
&= \sum_{p{(x,y)} \in \tilde{\mathcal{P}}(x,y)}w_{p(x,y)}\\
    &\int_{0}^{\infty} 
    \sum_{v \in V}\sum_{p{(v,x)} \in \mathcal{P}(x,y)} \frac{\Pi(p{(v,x)})}{B(p{(v,x)})} \mathbb{I}_{\{|p{(x_1,x)}| \leq r < |p{(v,x)}|\}} g(r)\,g(r+h_p) \,dr\\
   &=\sum_{p{(x,y)} \in \tilde{\mathcal{P}}(x,y)} w_{p{(x,y)}}
   \int_0^{+\infty} S_x(r)\,g(r)\,g(r+h_p)\,dr,
\end{align*}

where, as in the proof of \Cref{prop:variancestationary},
\[
    S_x(r) = \sum_{v \in V}\sum_{p{(v,x)}}
    \frac{\Pi(p{(v,x)})}{B(p{(v,x)})}
    \,\mathbb{I}_{\{|p{(b_1,x)}| \leq r < |p{(v,x)}|\}}.
\]
However, within the proof of \Cref{prop:variancestationary}, we already proved that $S_x(r) \equiv 1$ for all $r \geq 0$, so the integral reduces to $C(h_p) = \int_0^{+\infty} g(r)\,g(r+h_p)\,dr$. 

The same argument applies symmetrically to $I_v(y,x)$, with the last sum running over $\tilde{\mathcal{P}}(y,x)$. 
Summing $\sum_{v\in V}I_v(x,y)$ to $\sum_{v \in V}I_v(y,x)$ completes the proof of part (i).

To prove (ii), note that if $\tilde{\mathcal{P}}(x,y) \bigcup \tilde{\mathcal{P}}(y,x) = \emptyset,$ then $x$ and $y$ are not connected on the network \(\mathcal{L}.\) Hence $Cov(Z_x,Z_y)=0$. 

\end{proof}

\section{Data Management and Algorithm}
\label{sec:dataandalgo}

\subsection{Data Preprocessing and Network Construction}
\label{subsec:preprocessing_construction}

As detailed in Section 2 of the main text, the analysis relies on sea surface temperature and velocity fields from the CMEMS and C3S databases. The spatial domain is restricted to the bounding box $38^\circ\text{N}$--$44^\circ\text{N}$ and $7^\circ\text{E}$--$12^\circ\text{E}$. To ensure physical consistency when computing the Euclidean distance matrix and paths' lengths for the covariance models, the original geographic coordinates (WGS84) are projected onto the ETRS89-extended / LAEA Europe coordinate reference system (EPSG:3035). A land-mask is applied by pruning any edges connected to vertices with missing oceanographic data, strictly confining the network to valid water regions.

The network topology and transition probabilities are constructed by decomposing the observed velocity field onto the regular grid. For a given vertex $a$ with a valid velocity vector $\boldsymbol{v}_a$, we identify the two valid neighbors, say $e$ and $f$, whose connecting unit vectors $\boldsymbol{d}_{[a,e]}$ and $\boldsymbol{d}_{[a,f]}$ most closely align with $\boldsymbol{v}_a$ in terms of angle, following the procedure outlined in Sectin 2.2 of the main text.

Let $v_a^n$ and $v_a^e$ denote the northward and eastward components of $\boldsymbol{v}_a$. Similarly, let $d_{[a,e]}^n$ and $d_{[a,e]}^e$ be the corresponding components of the unit vector $\boldsymbol{d}_{[a,e]}$. We define a weighted adjacency matrix $\tilde{M}$ containing the flow magnitudes along the directed edges. In fact, $\tilde{M}_{[a,e]}$ and $\tilde{M}_{[a,f]}$, are obtained by solving the exact linear system:
\begin{equation}
\begin{bmatrix}
    d_{[a,e]}^n & d_{[a,f]}^n \\
    d_{[a,e]}^e & d_{[a,f]}^e
\end{bmatrix}
\begin{bmatrix}
    \tilde{M}_{[a,e]} \\
    \tilde{M}_{[a,f]}
\end{bmatrix}
=
\begin{bmatrix}
    v_a^n \\
    v_a^e
\end{bmatrix}.
\end{equation}
This decomposition guarantees that $\tilde{M}_{[a,e]}\boldsymbol{d}_{[a,e]} + \tilde{M}_{[a,f]}\boldsymbol{d}_{[a,f]} = \boldsymbol{v}_a$. The operation is repeated for all nodes of the grid \(\mathcal{G}\). The magnitudes of the two components in $\tilde{M}$, scaled by their sum, populate the weighted adjacency matrix $M$. Indeed, $M$ collects the inputs to build the transition matrix $\pi$, as detailed in Section 2.3 of the main text.

Specific software execution steps, including raw variable extraction and data formatting pipelines, are provided in the repository's README.

\subsection{Computation of the Non-Return Probabilities}
\label{subsubsec:computeG}

The covariance derivations in the main text (Section 3) require the non-return probability $U(x,y)$, defined as the probability that the Markov chain, starting at node $x\in V$, never returns to the set $A = \{x,y\} \subset V$. 

Let $H^A = \inf \{ n \ge 0 : X_n \in A \}$ denote the first hitting time of the set $A$. By conditioning on the first step of the chain ($X_1 = x_1$), $U(x,y)$ is expressed as:
\begin{equation}
    U(x,y) = \sum_{x_1 \notin A} \pi_{[x,x_1]} \left( 1 - \mathbb{P}(H^A < \infty| X_0=x_1) \right).
\end{equation}

Thus, the problem reduces to computing the hitting probabilities $\mathbb{P}(H^A < \infty| X_0=x_1)$ for all neighbors \(x_1\) of \(x\) different from \(y\).

Directly simulating or enumerating paths to solve for hitting probabilities is computationally prohibitive. Instead, we leverage the matrix $\mathbf{G}$, whose generic entry $G_{[v_1,v_2]}$ is indexed by \(v_1,v_2 \in V\) and is equal to the expected number of visits to $v_2$ when  $X_0=v_1$:
\begin{equation}
    \mathbf{G} = \mathbb{E} \left[ \sum_{n=0}^{\infty} \mathbb{I}_{\{X_n=v_2\}} \right] = (\mathbf{I} - \pi_V)^{-1},
\end{equation}
where \(\pi_V\) is the sub matrix of the transition matrix \(\pi\) restricted to the vertices in \(V,\) i.e. excluding the absorbing state \(S\). Since all states in \(V\) are transient, the spectral radius of \(\pi_V\) is less than 1, and thus$(\mathbf{I} - \pi_V)$ is invertible.

Applying the Strong Markov Property to the hitting time $H^A$, the expected number of visits to a target $v \in A$ from any starting node $x_1$ can be decomposed by conditioning on the specific node $k \in A$ where the chain first hits the set: $$G_{[x_1,v]} = \sum_{k \in A} \mathbb{P}(X_{H^A} = k | X_0=x_1) \, G_{[k,v]}.$$ Indeed, since $A = \{x,y\}$, we can express this relationship in matrix form for the unknowns $\mathbb{P}(X_{H^A} = x|X_0 =x_1)$ and $\mathbb{P}(X_{H^A} = y|X_0=x_1)$:
\begin{equation}
    \begin{bmatrix}
        G_{[x_1,x]} & G_{[x_1,y]}
    \end{bmatrix} 
    = 
    \begin{bmatrix}
        \mathbb{P}(X_{H^A} = x|X_0=x_1) & \mathbb{P}(X_{H^A} = y|X_0=x_1)
    \end{bmatrix}
    \begin{bmatrix}
        G_{[x,x]} & G_{[x,y]}\\ 
        G_{[y,x]} & G_{[y,y]}
    \end{bmatrix}. \nonumber
\end{equation}

Since $\mathbb{P}(H^A<\infty|X_0=x_1) = \mathbb{P}(X_{H^A} = x|X_0=x_1) + \mathbb{P}(X_{H^A} = y|X_0=x_1)$, we obtain:
\begin{equation}
    \mathbb{P}(H^A < \infty|X_0=x_1) = 
    \begin{bmatrix}
        G_{[x_1,x]} & G_{[x_1,y]}
    \end{bmatrix}
    \begin{bmatrix}
        G_{[x,x]} & G_{[x,y]}\\ 
        G_{[y,x]} & G_{[y,y]}
    \end{bmatrix}^{-1}
    \begin{bmatrix}
        1 \\
        1
    \end{bmatrix}.
\end{equation}

To evaluate $U(x,y)$ efficiently across the domain, we implement the following procedure using sparse linear algebra:
\begin{enumerate}
    \item Extract the $2 \times 2$ submatrix $\mathbf{G}_{AA}=\begin{bmatrix}
        G_{[x,x]} & G_{[x,y]}\\ 
        G_{[y,x]} & G_{[y,y]}
    \end{bmatrix} $ and compute the coefficient vector $$\boldsymbol{b} = \mathbf{G}_{AA}^{-1} \mathbf{1}$$
    \item For any downstream neighbor $x_1 \notin A$ connected to $x$, compute $h_{x_1}^A = [G_{[x_1, x]} \;G_{[x_1, y]}] \boldsymbol{b}$.
    \item Evaluate the final probability using the transition probabilities from $x$: $$U(x,y) = \sum_{x_1 \notin A} \pi_{[x,x_1]} (1 - h_{x_1}^A).$$
\end{enumerate}
Note that, once $\mathbf{G}$ has been evaluated, the previous procedure is totally parallelizable, making the algorithm efficient and scalable.

We also note that, for all \(x \in V,\) the probability of not returning to $U(x)$ is simply the reciprocal of \(G_{[x,x]}.\)

\subsection{Covariance Matrix Evaluation via the Exponential Kernel}
\label{subsec:exponential}

A significant computational advantage arises when adopting the Exponential covariance function, $C(h) = \theta_s \exp(-h/\theta_r)$. The exponential function possesses the semigroup property:
\begin{equation}
    \exp\left(-\frac{h_1 + h_2}{\theta_r}\right) = \exp\left(-\frac{h_1}{\theta_r}\right) \cdot \exp\left(-\frac{h_2}{\theta_r}\right). \nonumber
\end{equation}
This property implies that the covariance contribution of any path factors into the product of the contributions of its individual edges. Consequently, the summation over all possible paths (including cyclic ones) can be computed exactly via matrix inversion, utilizing the Von Neumann series expansion. This completely bypasses the need for explicit path enumeration or heuristic pruning.

Let $\mathbf{P}$ be the matrix encoding the transition probabilities between vertices in \(V,\) weighted by the normalization coefficients; that is the generic element of $\mathbf{P}$  is $$p{[v_1,v_2]} = \pi_{[v_1,v_2]} / \sqrt{\sum_{k \in V} \pi_{[k,v_2]}},$$ for \(v_1,v_2 \in V\). Let $\mathbf{D}$ be the matrix of Euclidean distances between connected vertices. We construct a decayed transition matrix $\mathbf{R}$ via the Hadamard (element-wise) product:
\begin{equation}
    R_{[v_1,v_2]} = p{[v_1,v_2]} \cdot \exp\left(-\frac{D_{[v_1,v_2]}}{\theta_r}\right). \nonumber
\end{equation}
The entry $R_{[v_1,v_2]}$ represents the transition weight from $v_1$ to $v_2$ discounted by the spatial correlation decay associated with that single step.

The total covariance accumulated along all paths composed of $k$ steps is given by $\mathbf{R}^k$. Summing over all possible path lengths $k \in \{0, 1, \dots, \infty\}$ yields the geometric series:
\begin{equation}
    \sum_{k=0}^{\infty} \mathbf{R}^k = (\mathbf{I} - \mathbf{R})^{-1}.
\end{equation}
Because all vertices \(v \in V\) are transient and the spatial decay strictly bounds the entries of $\mathbf{R}$ below the corresponding ones of the transition matrix \(\pi\), the spectral radius of $\mathbf{R}$ is strictly less than 1, guaranteeing the convergence of the series.

The raw inverse matrix accumulates the effect of self-loops (recirculation starting and ending at the origin node). To derive the correct covariance constrained to acyclic paths $\tilde{\mathcal{P}}(x,y) \bigcup \tilde{\mathcal{P}}(x,y)$ as defined in the main text, the exact final covariance matrix is computed through the following algebraic operations:

\begin{enumerate}
    \item \textbf{Neumann Inversion:} Compute the raw accumulation matrix $\mathbf{S} = (\mathbf{I} - \mathbf{R})^{-1}$. The diagonal entry $S_{[v,v]}$ represents the infinite sum of all cyclic paths starting and ending at $v \in V$.
    \item \textbf{Source Cycle Elimination:} Normalize the matrix to remove the inflation caused by self-loops at the origin. We divide each row by its corresponding diagonal entry to yield the normalized matrix $\mathbf{S}^*$, where $S^*_{[v_1,v_2]} = S_{[v_1,v_2]} / S_{[v_1,v_1]},$ for \(v_1,v_2 \in V.\)
    \item \textbf{Non-return probabilities weighting:} Apply the matrix of pairwise non-return probabilities $\mathbf{U}$ and scale by the marginal non-return probabilities. For \(x,y \in V,\) let $$U_{[x,x]} = U(x),\;U_{[x,y]} = U(y,x)$$ and define \(\mathbf{U}\) to be the matrix whose generic entry, indexed by the vertices in \(V,\) is \(U_{[x,y]}.\) Then set
    \begin{equation}
        \tilde{\mathbf{\Sigma}} = \mathbf{S}^* \mathbf{U}. \nonumber
    \end{equation}
    \item \textbf{Symmetrization and Scaling:} The final spatial covariance matrix is symmetrized and scaled by the marginal variance $\theta_s$, with the diagonal explicitly reset to $\theta_s$ to ensure exact variance matching:
    \begin{equation}
        \mathbf{\Sigma}_{final} = \theta_s \cdot (\tilde{\mathbf{\Sigma}} + \tilde{\mathbf{\Sigma}}^T), \quad \text{diag}(\mathbf{\Sigma}_{final}) = \theta_s. \nonumber
    \end{equation}
\end{enumerate}
The matrix \( \mathbf{\Sigma}_{final}\) collects the covariances \(Cov(Z_x,Z_y)\) of the process \(\{Z_x, x \in V\},\) when the assumptions of Propostion 4 are satisfied and the covariance function is Exponential.

\section{Simulation study}
\label{sec:simulation}

\subsection{Simulation setting}

The objective of this simulation study is to evaluate the differences between the proposed framework and the conventional Euclidean approach. Specifically, we examine the ability of the new framework to reconstruct the spatial domain after generating synthetic data from its covariance structure.

We use a real velocity field as the basis for our simulations, specifically the data collected in 2022. Once the velocity field and the spatial points at which water temperature was observed were obtained, we constructed a network (using the procedure outlined in the main text) and performed all simulations on this network. \Cref{img:lin_net_sim} illustrates the network used for the simulation study. It is coarser than the network employed for the case study illustrated in the main text. This choice stems from the computational complexity of running simulations on a highly detailed domain.

\begin{figure}[t]
    \centering
    \includegraphics[width=0.5\linewidth]{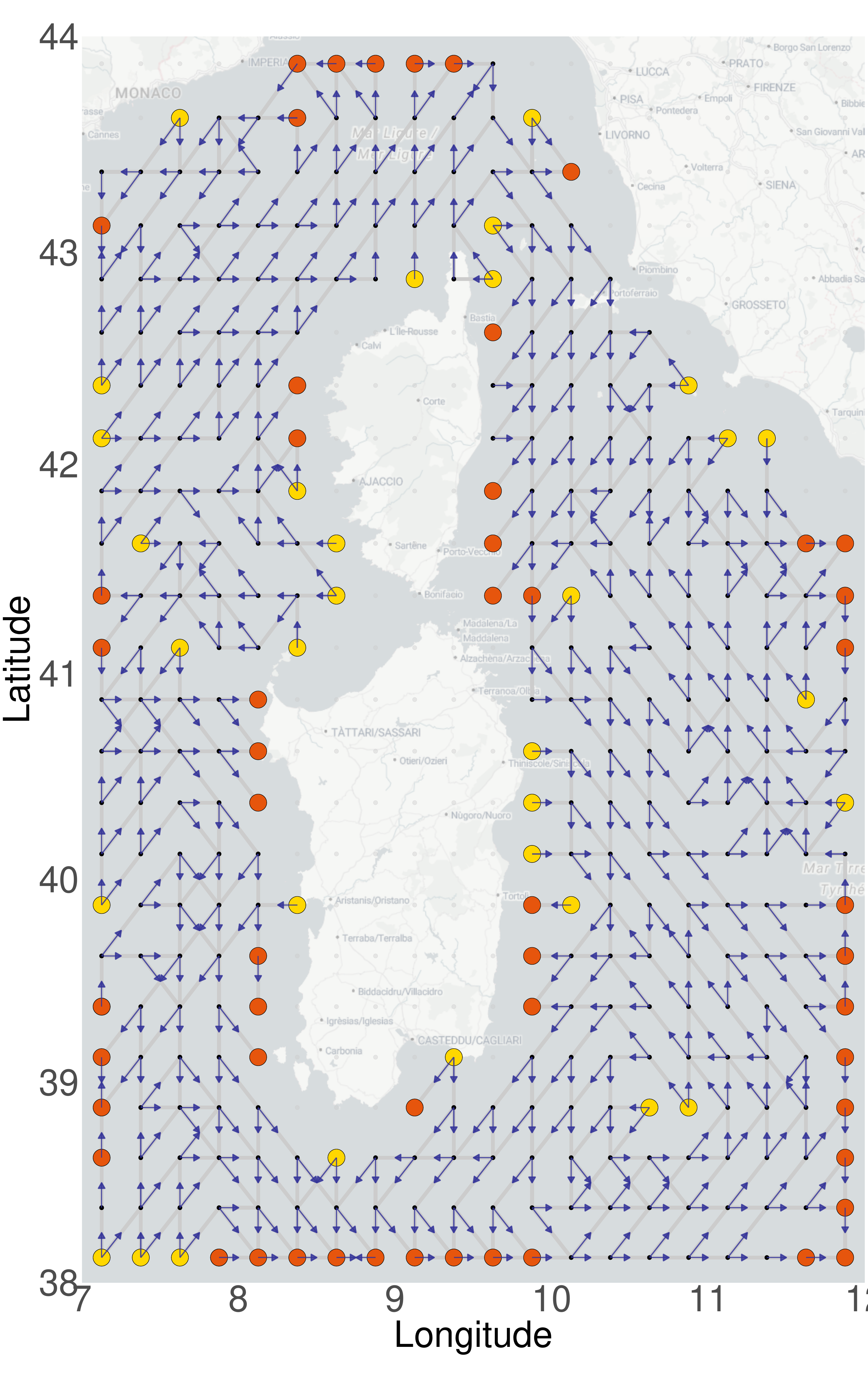}
    \caption{Linear network used in the simulation.}
    \label{img:lin_net_sim}
\end{figure}

For the simulation, $M = 100$ replicates of the random field $\boldsymbol{Y} \sim \mathcal{N}(\boldsymbol{0}, \Sigma)$ were generated, with $\Sigma$ following the covariance structure defined by the proposed framework. The covariance function is set to be Exponential: the sill parameter $\theta_s$ was fixed equal to 1 throughout the study to facilitate interpretability in terms of scale. For the range parameter $\theta_r$, we tested the newly defined process under five different range values (50, 87, 125, 162, 200 km in network distance) to assess how this parameter affects the behavior of the process.

Once the synthetic data were simulated, each iteration involved splitting the field into training and test sets, with the test set comprising one-fifth of the observations. This partitioning was randomized across all iterations to ensure robustness.

To quantify the predictive gain yielded by the network-based topology, we benchmarked the proposed method against a standard geostatistical alternative. For each replicate, we fitted a stationary, isotropic Gaussian process based on standard Euclidean distances. This model represents the baseline geostatistical approach. Consequently, it explicitly ignores the complex non-convexity of the domain (e.g., correlations crossing landmasses) and treats the space as homogeneous, disregarding the directional transport induced by the velocity field. 

\subsection{Covariance Parameters Estimation}

We first assess the estimators' ability to recover the parameters governing the random field. The estimation procedure follows the two-step algorithm detailed in Section 4 of the main text.

Figure \ref{img:sills} displays the kernel density estimates of the sill parameter ($\hat{\theta}_s$) obtained over the $M = 100$ replicates. Both the proposed physics-informed framework and the classical Euclidean approach yield unbiased distributions centered around the true value ($\theta_s = 1$). This result confirms that the total process variability is correctly identified by both methods, as it is primarily driven by the marginal variance of the data rather than the spatial dependence structure.

\begin{figure}[htbp]
    \centering
    \includegraphics[width=0.5\linewidth]{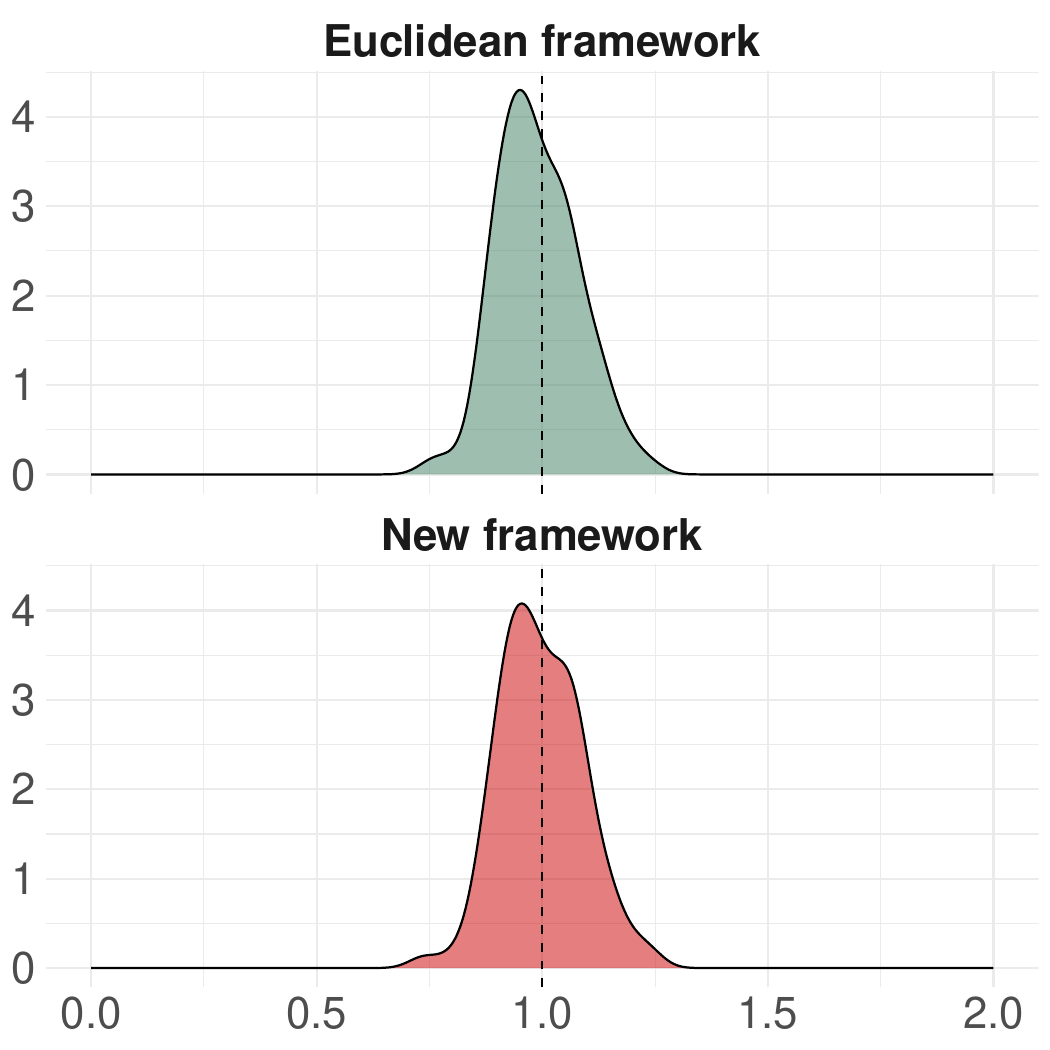}
    \caption{Empirical density of sill parameter estimates ($\hat{\theta}_s$) over 100 replicates. The vertical line indicates the true value.}
    \label{img:sills}
\end{figure}

A sharp contrast emerges when examining the spatial dependence structure. Figure \ref{img:variograms_eucl} presents the empirical covariance functions estimated under the Euclidean framework. Crucially, the Euclidean model fails to detect any significant spatial correlation, collapsing to a pure nugget effect at a significantly small range in the vast majority of replicates. This failure stems from a topological mismatch: pairs of points that are spatially proximal in Euclidean space (e.g., separated by a landmass) may be distant in the hydrographic network. The Euclidean metric "short-circuits" these connections, interpreting the high dissimilarity between these points as noise at short lags, thereby masking the true underlying signal.

\begin{figure}
    \centering
    \begin{minipage}[b]{0.45\linewidth}
        \centering
        \includegraphics[width=\linewidth]{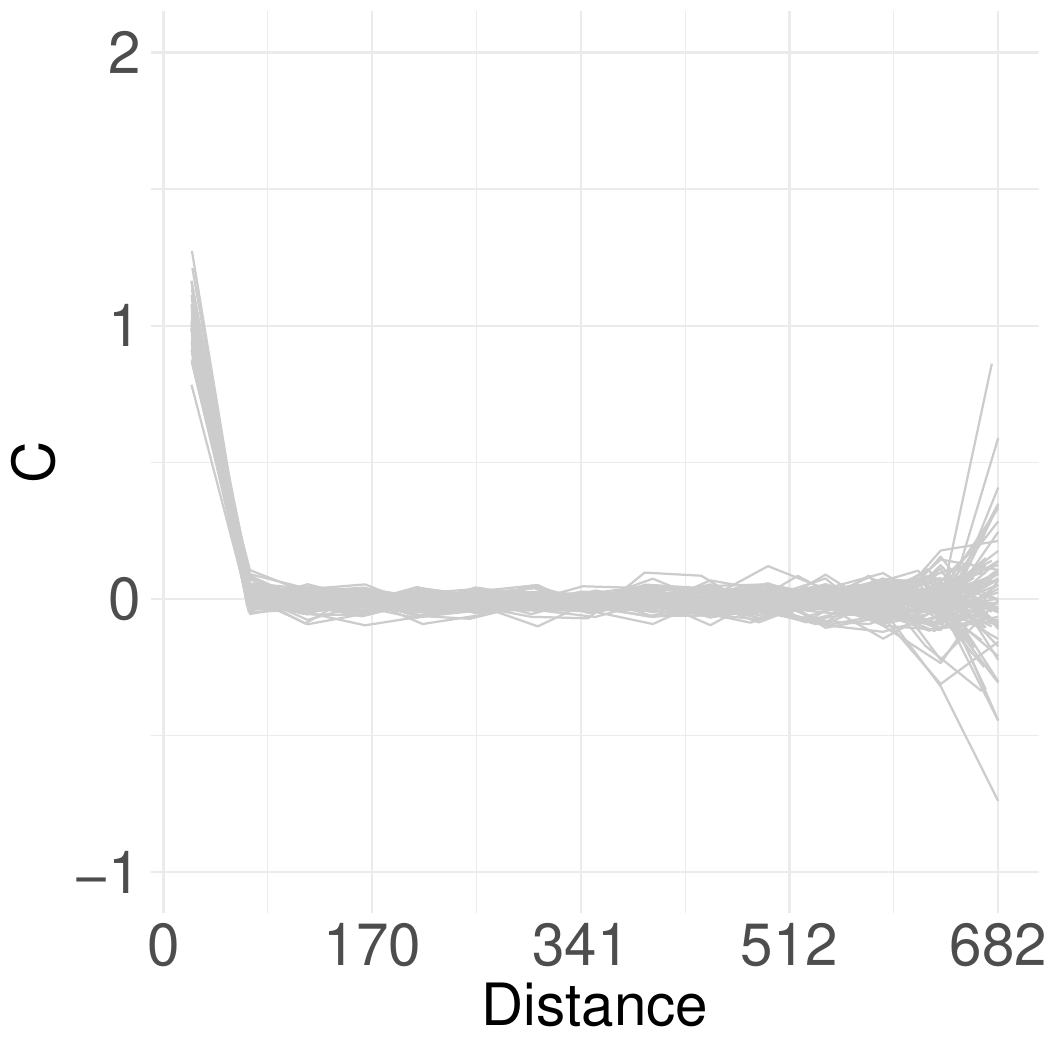}
        \vspace{0.2cm}
        \\ (a) True Range = 50 km
    \end{minipage}
    \hfill
    \begin{minipage}[b]{0.45\linewidth}
        \centering
        \includegraphics[width=\linewidth]{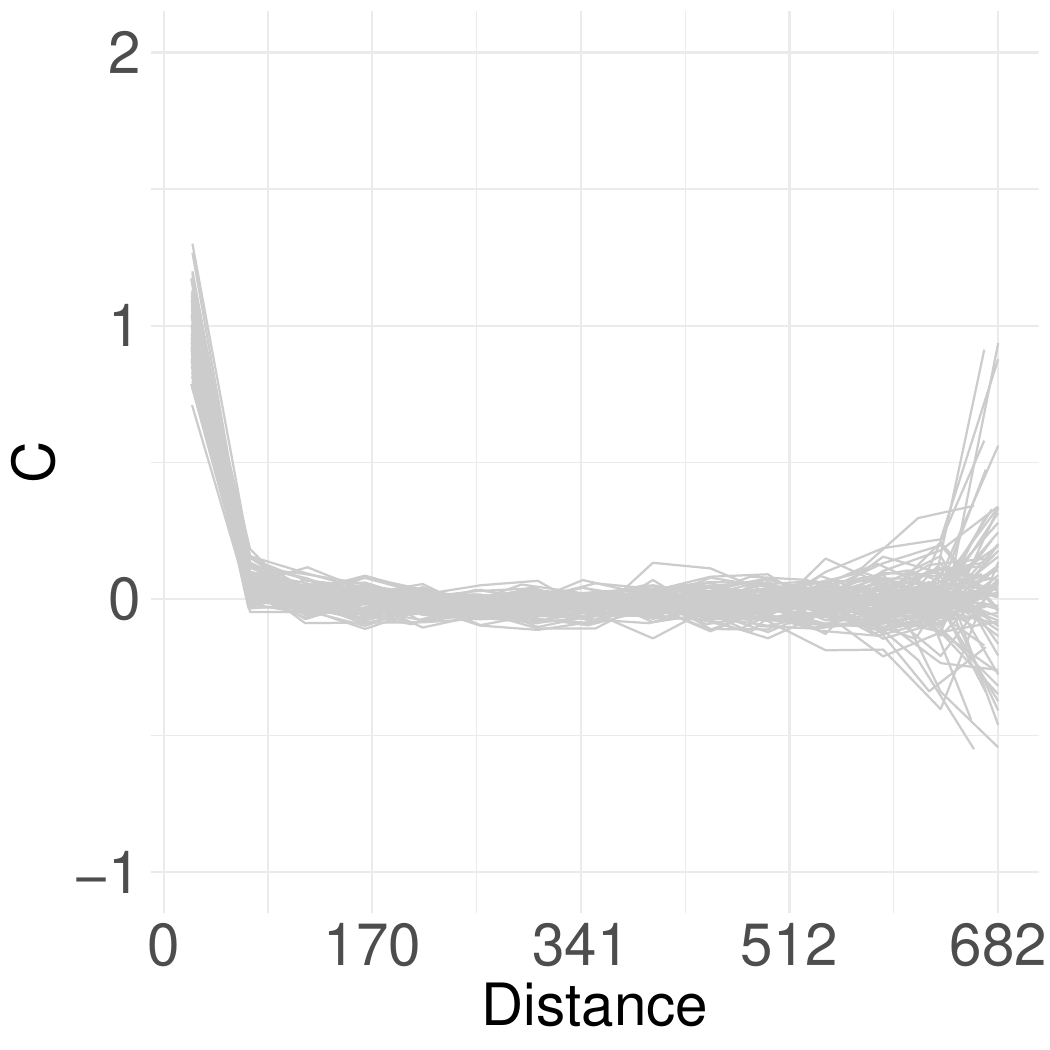}
        \vspace{0.2cm}
        \\ (b) True Range = 162 km
    \end{minipage}
    
    \caption{Empirical covariance functions estimates (gray lines) under the Euclidean framework. The lack of structure indicates a failure to capture the process dependence.}
    \label{img:variograms_eucl}
\end{figure}

Conversely, the proposed framework (Figure \ref{img:covariances}) successfully reconstructs the decaying profile of the covariance function, confirming its ability to capture the complex connectivity encoded in the velocity field. However, we observe a systematic negative bias in the estimation of the range parameter $\hat{\theta}_r$. This underestimation is a known phenomenon in network geostatistics \citep{torgegram}, often attributed to the confounding between network topology and metric distance. In our framework, this bias is further exacerbated by the regularization term $\lambda$ in the penalized least squares objective. Despite this bias in the magnitude of $\theta_r$, the framework correctly identifies the \textit{existence} and \textit{shape} of the spatial decay, which—as shown in the following section—is sufficient to outperform the Euclidean benchmark in predictive tasks.

\begin{figure}
    \centering
    \begin{minipage}[b]{0.45\linewidth}
        \centering
        \includegraphics[width=\linewidth]{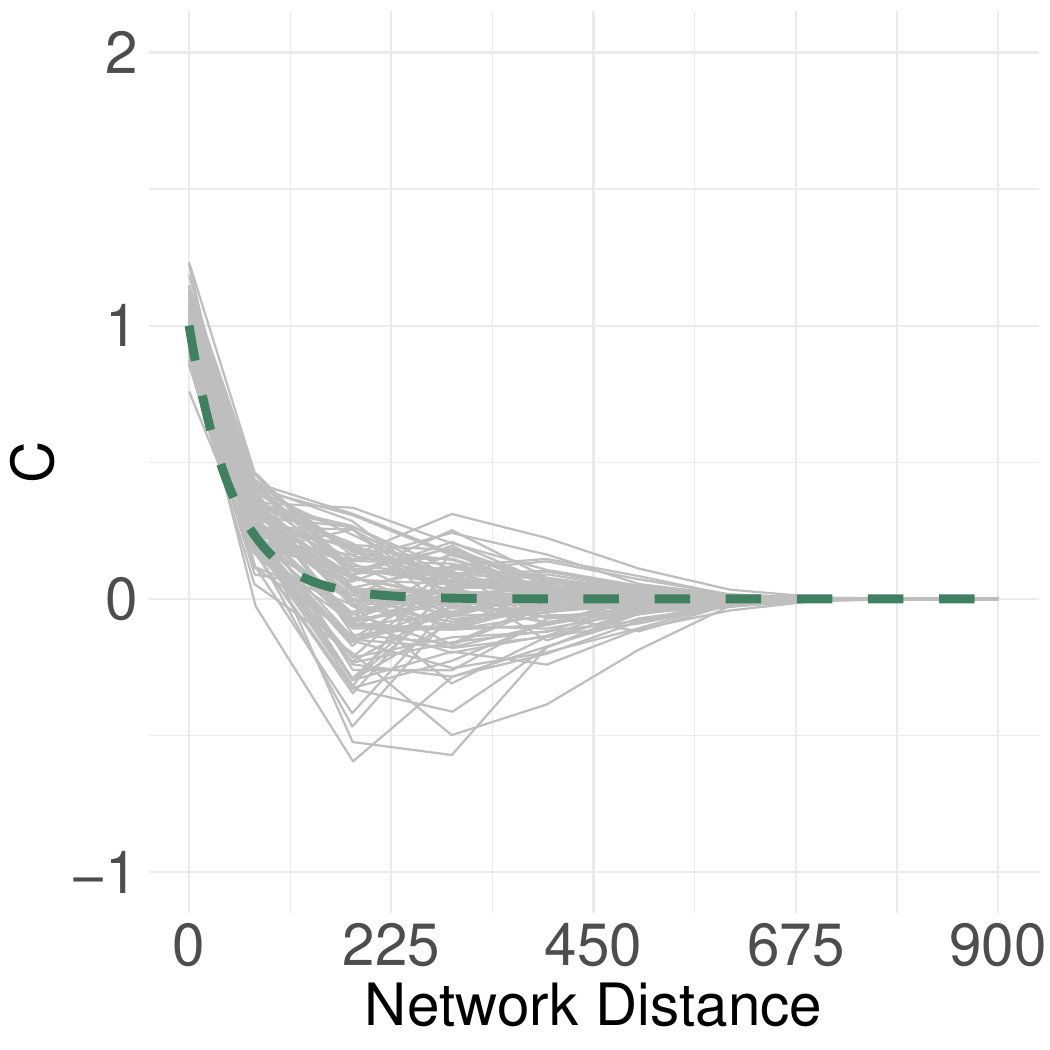}
        \vspace{0.2cm}
        \\ (a) True Range = 50 km
    \end{minipage}
    \hfill
    \begin{minipage}[b]{0.45\linewidth}
        \centering
        \includegraphics[width=\linewidth]{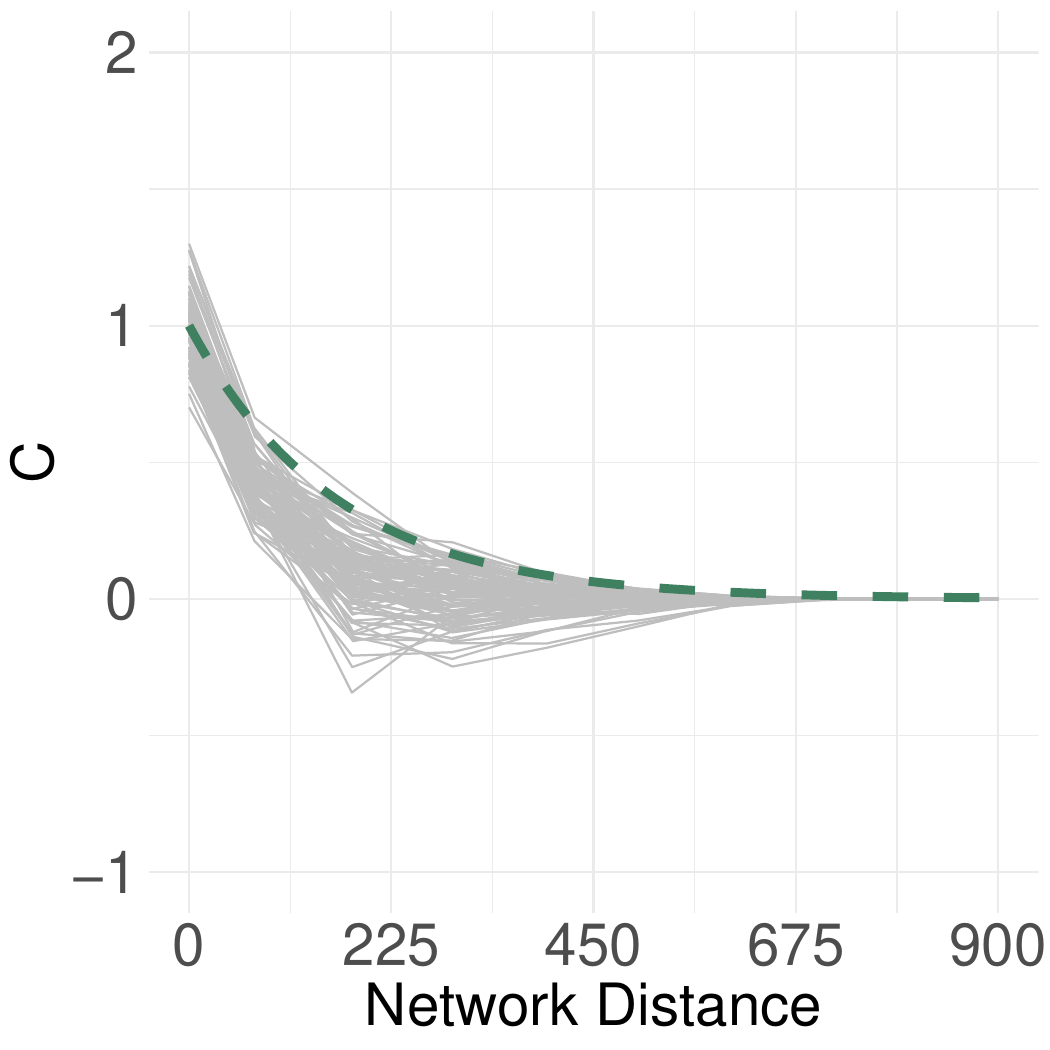}
        \vspace{0.2cm}
        \\ (b) True Range = 162 km
    \end{minipage}
    
    \caption{Empirical covariance function estimates under the proposed framework (gray lines) compared to the true generating function (green dashed line).}
    \label{img:covariances}
\end{figure}

To quantify the visual findings, we evaluated the Mean Squared Error (MSE) associated with the recovery of the exponential covariance structure. For the proposed framework, the MSE was computed as the average squared difference between the estimated covariance function $\hat{C}(h)$ and the true generating function $C(h)$.

Defining a comparable metric for the Euclidean framework is non-trivial, as the Euclidean distance $h_{eucl}$ does not translate directly into network distance. However, since the Euclidean model estimates a pure nugget effect, we compare the true covariance function against the following surrogate:
\begin{align*}
    \hat{C}_{eucl}(h) =
    \begin{cases}
        \hat{\theta}_{s,eucl} & \text{if } h = 0,\\
        0 & \text{if } h > 0.
    \end{cases}
\end{align*}

This function mimics the behavior of the Euclidean framework, under which no spatial correlation is detected. We then compute the MSE by comparing the true generating function $C(\cdot)$ against the two estimated covariance functions. \Cref{img:mse_estimation} reports the distribution of the MSE values across the replicates.

\begin{figure}
    \centering
    \includegraphics[width=0.6\linewidth]{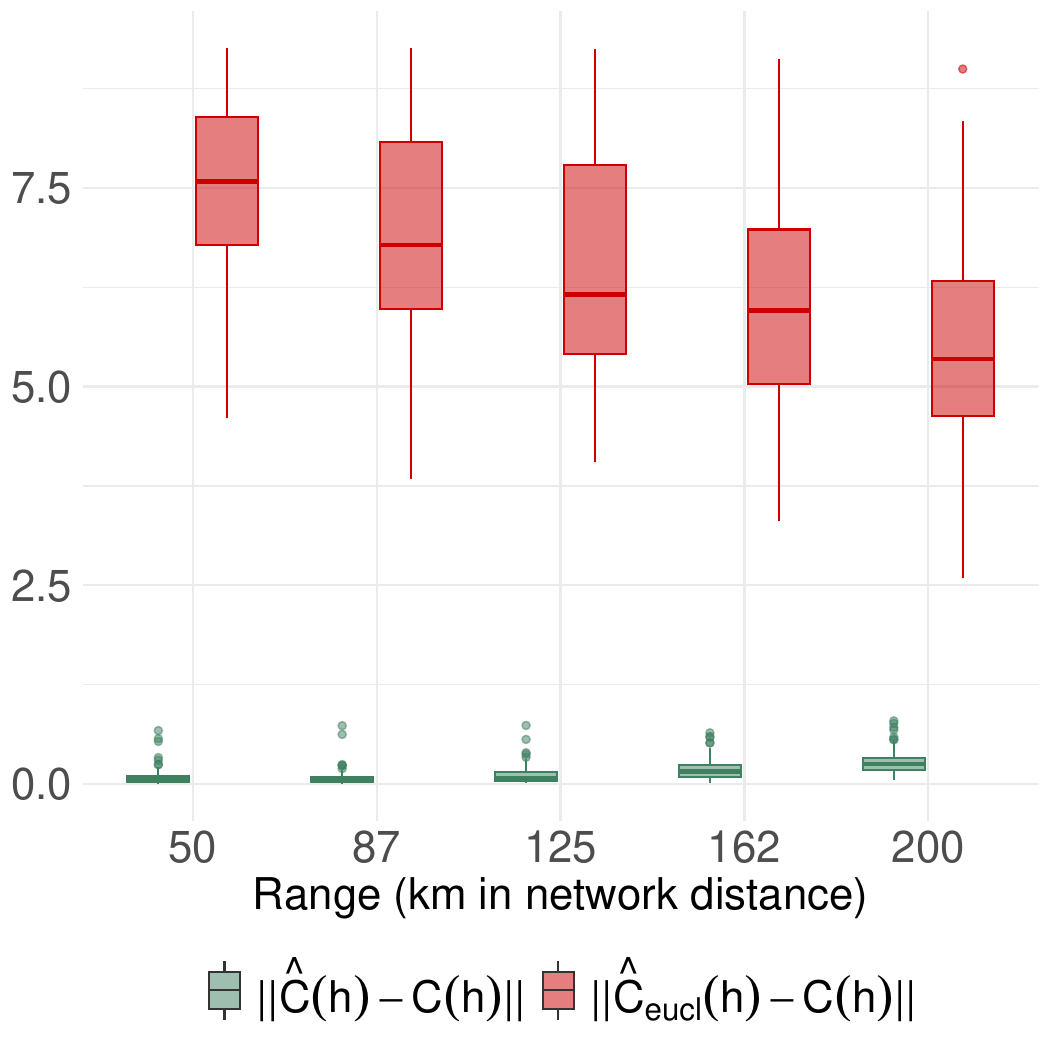}
    \caption{Distribution of Mean Squared Error (MSE) in covariance function estimation. The proposed framework, left/green, is compared with the euclidean benchmark, right/red.}
    \label{img:mse_estimation}
\end{figure}

To further dissect the source of the range underestimation observed in Figure \ref{img:covariances}, we performed a diagnostic experiment by fixing the sill parameter to its true value ($\hat{\theta}_s \equiv \theta_s$) during the optimization. This allows us to isolate the estimation of the range $\theta_r$ from potential identifiability issues between the sill and the range. \Cref{img:covariances_correctsill} presents the ensemble of estimated covariance functions under this constrained setting. Even with the sill correctly specified, a negative bias in the covariance magnitude persists at short-to-medium lags. This indicates that the underestimation is not merely a byproduct of sill-range correlation, but is largely attributable to the weighting in the covariance structure, as noted in \cite{torgegram, Barbi_Menafoglio_Secchi_2023}.

Nevertheless, comparing these curves with the Euclidean "flat line," it is evident that the proposed methodology—despite the regularization bias—effectively captures the functional form and the decay scale of the process, which is the primary requirement for accurate spatial interpolation.

\begin{figure}
    \centering
    \begin{minipage}[b]{0.45\linewidth}
        \centering
        \includegraphics[width=\linewidth]{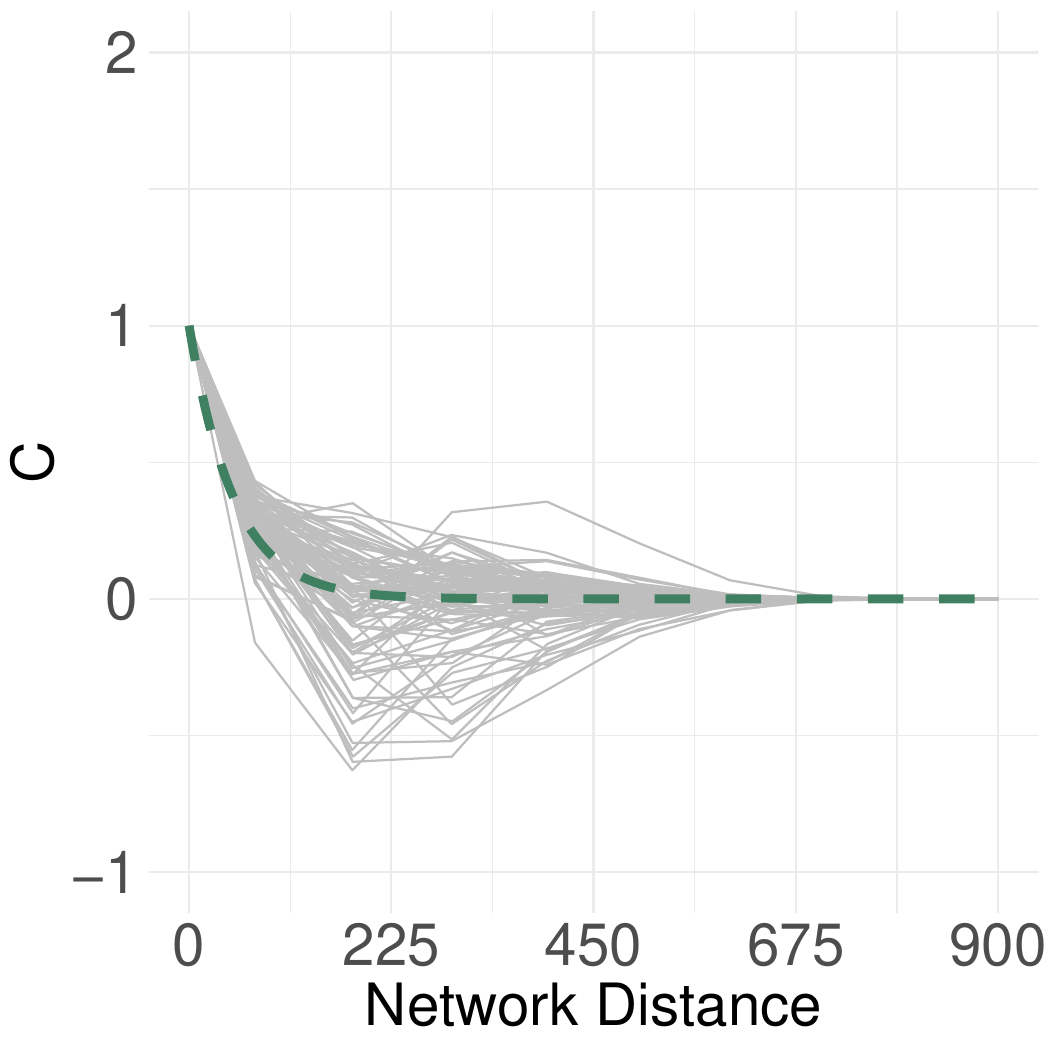}
        \vspace{0.2cm}
        \\ (a) True Range = 50 km
    \end{minipage}
    \hfill
    \begin{minipage}[b]{0.45\linewidth}
        \centering
        \includegraphics[width=\linewidth]{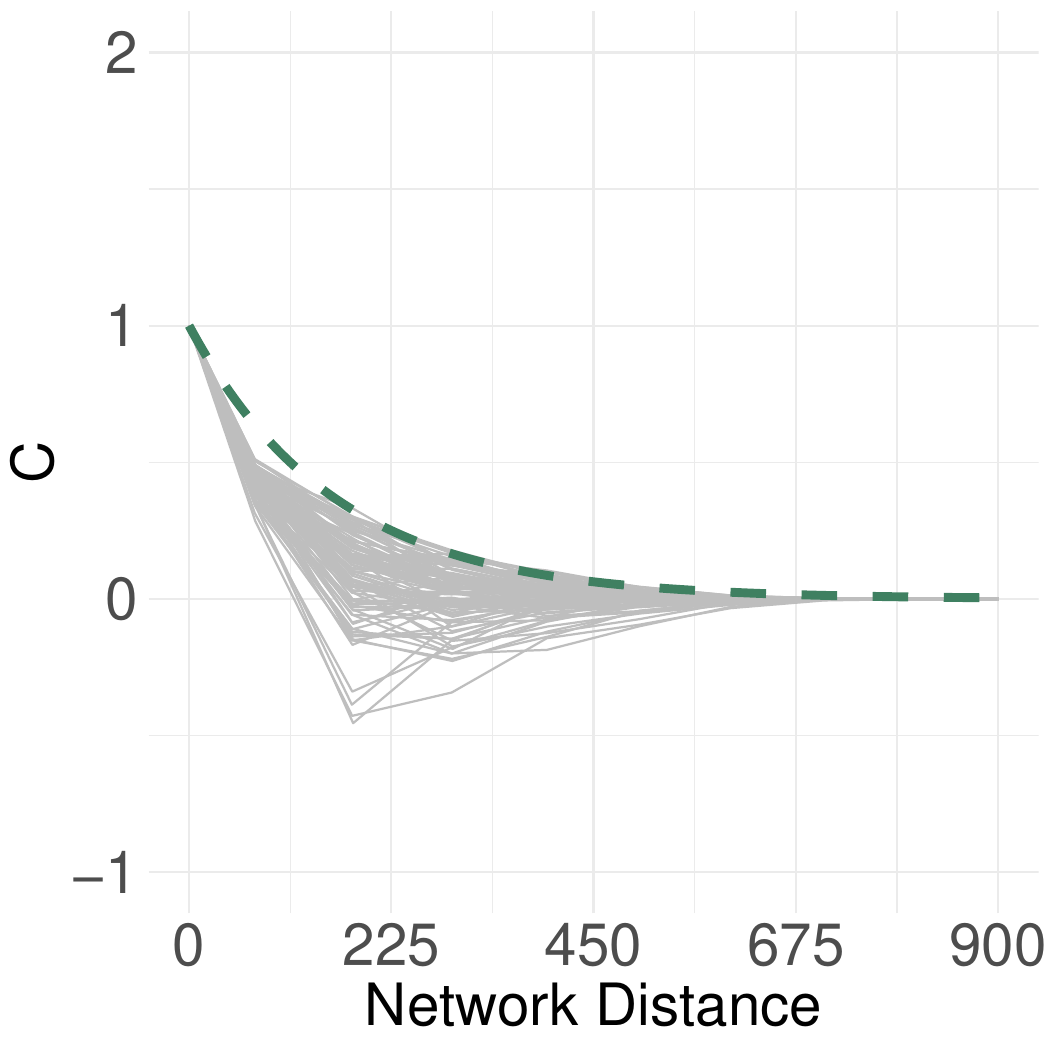}
        \vspace{0.2cm}
        \\ (b) True Range = 162 km
    \end{minipage}
    
    \caption{Diagnostic estimation with fixed sill ($\hat{\theta}_s = \theta_s$). The gray lines represent the empirical estimates, while the green dashed line is the truth. The persistence of the downward bias highlights the effect of the regularization penalty.}
    \label{img:covariances_correctsill}
\end{figure}

\subsection{Global Reconstruction of the Covariance Structure}

Beyond the estimation of individual parameters, valid spatial prediction (Kriging) requires that the entire covariance matrix $\Sigma$ is accurately reconstructed. To assess the global goodness-of-fit, we compared the covariance matrices implied by the estimated parameters against the true data-generating matrix. For the proposed framework, the matrix $\hat{\Sigma}$ was constructed using Equation (9) of the main text with the estimated $\hat{\theta}_s$ and $\hat{\theta}_r$. For the Euclidean benchmark, $\hat{\Sigma}_{eucl}$ was constructed using the standard isotropic exponential kernel.

We evaluated the reconstruction accuracy using two complementary metrics:
\begin{enumerate}
    \item \textbf{Frobenius Norm:} Measures the element-wise distance between matrices: $\|\Sigma - \hat{\Sigma}\|_F = \sqrt{\sum_{i,j} |\Sigma_{\left[i,j\right]} - \hat{\Sigma}_{\left[i,j\right]}|^2}$.
    \item \textbf{Kullback–Leibler (KL) Divergence:} Measures the information loss when approximating the true multivariate Gaussian distribution $\mathcal{N}(\boldsymbol{0}, \Sigma)$ with the estimated one $\mathcal{N}(\boldsymbol{0}, \hat{\Sigma})$.
\end{enumerate}

Figure \ref{img:Norms} displays the box-plots of these errors across the simulation replicates. The proposed framework yields reconstruction errors that are consistently lower and less dispersed than those of the Euclidean model.
Specifically, the Euclidean approach exhibits high KL divergence, indicating a significant distortion of the probabilistic structure of the field—a direct consequence of its inability to model the physical barriers and flow-directed correlations. Conversely, our framework maintains a low divergence across all range scenarios, confirming that the estimated covariance matrix preserves the essential information required for reliable inference.

\begin{figure}
    \centering
    \begin{minipage}[b]{0.45\linewidth}
        \centering
        \includegraphics[width=\linewidth]{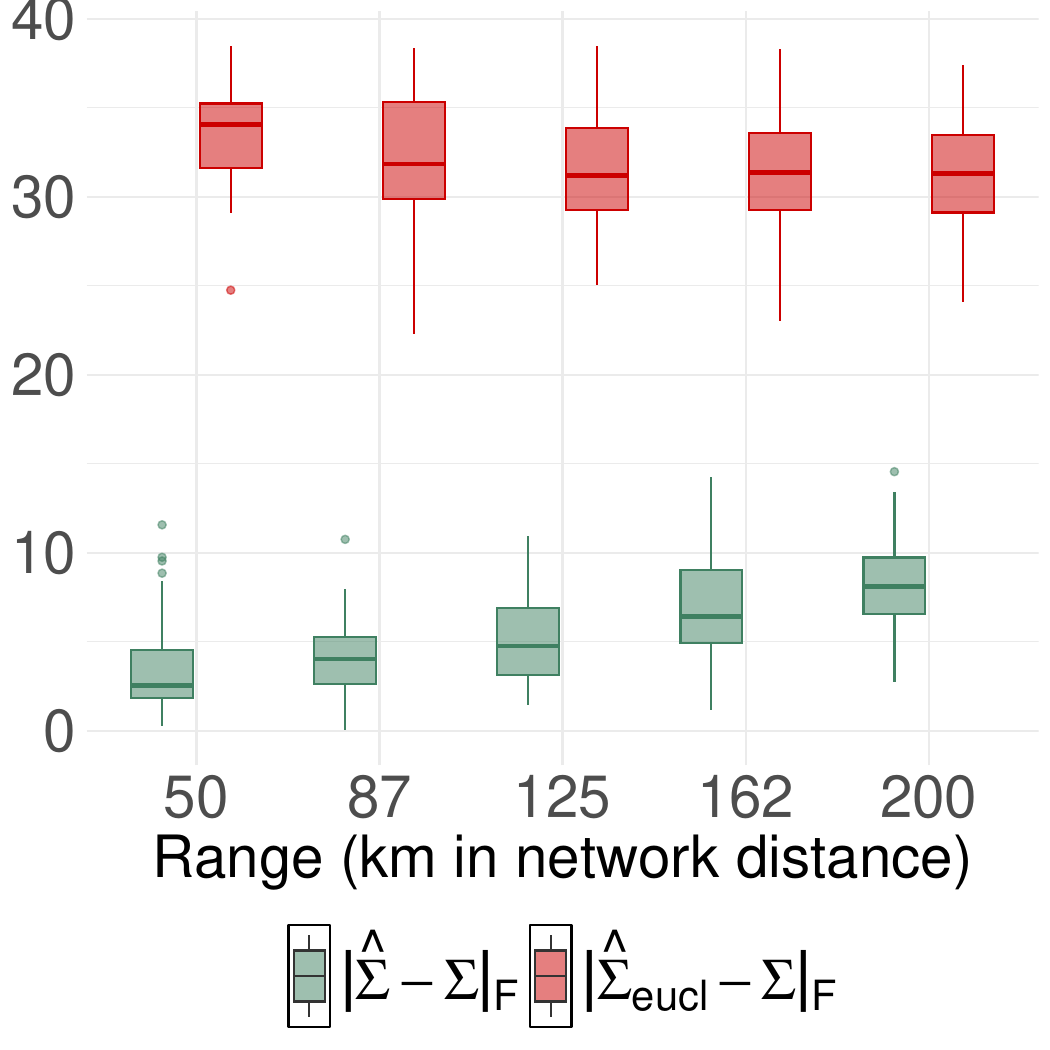}
        \vspace{0.2cm}
        \\ (a) Frobenius Norm
        \label{img:frobenius}
    \end{minipage}
    \hfill
    \begin{minipage}[b]{0.45\linewidth}
        \centering
        \includegraphics[width=\linewidth]{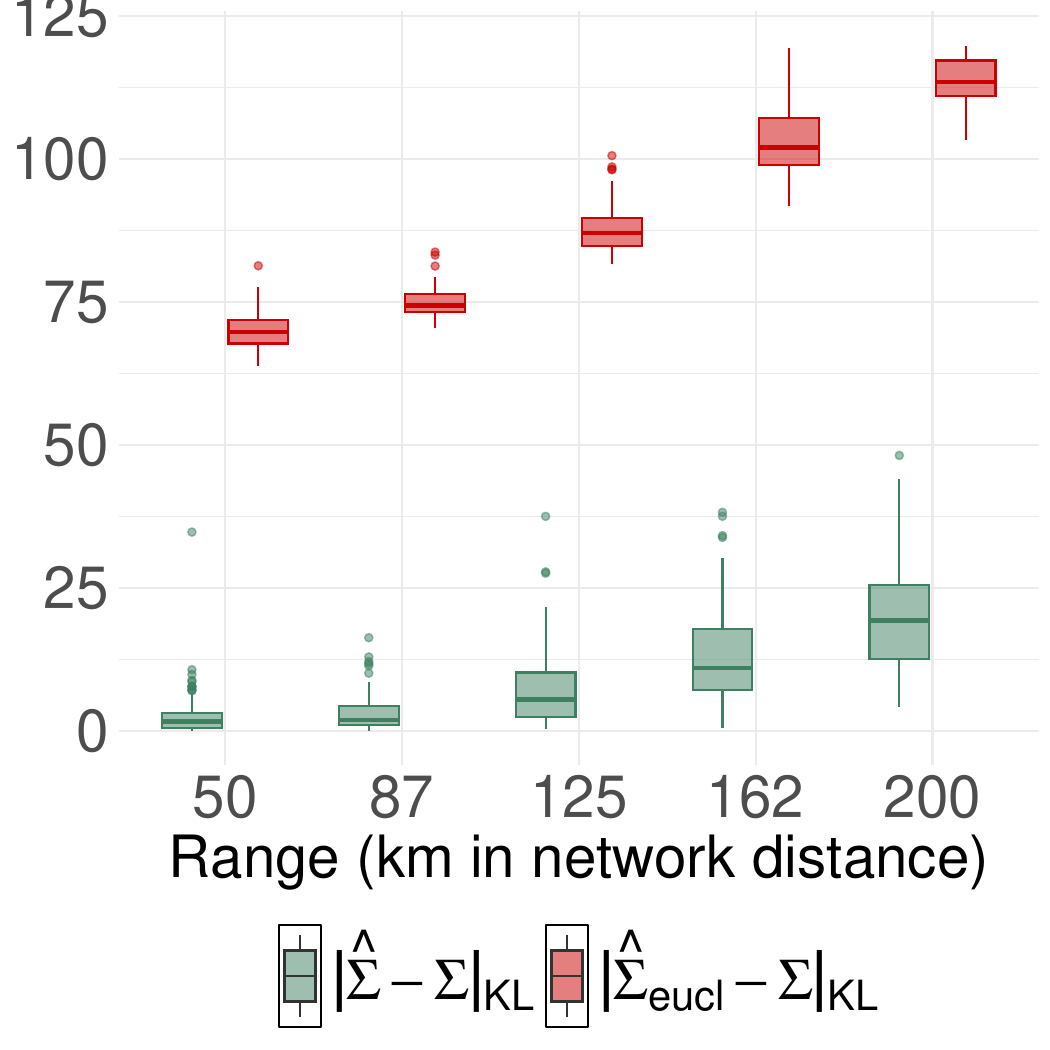}
        \vspace{0.2cm}
        \\ (b) Kullback–Leibler Divergence
        \label{img:kl_div}
    \end{minipage}
    
    \caption{Distribution of reconstruction errors between the true covariance matrix and the estimates. Left: Frobenius Norm (element-wise accuracy). Right: Kullback–Leibler Divergence (distributional accuracy).}
    \label{img:Norms}
\end{figure}

\subsection{Out-of-Sample Predictive Performance}

Finally, we assessed the ultimate goal of the geostatistical analysis: the ability to reconstruct unobserved locations via spatial interpolation.
For each simulation replicate, we performed Simple Kriging on the hold-out test set (20\% of sites), using the parameters and covariance structures estimated in the training phase. We quantified the accuracy using the Mean Squared Error (MSE).

Figure \ref{img:MSE} displays the distribution of MSE values across the $M = 100$ replicates. The results demonstrate the significance of the difference between the two approaches.
The Euclidean model consistently yields an MSE approximately equal to the process variance (fixed at $\theta_s = 1$). This confirms that, having estimated a pure nugget effect (as shown in Figure \ref{img:variograms_eucl}), the Euclidean Kriging collapses to the trivial predictor (i.e., predicting the global mean everywhere), failing to exploit any spatial information from neighboring data points.

In contrast, the proposed physics-informed framework achieves substantially lower prediction errors across all scenarios. By correctly modeling the connectivity induced by the currents, the network-based Kriging effectively leverages information from upstream and downstream neighbors, significantly reducing the predictive uncertainty.
Furthermore, the relative advantage of the proposed method scales with the spatial range: as the true correlation length increases, the potential for information transfer across the domain grows. Our framework successfully captures this long-range dependence, leading to a widening performance gap relative to the Euclidean baseline, which remains blind to the spatial structure regardless of the theoretical range.

\begin{figure}[t]
    \centering
    \includegraphics[width=0.6\linewidth]{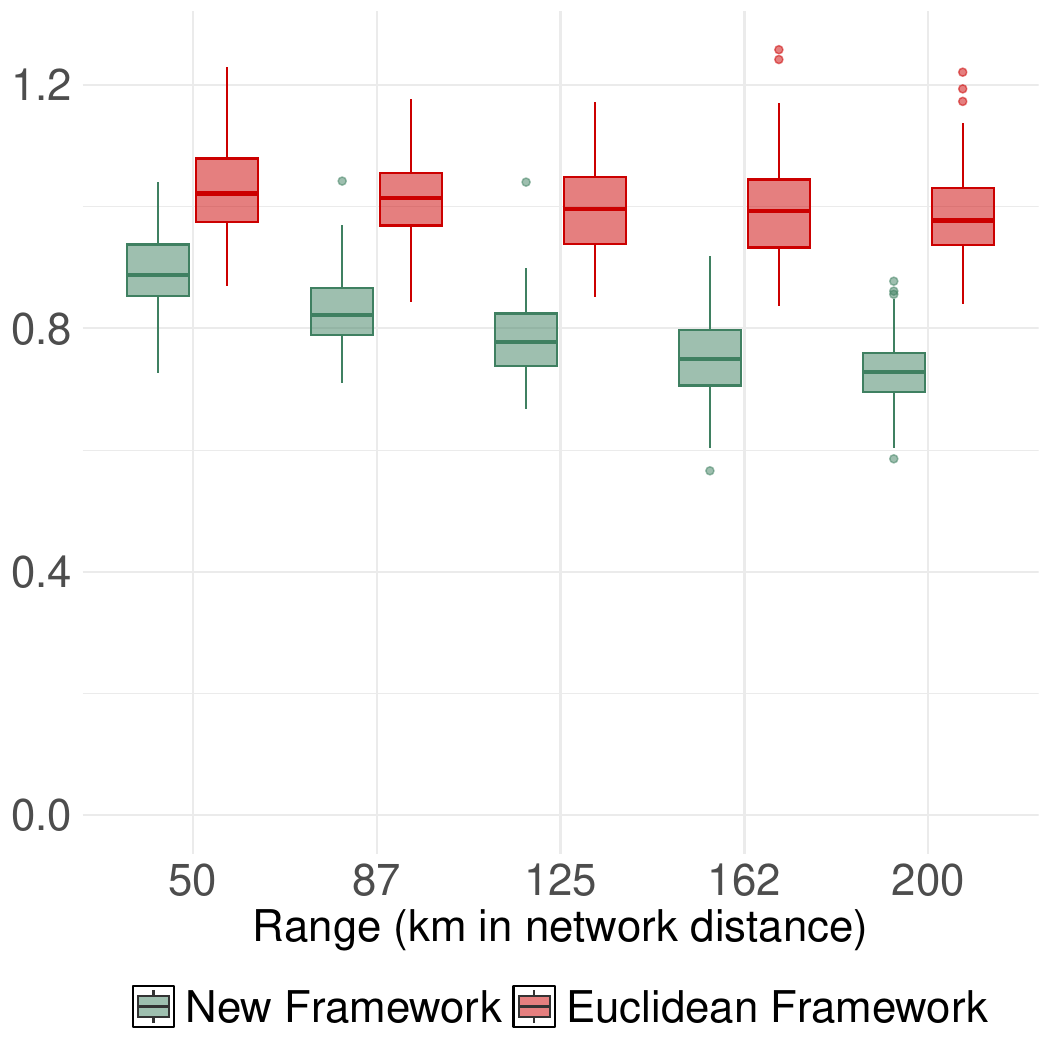}
    \caption{Distribution of Mean Squared Prediction Error (MSE) on the hold-out test set. The Euclidean model (red/right boxplots) stays close to the process variance (1.0), indicating no predictive gain over the mean. The proposed framework (blue/left boxplots) significantly reduces the error.}
    \label{img:MSE}
\end{figure}

\section{A particular case: convolution process over stream networks}
\label{sec:processinstream}

This section demonstrates analytically that the proposed physics-informed framework encompasses the classical stream network models as a particular case. Specifically, we show that when the domain is restricted to a simple dendritic river network (a directed binary tree), our formulation is consistent with the \textit{tail-up} model introduced by \cite{VerHoef2006} and further developed by \cite{rivers2010}.

\subsection{The Tail-Up Stream Topology}
\label{subsec:streamprocess}

A stream network is topologically defined as a directed acyclic graph where flow direction is unambiguous. Let $(V,\cal L)$ be this network. Water flows downstream from multiple sources towards a single outlet without bifurcating. The fundamental distinction between this dendritic structure and the marine environment (defined in Section 3 of the main text) is the property of path uniqueness: for any pair of hydrologically connected points $x$ and $y$, there exists exactly one directed path $\mathcal{P}(x,y) = \{p(x,y)\}$ connecting them. In the formulation of \cite{VerHoef2006}, the spatial process relies on a weighting scheme to ensure stationarity. Let $\nu_{[a,b]} \in (0,1]$ be the weighting corresponding to the edge $l_{[a,b]}$ of the network. The stationarity condition requires that the weights of all upstream segments merging at $b$ sum to unity:

\begin{equation}
\sum_{a \in V} \nu_{[a,b]} = 1.
\label{eq:stationaritystream}
\end{equation}

The resulting "tail-up" spatial covariance between two connected points is given by the product of the square roots of these stream weights along the unique path, multiplied by a valid 1D covariance function $C_0(h)$: for two points $x,y$, connected by the path $p(x,y)$,

\[
Cov_{stream}(x, y) = \left( \prod_{l \in p(x,y)} \sqrt{\nu_l} \right) C_0(|p(x,y)|).
\]

\subsection{Mapping the Proposed Framework to the Stream Case}

In our random walk interpretation, dynamic is governed by the transition probability matrix $\pi$. In a stream network where flow does not split (no bifurcations downstream), each vertex $a$ has at most one outgoing edge to a vertex $b$. Consequently, the transition probability becomes deterministic:
\begin{equation}
    \pi_{\left[a,b\right]} = 
    \begin{cases}
        1 & \text{if } a \to b \text{ is the unique downstream edge,} \\
        0 & \text{otherwise.}
    \end{cases}
    \label{eq:piinstream}
\end{equation}
This implies that the random variable $T$ introduced in the main text to handle directional uncertainty at junctions becomes degenerate (i.e., the path is deterministic). Notably, we have that $U(x) = 1$ for all $x \in V$. Hence we omit the quantity in the covariance formulation.

Let $\nu_{[a,b]} = 1/(\sum_{k \in V} \pi_{[k,b]})$, and note that the stationarity condition is satisfied. The normalization constants $\beta$'s defined in Equation (5) of the main text can be thus expressed in terms of $\nu$:

\[
\beta_{p(v,x)} = \prod_{l_{[a,b]} \in p(v,x)} \sum_k \pi_{[k,b]} = \prod_{l_{[a,b]} \in p(v,x)} \frac{1}{\nu_{[a,b]}}.
\]

Given the specification of the transition probability and of the normalization constants, the weight in the covariance model for a specific path $p(x,y)$ is
\begin{equation}
    w_{p(x,y)} = \frac{\mathbb{P}(T(v,x) = p(v,x))}{\beta_{p(v,x)}} = \prod_{l \in p(x,y)} \sqrt{\nu_l}.
\end{equation}

Since $\mathcal{P}(x,y)$ contains at most one path $p(x,y)$, the sum for the covariance model reduces to a single term. Thus, for connected points $x$ and $y$, the covariance becomes:
\begin{equation}
    Cov(Z_x,Z_y) = \left(\prod_{l \in p(x,y)} \sqrt{\nu_l}\right) C(|p(x,y)|),
    \label{eq:CovarianceVerHoefasmine}
\end{equation}
which is identical to the tail-up covariance structure of \cite{VerHoef2006}.

This derivation confirms that the proposed framework is a consistent generalization of the stream network methodology. The key advancement lies in the ability to handle scenarios where $\pi_{\left[a,b\right]} < 1$ (flow splitting) and $|\mathcal{P}(x,y)| > 1$ (multiple paths/cycles), features that are essential for marine domains but absent in river systems.

\subsection{The Tail-Down Stream Topology}

The stream models of \citet{VerHoef2006, rivers2010} provide a double interpretation and a consequent different model when looking at the flow direction in the opposite sense, yielding the so called "Tail-Down" model. In this section, we demonstrate that our proposed framework yields a consistent covariance structure regardless of the chosen flow direction.

While the physical topology of the domain remains unchanged, the representation of the network dynamics is reversed. The property of path uniqueness between connected locations still holds. However, because flow can diverge when moving in the upstream direction (i.e., network splits), the transition probability matrix now contains entries strictly less than 1.

Let $\pi$ be the transition probability matrix for the Markov chain with state space $V$, and let the entries of $\pi$ be defined as $\pi_{[a,b]} = \nu_{[b,a]}$. Note that this is a probability matrix, given the condition of stationarity enforced in \Cref{eq:stationaritystream}.

Consider an upstream vertex of a given stream segment. In a strictly branching stream network, this vertex can be reached from exactly one downstream node. Hence, $\sum_k \pi_{[k,b]} = \pi_{[a,b]}$ where $a \in V$ is the only downstream vertex of $b$ connected with it. The full specification of the $\beta$s is

\[
\beta_{p(x,y)} = \prod_{l_{[a,b]} \in p(x,y)} \sum_k \pi_{[k,b]} = \prod_{l_{[a,b]} \in p(x,y)} \pi_{[a,b]}
\]

Summing all up, the weight in the covariance model is

\[
w_{p(x,y)} = \prod_{l_{[a,b]} \in p(v,x)} \frac{\pi_{[a,b]}}{\sqrt{\pi_{[a,b]}}} = \prod_{l_{[a,b]} \in p(v,x)} \sqrt{\pi_{[a,b]}}
\]

Substituting $\pi_{[a,b]} = \nu_{[b,a]} = \nu_l$, the full covariance mode is expressed as

\begin{equation}
    Cov(Z_x, Z_y) = \left(\prod_{l \in p(v,x)} \sqrt{\nu_l} \right) C(|p(x,y)|).
\end{equation}

This confirms that the covariance model remains structurally identical under both interpretations of flow direction, unifying the tail-up and tail-down paradigms within our generalized spatial framework.

%% file: bibliography.bib
@misc{Cressie_1993, title={Statistics for Spatial Data}, journal = {Wiley Series in Probability and Statistics}, DOI={10.1002/9781119115151}, author={Cressie, Noel}, year={1993}, publisher={Wiley}, location={New York}, isbn={9780471002550}}

@article{Curriero_2006, title={On the Use of Non-Euclidean Distance Measures in Geostatistics}, volume={38}, ISSN={1573-8868}, DOI={10.1007/s11004-006-9055-7}, number={8}, journal={Mathematical Geology}, author={Curriero, Frank C.}, year={2006}, pages={907–926}, language={en} }

@article{hydroProximityMeasure, title={Extraction of hydrological proximity measures from DEMs using parallel processing}, volume={26}, ISSN={1364-8152}, DOI={10.1016/j.envsoft.2011.07.018}, author={Tesfa, Teklu K. and Tarboton, David G. and Watson, Daniel W. and Schreuders, Kimberly A. T. and Baker, Matthew E. and Wallace, Robert M.}, journal = {Environmental Modelling \& Software}, year={2011}, pages={1696–1709} }

@article{VerHoef2006, title={Spatial statistical models that use flow and stream distance}, url={https://digitalcommons.unl.edu/usdeptcommercepub/185}, journal={United States Department of Commerce: Staff Publications}, DOI = {10.1007/s10651-006-0022-8}, author={Ver Hoef, Jay M. and Peterson, Erin and Theobald, David}, year={2006} }

@misc{copernicus, title = {Multi Observation Global Ocean 3D Temperature Salinity Height Geostrophic Current and MLD}, author = {E.U. Copernicus Marine Service Information, (CMEMS)}, year = {2020}, howpublished = {Marine Data Store (MDS)}, doi = {https://doi.org/10.48670/moi-00052}}

@article{Dinf, title={A new method for the determination of flow directions and upslope areas in grid digital elevation models}, volume={33}, rights={Copyright 1997 by the American Geophysical Union.}, ISSN={1944-7973}, DOI={10.1029/96WR03137}, number={2}, journal={Water Resources Research}, author={Tarboton, David G.}, year={1997}, pages={309–319}, language={en} }

@article{cressie2006, title={Spatial prediction on a river network}, volume={11}, ISSN={1537-2693}, DOI={10.1198/108571106X110649}, number={2}, journal={Journal of Agricultural, Biological, and Environmental Statistics}, author={Cressie, Noel and Frey, Jesse and Harch, Bronwyn and Smith, Mick}, year={2006}, pages={127}, language={en} }

@article{torgegram, title={The Torgegram for Fluvial Variography: Characterizing Spatial Dependence on Stream Networks}, volume={26}, DOI={10.1080/10618600.2016.1247006}, journal={Journal of Computational and Graphical Statistics}, author={Zimmerman, Dale and Ver Hoef, Jay}, year={2016} }

@article{peterson2007, title={Geostatistical modelling on stream networks: Developing valid covariance matrices based on hydrologic distance and stream flow}, volume={52}, DOI={10.1111/j.1365-2427.2006.01686.x}, journal={Freshwater Biology}, author={Peterson, Erin and Theobald, David and Ver Hoef, Jay}, year={2007}, pages={267–279} }

@article{rivers20102, title={A mixed-model moving-average approach to geostatistical modeling in stream networks}, volume={91}, ISSN={0012-9658}, number={3}, journal={Ecology}, publisher={Ecological Society of America}, author={Peterson, Erin E. and Ver Hoef, Jay M.}, year={2010}, pages={644–651} }

@article{clemente2026, title={Nonparametric estimators over metric graphs}, ISSN={1464-3510}, DOI={10.1093/biomet/asag029}, journal={Biometrika}, author={Clemente, Aldo and Arnone, Eleonora and Mateu, Jorge and Sangalli, Laura M}, year={2026}, month={apr}, pages={asag029} }

@article{Tomasetto_Arnone_Sangalli_2024, title={Modeling Anisotropy and Non-Stationarity Through Physics-Informed Spatial Regression}, volume={35}, rights={© 2024 The Author(s). Environmetrics published by John Wiley & Sons Ltd.}, ISSN={1099-095X}, DOI={10.1002/env.2889}, number={8}, journal={Environmetrics}, author={Tomasetto, Matteo and Arnone, Eleonora and Sangalli, Laura M.}, year={2024}, pages={e2889}, language={en} }

@article{rivers2010, title={A Moving Average Approach for Spatial Statistical Models of Stream Networks}, volume={105}, DOI={10.1198/jasa.2009.ap08248}, journal={Journal of the American Statistical Association}, author={Ver Hoef, Jay and Peterson, Erin}, year={2010}, pages={6–18} }

@article{heatwaves3, title={Marine heat waves in the Mediterranean Sea: An assessment from the surface to the subsurface to meet national needs}, volume={10}, ISSN={2296-7745}, url={https://www.frontiersin.org/journals/marine-science/articles/10.3389/fmars.2023.1045138/full}, DOI={10.3389/fmars.2023.1045138}, journal={Frontiers in Marine Science}, publisher={Frontiers}, author={Dayan, Hugo and McAdam, Ronan and Juza, Mélanie and Masina, Simona and Speich, Sabrina}, year={2023}, language={English} }

@article{hotspot, title={Climate change hot-spots}, volume={33}, rights={Copyright 2006 by the American Geophysical Union.}, ISSN={1944-8007}, url={https://onlinelibrary.wiley.com/doi/abs/10.1029/2006GL025734}, DOI={10.1029/2006GL025734}, number={8}, journal={Geophysical Research Letters}, author={Giorgi, F.}, year={2006}, language={en} }

@article{heatwaves2, title={Biological Impacts of Marine Heatwaves}, volume={15}, ISSN={1941-1405, 1941-0611}, DOI={10.1146/annurev-marine-032122-121437}, journal={Annual Review of Marine Science}, publisher={Annual Reviews}, author={Smith, Kathryn E. and Burrows, Michael T. and Hobday, Alistair J. and King, Nathan G. and Moore, Pippa J. and Gupta, Alex Sen and Thomsen, Mads S. and Wernberg, Thomas and Smale, Dan A.}, year={2023}, pages={119–145}, language={en} }

@article{ERSEM, title={ERSEM 15.06: a generic model for marine biogeochemistry and the ecosystem dynamics of the lower trophic levels}, volume={9}, ISSN={1991-959X}, DOI={10.5194/gmd-9-1293-2016}, number={4}, journal={Geoscientific Model Development}, publisher={Copernicus GmbH}, author={Butenschön, Momme and Clark, James and Aldridge, John N. and Allen, Julian Icarus and Artioli, Yuri and Blackford, Jeremy and Bruggeman, Jorn and Cazenave, Pierre and Ciavatta, Stefano and Kay, Susan and Lessin, Gennadi and van Leeuwen, Sonja and van der Molen, Johan and de Mora, Lee and Polimene, Luca and Sailley, Sevrine and Stephens, Nicholas and Torres, Ricardo}, year={2016}, pages={1293–1339}, language={English} }

@article{Barbi_Menafoglio_Secchi_2023, title={An object-oriented approach to the analysis of spatial complex data over stream-network domains}, volume={58}, ISSN={2211-6753}, DOI={10.1016/j.spasta.2023.100784}, journal={Spatial Statistics}, author={Barbi, Chiara and Menafoglio, Alessandra and Secchi, Piercesare}, year={2023}, pages={100784} }

@misc{copernicusprojections, title = {Marine biogeochemistry data for the Northwest European Shelf and Mediterranean Sea from 2006 up to 2100 derived from climate projections}, author = {Copernicus Climate Change Service, (C3S)}, year = {2020}, howpublished = {Copernicus Climate Change Service (C3S) Climate Data Store (CDS)}, doi = {10.24381/cds.dcc9295c}}

@article{spdeLindgren, title={Spatial Models Generated by Nested Stochastic Partial Differential Equations, with an Application to Global Ozone Mapping}, volume={5}, ISSN={1932-6157}, number={1}, journal={The Annals of Applied Statistics}, publisher={Institute of Mathematical Statistics}, author={Bolin, David and Lindgren, Finn}, year={2011}, pages={523–550} }

@article{spdeLindgren-2, title={An explicit link between Gaussian fields and Gaussian Markov random fields: the stochastic partial differential equation approach}, volume={73}, rights={© 2011 Royal Statistical Society}, ISSN={1467-9868}, DOI={10.1111/j.1467-9868.2011.00777.x}, number={4}, journal={Journal of the Royal Statistical Society: Series B (Statistical Methodology)}, author={Lindgren, Finn and Rue, Håvard and Lindström, Johan}, year={2011}, pages={423–498}, language={en} }

@article{clarotto, title={The SPDE approach for spatio-temporal datasets with advection and diffusion}, volume={62}, ISSN={2211-6753}, DOI={10.1016/j.spasta.2024.100847}, journal={Spatial Statistics}, author={Clarotto, Lucia and Allard, Denis and Romary, Thomas and Desassis, Nicolas}, year={2024}, month={aug}, pages={100847} }

@article{ramsay2002, title={Spline Smoothing over Difficult Regions}, volume={64}, ISSN={1369-7412}, DOI={10.1111/1467-9868.00339}, number={2}, journal={Journal of the Royal Statistical Society Series B: Statistical Methodology}, author={Ramsay, Tim}, year={2002}, pages={307–319} }

@article{objectoriented, title={Statistical analysis of complex and spatially dependent data: A review of Object Oriented Spatial Statistics}, volume={258}, ISSN={0377-2217}, DOI={10.1016/j.ejor.2016.09.061}, number={2}, journal={European Journal of Operational Research}, author={Menafoglio, Alessandra and Secchi, Piercesare}, year={2017}, pages={401–410} }

@Manual{R, title = {R: A Language and Environment for Statistical Computing}, author = {{R Core Team}}, organization = {R Foundation for Statistical Computing}, address = {Vienna, Austria}, year = {2023}, url = {https://www.R-project.org/}}

@article{extremeeventsHuser, title={Estimating high-resolution Red Sea surface temperature hotspots, using a low-rank semiparametric spatial model}, volume={15}, ISSN={1932-6157, 1941-7330}, DOI={10.1214/20-AOAS1418}, number={2}, journal={The Annals of Applied Statistics}, publisher={Institute of Mathematical Statistics}, author={Hazra, Arnab and Huser, Raphaël}, year={2021}, month={jun}, pages={572–596} }

@article{ExtremeeventsBolin, title={Excursion and Contour Uncertainty Regions for Latent Gaussian Models}, volume={77}, ISSN={1369-7412}, DOI={10.1111/rssb.12055}, number={1}, journal={Journal of the Royal Statistical Society Series B: Statistical Methodology}, author={Bolin, David and Lindgren, Finn}, year={2015}, month={jan}, pages={85–106} }

@article{ExtremeeventsFrench, title={Spatio-temporal exceedance locations and confidence regions}, volume={7}, ISSN={1932-6157, 1941-7330}, DOI={10.1214/13-AOAS631}, number={3}, journal={The Annals of Applied Statistics}, publisher={Institute of Mathematical Statistics}, author={French, Joshua P. and Sain, Stephan R.}, year={2013}, month={sep}, pages={1421–1449} }

@article{CMIP5, title={An Overview of CMIP5 and the Experiment Design}, url={https://journals.ametsoc.org/view/journals/bams/93/4/bams-d-11-00094.1.xml}, DOI={10.1175/BAMS-D-11-00094.1}, journal = {Bulletin of the American Meteorological Society}, author={Taylor, Karl E. and Stouffer, Ronald J. and Meehl, Gerald A.}, year={2012}, month={apr}, language={en} }

@article{SSTGaussian, title={Linear Gaussian state-space model with irregular sampling: application to sea surface temperature}, volume={25}, ISSN={1436-3259}, DOI={10.1007/s00477-010-0442-8}, number={6}, journal={Stochastic Environmental Research and Risk Assessment}, author={Tandeo, Pierre and Ailliot, Pierre and Autret, Emmanuelle}, year={2011}, month={aug}, pages={793–804}, language={en} }

@article{anderes2020, title={Isotropic covariance functions on graphs and their edges}, volume={48}, ISSN={0090-5364, 2168-8966}, DOI={10.1214/19-AOS1896}, number={4}, journal={The Annals of Statistics}, publisher={Institute of Mathematical Statistics}, author={Anderes, Ethan and Møller, Jesper and Rasmussen, Jakob G.}, year={2020}, month={aug}, pages={2478–2503} }

@article{bolin2024, title={Gaussian Whittle–Matérn fields on metric graphs}, volume={30}, ISSN={1350-7265}, DOI={10.3150/23-BEJ1647}, number={2}, journal={Bernoulli}, publisher={Bernoulli Society for Mathematical Statistics and Probability}, author={Bolin, David and Simas, Alexandre B. and Wallin, Jonas}, year={2024}, month={may}, pages={1611–1639} }

@article{funcboxplot, title={Functional Boxplots}, volume={20}, ISSN={1061-8600}, DOI={10.1198/jcgs.2011.09224}, number={2}, journal={Journal of Computational and Graphical Statistics}, publisher={Taylor & Francis}, author={Sun, Ying and Genton, Marc G.}, year={2011}, month={jan}, pages={316–334} }

@article{rcp, title={The representative concentration pathways: an overview}, volume={109}, ISSN={1573-1480}, DOI={10.1007/s10584-011-0148-z}, number={1}, journal={Climatic Change}, author={van Vuuren, Detlef P. and Edmonds, Jae and Kainuma, Mikiko and Riahi, Keywan and Thomson, Allison and Hibbard, Kathy and Hurtt, George C. and Kram, Tom and Krey, Volker and Lamarque, Jean-Francois and Masui, Toshihiko and Meinshausen, Malte and Nakicenovic, Nebojsa and Smith, Steven J. and Rose, Steven K.}, year={2011}, month={aug}, pages={5}, language={en} }

@article{watertemperature, title={Climate change and freshwater biodiversity: detected patterns, future trends and adaptations in northern regions}, volume={84}, rights={© 2008 The Authors Journal compilation © 2008 Cambridge Philosophical Society}, ISSN={1469-185X}, DOI={10.1111/j.1469-185X.2008.00060.x}, number={1}, journal={Biological Reviews}, author={Heino, Jani and Virkkala, Raimo and Toivonen, Heikki}, year={2009}, pages={39–54}, language={en} }

@article{Azzimonti_Sangalli_Secchi_Domanin_Nobile_2015, title={Blood Flow Velocity Field Estimation Via Spatial Regression With PDE Penalization}, volume={110}, ISSN={0162-1459}, DOI={10.1080/01621459.2014.946036}, number={511}, journal={Journal of the American Statistical Association}, publisher={Taylor & Francis}, author={Azzimonti, Laura and Sangalli, Laura M. and Secchi, Piercesare and Domanin, Maurizio and Nobile, Fabio}, year={2015}, pages={1057–1071} }
